\documentclass[aps,prl,notitlepage,twocolumn,nofootinbib,superscriptaddress]{revtex4-2}

\usepackage[dvipsnames]{xcolor}
\usepackage{framed}
\definecolor{shadecolor}{rgb}{0.9,0.9,0.9}

\usepackage{mathtools}
\usepackage{amsmath}
\usepackage[shortlabels]{enumitem}

\usepackage{graphicx,epic,eepic,epsfig,amsmath,latexsym,amssymb,verbatim,color}
 
\usepackage{amsfonts}       
\usepackage{nicefrac}       

\usepackage{amsmath}
\usepackage{bbm}

\usepackage{float}
\usepackage{tikz}
\usetikzlibrary{chains}
\usetikzlibrary{fit}
\usetikzlibrary{arrows} 
\usetikzlibrary{decorations} 

\usepackage{epsfig}
\usetikzlibrary{shapes.symbols,patterns} 
\usepackage{pgfplots}

\usepackage[strict]{changepage}
\usepackage{hyperref}
\hypersetup{colorlinks=true,citecolor=blue,linkcolor=blue,filecolor=blue,urlcolor=blue,breaklinks=true}

\usepackage[marginal]{footmisc}
\usepackage{url}
\usepackage{theorem}

\newtheorem{definition}{Definition}
\newtheorem{proposition}{Proposition}
\newtheorem{lemma}[proposition]{Lemma}

\newtheorem{theorem}[proposition]{Theorem}

\newtheorem{corollary}[proposition]{Corollary}

\def\squareforqed{\hbox{\rlap{$\sqcap$}$\sqcup$}}
\def\qed{\ifmmode\squareforqed\else{\unskip\nobreak\hfil
\penalty50\hskip1em\null\nobreak\hfil\squareforqed
\parfillskip=0pt\finalhyphendemerits=0\endgraf}\fi}
\def\endenv{\ifmmode\;\else{\unskip\nobreak\hfil
\penalty50\hskip1em\null\nobreak\hfil\;
\parfillskip=0pt\finalhyphendemerits=0\endgraf}\fi}
\newenvironment{proof}{\noindent \textbf{{Proof~} }}{\hfill $\blacksquare$}

\newcounter{remark}
\newenvironment{remark}[1][]{\refstepcounter{remark}\par\medskip\noindent%
\textbf{Remark~\theremark #1} }{\medskip}

\newcounter{example}

\mathchardef\ordinarycolon\mathcode`\:
\mathcode`\:=\string"8000
\def\vcentcolon{\mathrel{\mathop\ordinarycolon}}
\begingroup \catcode`\:=\active
  \lowercase{\endgroup
  \let :\vcentcolon
  }

\usepackage{cleveref}
\usepackage{graphicx}
\usepackage{xcolor}

\RequirePackage[framemethod=default]{mdframed}
\newmdenv[skipabove=7pt,
skipbelow=7pt,
backgroundcolor=darkblue!15,
innerleftmargin=5pt,
innerrightmargin=5pt,
innertopmargin=5pt,
leftmargin=0cm,
rightmargin=0cm,
innerbottommargin=5pt,
linewidth=1pt]{tBox}

\newmdenv[skipabove=7pt,
skipbelow=7pt,
backgroundcolor=red!15,
innerleftmargin=5pt,
innerrightmargin=5pt,
innertopmargin=5pt,
leftmargin=0cm,
rightmargin=0cm,
innerbottommargin=5pt,
linewidth=1pt]{rBox}

\newmdenv[skipabove=7pt,
skipbelow=7pt,
backgroundcolor=blue2!25,
innerleftmargin=5pt,
innerrightmargin=5pt,
innertopmargin=5pt,
leftmargin=0cm,
rightmargin=0cm,
innerbottommargin=5pt,
linewidth=1pt]{dBox}
\newmdenv[skipabove=7pt,
skipbelow=7pt,
backgroundcolor=darkkblue!15,
innerleftmargin=5pt,
innerrightmargin=5pt,
innertopmargin=5pt,
leftmargin=0cm,
rightmargin=0cm,
innerbottommargin=5pt,
linewidth=1pt]{sBox}
\definecolor{darkblue}{RGB}{0,76,156}
\definecolor{darkkblue}{RGB}{0,0,153}
\definecolor{blue2}{RGB}{102,178,255}
\definecolor{darkred}{RGB}{195,0,0}

\newcommand{\nc}{\newcommand}
\nc{\rnc}{\renewcommand}
\nc{\lbar}[1]{\overline{#1}}
\nc{\bra}[1]{\langle#1|}
\nc{\ket}[1]{|#1\rangle}
\nc{\ketbra}[2]{|#1\rangle\!\langle#2|}
\nc{\braket}[2]{\langle#1|#2\rangle}

\nc{\dketbra}[2]{|#1\rangle\!\rangle\!\langle\!\langle#2|}
\nc{\dket}[1]{|#1\rangle\!\rangle}
\nc{\dbra}[1]{\langle\!\langle#1|}

\nc{\proj}[1]{| #1\rangle\!\langle #1 |}
\nc{\avg}[1]{\langle#1\rangle}
\nc{\smfrac}[2]{\mbox{$\frac{#1}{#2}$}}
\nc{\tr}{\operatorname{Tr}}
\nc{\ox}{\otimes}
\nc{\dg}{\dagger}
\nc{\dn}{\downarrow}
\nc{\cA}{{\cal A}}
\nc{\cB}{{\cal B}}
\nc{\cC}{{\cal C}}
\nc{\cD}{{\cal D}}
\nc{\cE}{{\cal E}}
\nc{\cF}{{\cal F}}
\nc{\cG}{{\cal G}}
\nc{\cH}{{\cal H}}
\nc{\cI}{{\cal I}}
\nc{\cJ}{{\cal J}}
\nc{\cK}{{\cal K}}
\nc{\cL}{{\cal L}}
\nc{\cM}{{\cal M}}
\nc{\cN}{{\cal N}}
\nc{\cO}{{\cal O}}
\nc{\cP}{{\cal P}}
\nc{\cQ}{{\cal Q}}
\nc{\cR}{{\cal R}}
\nc{\cS}{{\cal S}}
\nc{\cT}{{\cal T}}
\nc{\cU}{{\cal U}}
\nc{\cV}{{\cal V}}
\nc{\cX}{{\cal X}}
\nc{\cY}{{\cal Y}}
\nc{\cZ}{{\cal Z}}
\nc{\cW}{{\cal W}}
\nc{\csupp}{{\operatorname{csupp}}}
\nc{\qsupp}{{\operatorname{qsupp}}}
\nc{\var}{{\operatorname{var}}}
\nc{\rar}{\rightarrow}
\nc{\lrar}{\longrightarrow}
\nc{\polylog}{{\operatorname{polylog}}}
\nc{\wt}{{\operatorname{wt}}}
\nc{\av}[1]{{\left\langle {#1} \right\rangle}}
\nc{\supp}{{\operatorname{supp}}}

\nc{\argmin}{{\operatorname{argmin}}}

\def\a{\alpha}
\def\b{\beta}

\def\x{\xi}

\nc{\RR}{{{\mathbb R}}}
\nc{\CC}{{{\mathbb C}}}
\nc{\FF}{{{\mathbb F}}}
\nc{\NN}{{{\mathbb N}}}
\nc{\ZZ}{{{\mathbb Z}}}
\nc{\PP}{{{\mathbb P}}}
\nc{\QQ}{{{\mathbb Q}}}
\nc{\UU}{{{\mathbb U}}}
\nc{\EE}{{{\mathbb E}}}
\nc{\id}{{\operatorname{id}}}

\nc{\CHSH}{{\operatorname{CHSH}}}

\newcommand{\Op}{\operatorname}
\nc{\be}{\begin{equation}}
\nc{\ee}{{\end{equation}}}
\nc{\bea}{\begin{eqnarray}}
\nc{\eea}{\end{eqnarray}}
\nc{\<}{\langle}
\rnc{\>}{\rangle}
\nc{\rU}{\mbox{U}}

\nc{\ob}[1]{#1}

\nc{\SEP}{{\text{\rm SEP}}}
\nc{\NS}{{\text{\rm NS}}}
\nc{\LOCC}{{\text{\rm LOCC}}}
\nc{\PPT}{{\text{\rm PPT}}}
\nc{\EXT}{{\text{\rm EXT}}}
\nc{\Sym}{{\operatorname{Sym}}}

\nc{\ERLO}{{E_{\text{r,LO}}}}
\nc{\ERLOCC}{{E_{\text{r,LOCC}}}}
\nc{\ERPPT}{{E_{\text{r,PPT}}}}
\nc{\ERLOCCinfty}{{E^{\infty}_{\text{r,LOCC}}}}
\nc{\Aram}{{\operatorname{\sf A}}}

\nc{\cptp}{\operatorname{CPTP}}
\nc{\hptp}{\operatorname{HPTP}}

\usepackage{tikz}
\usepackage{hyperref}
\hypersetup{colorlinks=true,citecolor=blue,linkcolor=blue,filecolor=blue,urlcolor=blue,breaklinks=true}

\makeatletter
\def\grd@save@target#1{%
  \def\grd@target{#1}}
\def\grd@save@start#1{%
  \def\grd@start{#1}}
\tikzset{
  grid with coordinates/.style={
    to path={%
      \pgfextra{%
        \edef\grd@@target{(\tikztotarget)}%
        \tikz@scan@one@point\grd@save@target\grd@@target\relax
        \edef\grd@@start{(\tikztostart)}%
        \tikz@scan@one@point\grd@save@start\grd@@start\relax
        \draw[minor help lines,magenta] (\tikztostart) grid (\tikztotarget);
        \draw[major help lines] (\tikztostart) grid (\tikztotarget);
        \grd@start
        \pgfmathsetmacro{\grd@xa}{\the\pgf@x/1cm}
        \pgfmathsetmacro{\grd@ya}{\the\pgf@y/1cm}
        \grd@target
        \pgfmathsetmacro{\grd@xb}{\the\pgf@x/1cm}
        \pgfmathsetmacro{\grd@yb}{\the\pgf@y/1cm}
        \pgfmathsetmacro{\grd@xc}{\grd@xa + \pgfkeysvalueof{/tikz/grid with coordinates/major step}}
        \pgfmathsetmacro{\grd@yc}{\grd@ya + \pgfkeysvalueof{/tikz/grid with coordinates/major step}}
        \foreach \x in {\grd@xa,\grd@xc,...,\grd@xb}
        \node[anchor=north] at (\x,\grd@ya) {\pgfmathprintnumber{\x}};
        \foreach \y in {\grd@ya,\grd@yc,...,\grd@yb}
        \node[anchor=east] at (\grd@xa,\y) {\pgfmathprintnumber{\y}};
      }
    }
  },
  minor help lines/.style={
    help lines,
    step=\pgfkeysvalueof{/tikz/grid with coordinates/minor step}
  },
  major help lines/.style={
    help lines,
    line width=\pgfkeysvalueof{/tikz/grid with coordinates/major line width},
    step=\pgfkeysvalueof{/tikz/grid with coordinates/major step}
  },
  grid with coordinates/.cd,
  minor step/.initial=.2,
  major step/.initial=1,
  major line width/.initial=2pt,
}
\makeatother

\usepackage{thmtools}
\usepackage{thm-restate}
\usepackage{etoolbox}
\makeatletter
\def\problem@s{}
\newcounter{problems@cnt}

\newcommand{\allproblems}{\problem@s}
\makeatother

\usepackage{multirow,bm,booktabs}
\usepackage{mathrsfs}
\usepackage[utf8]{inputenc}
\usepackage[T1]{fontenc}
\newcommand{\priorcell}[2]{\parbox[t]{#1}{\raggedright #2}}
\allowdisplaybreaks
\begin{document}
\title{Exact Virtual Channel Programming with Vanishing Excess Overhead}
\author{Mingrui Jing}
\affiliation{Thrust of Artificial Intelligence, Information Hub,\\
The Hong Kong University of Science and Technology (Guangzhou), Guangzhou 511453, China}
\affiliation{QudeLeap Research, Shanghai 200030, China}
\author{Mengbo Guo}
\affiliation{Thrust of Artificial Intelligence, Information Hub,\\
The Hong Kong University of Science and Technology (Guangzhou), Guangzhou 511453, China}
\author{Hongshun Yao}
\affiliation{Thrust of Artificial Intelligence, Information Hub,\\
The Hong Kong University of Science and Technology (Guangzhou), Guangzhou 511453, China}
\affiliation{QudeLeap Research, Shanghai 200030, China}
\author{Xin Wang}
\email{felixxinwang@hkust-gz.edu.cn}
\affiliation{Thrust of Artificial Intelligence, Information Hub,\\
The Hong Kong University of Science and Technology (Guangzhou), Guangzhou 511453, China}

\begin{abstract}
A finite-dimensional physical processor cannot exactly program a continuous family of distinct unitary channels. We show that this obstruction becomes quantitative when the target channel is stored in a normalized Choi state and its output observables are reconstructed by sampling physical channels and classically post-processing their measurement outcomes. For arbitrary $d$-dimensional channels, we construct a target-independent exact reconstruction protocol and prove the optimal one-copy sampling overhead, which grows quadratically with system dimension. We further prove the sharp fixed-$d$ law that the excess overhead vanishes inversely with the number of identical Choi programs. The upper bound combines deterministic port-based teleportation with a quasi-decomposition that corrects its depolarizing distortion. The converse maps any low-overhead reconstruction protocol to a physical learner of unknown unitaries and uses local quantum estimation to recover the same leading coefficient. These results recast the universal no-programming obstruction as a quantitative trade-off between quantum program memory and classical sampling, with a leading cost that reflects the locally learnable unitary degrees of freedom.
\end{abstract}

\maketitle
\emph{Introduction.}--- Programmable quantum processing uses a
quantum state to select the operation performed by a fixed
device~\cite{nielsen1997programmable}, allowing dynamics to be stored,
transmitted, and retrieved without redesigning the
hardware~\cite{Vidal2002storing,sedlak2019optimal,sedlak2024storage,Yoshida2026quantum}.
This interface separates the information in the program from the processing
power of the retriever, a resource-theoretic distinction~\cite{Chitambar2019quantum}.
Exact deterministic programming nevertheless faces a quantum obstruction:
program states for distinct unitaries must be orthogonal, so no
finite-dimensional register can exactly encode a continuous unitary
family~\cite{nielsen1997programmable}.

Quasiprobability methods offer a different route: they replace an unavailable
transformation with randomized sampling of physical channels and classical
reweighting of measured outcomes. This technique underlies error
mitigation and general quasiprobability methods~\cite{endo2018practical,Piveteau2022quasiprobability,Gherardini2024quasiprobabilities},
virtual resource distillation and nonlinear information
recovery~\cite{Huggins2021virtual,Yuan2024virtual,Zhao2024retrieving}, virtual
channel transformations~\cite{Regula2021operational,zhu2024reversing}, and
hardware demonstrations~\cite{Zhang2024experimental}.
For programming, its appeal is that both the program state and every sampled
evolution remain physical. The quasiprobability reconstruction enters only after
measurement. Its total weight then directly controls the estimator
variance. Figure~\ref{fig:main_fig} contrasts physical channel implementation
with exact reconstruction of output statistics. It also motivates our central
question: can more copies of a physical program reduce the  sampling
cost of exact reconstruction?

\begin{figure}[h!]
    \centering
    \includegraphics[width=0.98\linewidth]{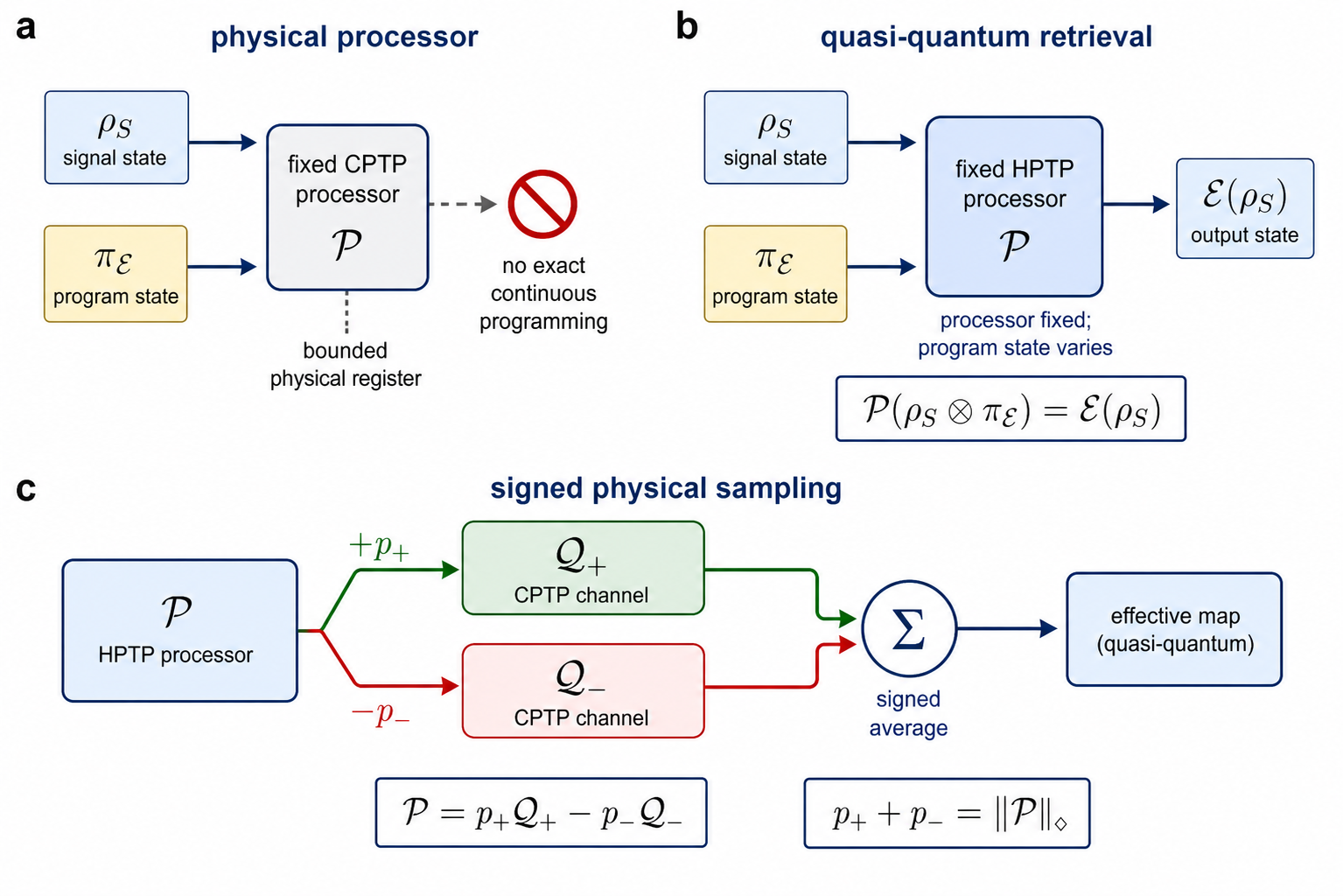}
    \caption{Exact virtual channel programming. (a) A fixed physical processor
    with a finite-dimensional program register cannot exactly realize a
    continuous family of distinct unitary channels. (b) A physical
    normalized Choi state $\pi_\cE$ selects the target channel for a fixed
    linear retrieval rule $\cP$. Its averaged action satisfies
    $\cP(\rho_S\ox\pi_\cE)=\cE(\rho_S)$ at the level of reconstructed
    output observables. (c) Each trial samples one of the physical channels
    in a norm-optimal quasi-decomposition
    $\cP=p_+\cQ_+-p_-\cQ_-$ and classically reweights the measured outcome.
    The total weight
    $p_++p_-=\|\cP\|_\diamond$ quantifies the cost of $\cP$. Minimizing
    this total weight over exact retrievers defines the programming overhead.}
    \label{fig:main_fig}
\end{figure}

\begin{table*}[t]
    \centering
    \caption{Resource scaling for programming all $d$-dimensional
    channels. Physical retrieval minimizes program dimension, whereas exact
    virtual retrieval fixes Choi memory and minimizes signed-sampling overhead
    $\nu$ (worst-case shot factor $\nu^2$). The approximate physical entry is a
    lower bound valid for every $\alpha<(d^2-1)/2$ at fixed
    $d$~\cite{Yang2020optimal,Gschwendtner2021programmabilityof}.}
    \label{tab:resource_axes}
    \vspace{2pt}
    \small
    \setlength{\tabcolsep}{5pt}
    \renewcommand{\arraystretch}{1.12}
    \begin{tabular*}{\textwidth}{@{\extracolsep{\fill}}lccc@{}}
        \toprule
        Retrieval &
        Uniform error &
        Program dimension &
        Quasi-sampling overhead \\
        \midrule
        Physical, exact &
        $\varepsilon=0$ &
        $d_{\rm P}^{\star}(0)=\infty$ &
        $1$ \\
        Physical, approximate &
        $\varepsilon\to0$ &
        $d_{\rm P}^{\star}(\varepsilon)
        =\Omega(\varepsilon^{-\alpha})$ &
        $1$ \\
        \addlinespace[2pt]
        \textbf{Virtual exact, $k=1$} &
        $\varepsilon=0$ &
        $d_{\rm P}^{\rm Choi}=d^2$ &
        $\nu_1=2d^2-3+2/d^2$ \\
        \textbf{Virtual exact, $k\to\infty$} &
        $\varepsilon=0$ &
        $d_{\rm P}^{\rm Choi}=d^{2k}$ &
        $\nu_k=1+(d^2-1)/(2k)+o(k^{-1})$ \\
        \bottomrule
    \end{tabular*}
\end{table*}

Physical processors relax the no-programming obstruction through distinct
compromises. With finite memory, a universal unitary processor can only
approximate its targets
deterministically~\cite{Kubicki2019resource,Yang2020optimal}.
Low-depth brickwork circuits obey a distinct joint large-system
program-cost law for
inverse-polylogarithmic error~\cite{He2026resource}. Exact unitary
retrieval can instead be conditioned on a successful
outcome~\cite{Hillery2002probabilistic,Vidal2002storing,sedlak2019optimal,sedlak2020probabilistic}. Exact deterministic
programming is also possible when the target is restricted, for example to
input-irreducible covariant
channels~\cite{Gschwendtner2021programmabilityof}. Thus approximation
error, success probability, target restriction, observable-specific
inversion~\cite{Zhen2026shadow}, and learning a Lindbladian from its physical
time evolution~\cite{ChenYu2026learning} define different operational tasks. This distinction leads to two complementary resource questions. Let
$d_{\rm P}^{\star}(\varepsilon)$ denote the minimum program dimension of a
deterministic physical retriever with uniform diamond distance error at most
$\varepsilon$. Our setting instead fixes $k$ product normalized Choi programs
and minimizes $\nu_k(\cS)$, the quasiprobability sampling overhead of exact observable
reconstruction. Table~\ref{tab:resource_axes} compares the two resource
scalings for the common target $\Op{CPTP}_d$.

In this Letter, we determine this memory--sampling trade-off for all
$d$-dimensional channels. We use $k$ product normalized Choi programs and one
target-independent protocol, with every trial applying a physical channel.
For one copy, we construct an exact protocol and prove the global optimum
$\nu_1(\Op{CPTP}_d)=2d^2-3+2/d^2$. For $k$ copies, we establish the sharp
fixed-$d$ law
$\nu_k(\Op{CPTP}_d)=1+(d^2-1)/(2k)+o(k^{-1})$. Achievability combines
deterministic port-based teleportation with a quasi-decomposition that corrects its
depolarizing distortion~\cite{Christandl2021asymptotic}. The converse converts
any lower-overhead protocol into a physical learner of unknown
unitaries~\cite{Bisio2010optimal}, linking the coefficient $(d^2-1)/2$ to the
local geometry of channel learning. Exact results for unitary, unital, real,
and covariant families isolate the roles of symmetry and affine structure.
Together, these results establish a quantitative trade-off between quantum
program memory and classical sampling. They show that physical
channel implementation and exact statistical reconstruction are distinct
operational notions of quantum programmability.

\vspace{2mm}
\emph{Exact virtual programming.}--- Exact virtual programming keeps every
program state and sampled operation physical. Exactness is required only of
the reconstructed output statistics. Let $\Op{CPTP}_d$ denote the quantum
channels on a $d$-dimensional system, with $d\geq2$. A target family
$\cS\subseteq\Op{CPTP}_d$ is encoded by physical states $\pi_\cE$. Given
$k$ identical program copies, a fixed linear retrieval rule $\cP$ induces
$\widetilde{\cE}_{\cP}(\cdot):=\cP(\cdot\ox\pi_\cE^{\ox k})$.

Operationally, $\cP$ is evaluated by sampling physical channels and
reweighting their measurement outcomes. A quasi-decomposition
$\cP=p_+\cQ_+-p_-\cQ_-$, with $p_\pm\geq0$ and
$\cQ_\pm\in\Op{CPTP}$, gives such an implementation using only physical
channels~\cite{Piveteau2022quasiprobability}. We call a
Hermiticity-preserving and trace-preserving (HPTP) rule a
\emph{quasi-quantum retriever}. When $\cP$ is CPTP, it is a
\emph{quantum retriever}. In finite dimensions, the minimum of $p_++p_-$
over all quasi-decompositions equals
$\|\cP\|_\diamond$~\cite{Regula2021operational,Jiang2021physical}. We
therefore take the diamond norm as the programming overhead. This
signed-weight construction belongs to the broader lineage of generalized
robustness for quantum states~\cite{Vidal1999robustness,Steiner2003generalized},
quasiprobability costs for quantum operations~\cite{seddon2019quantifying},
and resource theories formulated directly for quantum
channels~\cite{wang2019resource}.

We use the normalized Choi state $\pi_\cE:=J_\cE/d$ as the
program~\cite{wilde2013quantum}.
This encoding also underlies channel retrieval by
port-based teleportation
(PBT)~\cite{ishizaka2008,beigi2011,Studzinski2017port,mozrzymas2018,Christandl2021asymptotic}.
To allow nonzero reconstruction error, we follow the error-tolerant
virtual-process framework~\cite{Takagi2024general}. For $k\geq1$ and
$\varepsilon\geq0$, we define the uniform $\varepsilon$-approximate $k$-copy
programming overhead as
\begin{equation}\label{eq:overhead_def}
\nu_{k,\varepsilon}(\cS)
:=\min_{\cP\in\Op{HPTP}}
\left\{\|\cP\|_\diamond\;\middle|\;
\sup_{\cE\in\cS}\frac{1}{2}
\bigl\|\widetilde{\cE}_{\cP}-\cE\bigr\|_\diamond
\leq\varepsilon\right\}.
\end{equation}
Let $\cP$ be feasible and choose a norm-optimal quasi-decomposition. Write
$w=p_++p_-=\|\cP\|_\diamond$. Trace preservation gives
$p_+-p_-=1$. For a reference-assisted input state $\rho_{RS}$ and program
$\pi_\cE^{\ox k}$, one trial samples $\cI_R\ox\cQ_+$ with probability
$p_+/w$ or $\cI_R\ox\cQ_-$ with probability $p_-/w$. It then measures
a Hermitian observable $O_{RB}$ with $\|O_{RB}\|_\infty\leq1$ and records
$X\in[-1,1]$. The reported value is $Z=wX$ for the positive branch and
$Z=-wX$ for the negative branch. This estimator satisfies
\begin{equation*}
\begin{aligned}
\mathbb{E}[Z]
=\tr\!\left[O_{RB}(\cI_R\ox\widetilde{\cE}_{\cP})(\rho_{RS})\right],
\end{aligned}
\end{equation*}
where $|Z|\leq w, \operatorname{Var}(Z)\leq w^2$.
At the exact endpoint $\varepsilon=0$, the estimator is unbiased for
the target expectation value,
$\mathbb{E}[Z]=\tr[O_{RB}(\cI_R\ox\cE)(\rho_{RS})]$. Throughout, $\log$
denotes the base-two logarithm. Hoeffding's inequality then shows that
$N\geq2w^2\eta^{-2}\log(2/\delta)$ independent trials
achieve additive error $\eta$ with failure probability at most $\delta$.
The protocol reconstructs output observables without implementing $\cE$ as
a reusable physical channel.

Each trial consumes one such input state and all $k$ normalized Choi
programs. The full program register has dimension $d^{2k}$, equivalent to
$2k\log d$ qubit-equivalents. If each fresh program copy requires one target-channel
use, $N$ trials consume $kN$ such uses. The factor $w^2$ accounts only for
quasiprobability sampling. It excludes program preparation, memory lifetime,
retriever implementation, measurement, classical post-processing, and
fault-tolerant costs.

For $\varepsilon>0$, the estimator remains unbiased for the retrieved map,
while $\varepsilon$ controls its bias relative to the target. Additional
trials reduce sampling error but not this approximation bias. Following
Ref.~\cite{jing2025programmable}, a channel family is quantum programmable
when a quantum retriever is feasible. It is quasi-quantum programmable when
feasibility requires quasiprobability sampling of physical channels. In finite dimensions,
$\nu_{k,\varepsilon}(\cS)\geq1$, with equality exactly when quantum retrieval
is feasible. We focus below on $\varepsilon=0$ and write
$\nu_k(\cS):=\nu_{k,0}(\cS)$.

At $\varepsilon=0$, the Choi--Jamio{\l}kowski representation turns
Eq.~\eqref{eq:overhead_def} into an SDP for any finite target set. For the
covariant families below, the one-copy SDP reduces to finite linear programs.
The many-copy SDP admits a block reduction through walled Brauer symmetry.
The exact one-copy overhead also satisfies three resource laws for every
nonempty target family. Enlarging the target family cannot decrease the
overhead. Extending the family to every physical channel in its real affine
span leaves the overhead unchanged. Parallel composition makes the overhead
submultiplicative. These properties follow from feasible-set inclusion, linearity, and tensor
products of optimal retrievers. Proofs are given in Supplemental Material,
Sec.~\ref{sec:property}.

\vspace{2mm}
\emph{Universal one-copy optimum.}--- We first consider the family of
all $d$-dimensional channels. This family is invariant under independent
unitary rotations of the input and output. Averaging over these rotations
places the retriever's Choi operator in a four-dimensional commutant spanned
by projectors onto four joint invariant subspaces. The exact retrieval
conditions then select a unique covariant map.

\begin{theorem}\label{thm:cptp}
Let $\cS=\Op{CPTP}_d$ be the family of all $d$-dimensional quantum
channels. The map $\cP_*$ defined below is the unique covariant exact
retriever for $\cS$. It satisfies
$\cP_*(\rho\otimes\pi_\cE)=\cE(\rho)$ for every state $\rho$ and every
$\cE\in\cS$. Among all exact sampling-and-post-processing protocols in
Eq.~\eqref{eq:overhead_def}, $\cP_*$ attains the global minimum
$\nu_1(\Op{CPTP}_d)=2d^2-3+\frac{2}{d^2}$.
\end{theorem}

\emph{Construction and optimality.} For an input operator $X$ on
$\cH_S\ox\cH_{P_1}\ox\cH_{P_2}$, define
\begin{equation*}
\begin{aligned}
\cM(X)&:=\tr_{SP_1}\!\left[(\Omega_{SP_1}\ox I_{P_2})X\right],\\
\cP_*(X)&:=\left(\frac{\tr X}{d}-\tr[\cM(X)]\right)I_{S'}
    +d\,\cM(X).
\end{aligned}
\end{equation*}
Here $\Omega_{AB}=\dketbra{I_d}{I_d}$, with
$\dket{I_d}:=\sum_{i=0}^{d-1}\ket{i}_A\ket{i}_B$, is the unnormalized
maximally entangled operator on registers $A$ and $B$. We write $\Omega_d$
when the register labels are clear. The contraction $\cM$ joins the input $S$
to the Choi input $P_1$ and retains the Choi output $P_2$. The remaining term
makes $\cP_*$ trace preserving. Substituting $X=\rho_S\ox\pi_\cE$ gives
$\cP_*(\rho_S\ox\pi_\cE)=\cE(\rho_S)$ and proves exact retrieval.

Twirling the physical branches over independent input and output rotations
preserves their total weight. It also restricts the retriever's Choi operator
to the four-dimensional commutant. There, the exact retrieval conditions
uniquely determine $\cP_*$. It admits the following quasi-decomposition
\begin{equation*}
J_{\cP_*}=\mu_+J_{\cQ_+}-\mu_-J_{\cQ_-},
\end{equation*}
with $\mu_+=d^2-1+d^{-2}$ and $\mu_-=d^2-2+d^{-2}$.
The two physical channels have Choi operators
\begin{equation*}
\begin{aligned}
J_{\cQ_+}
&=\frac{I_{d^2}\ox I_{d^2}}{d}
  -\frac{\Omega_d\ox I_{d^2}}{d^2}
  +\frac{\Omega_d\ox\Omega_d}{d},\\
J_{\cQ_-}
&=\frac{I_{d^2}\ox I_{d^2}}{d}
  +\frac{\Omega_d\ox I_{d^2}}{d^2(d^2-1)}
  -\frac{\Omega_d\ox\Omega_d}{d(d^2-1)}.
\end{aligned}
\end{equation*}
Tensor products of $\Phi_d=\Omega_d/d$ and $I-\Phi_d$ define four joint
invariant subspaces. On these subspaces, both Choi operators are positive and
satisfy $\tr_{S'}J_{\cQ_\pm}=I_{SP}$. Thus $\cQ_\pm$ are physical channels,
and $\mu_++\mu_-=2d^2-3+2/d^2$ is achievable. For the lower bound, take the
normalized state $\sigma=\Omega_{SP_1}/d\ox\Omega_{P_2E}/d$. Its output is
\begin{equation*}
(\cP_*\ox\cI_E)(\sigma)
=\left(\frac1d-d\right)\frac{I_{d^2}}d+d\,\Omega_{P_2E},
\end{equation*}
with $\bigl\|(\cP_*\ox\cI_E)(\sigma)\bigr\|_1
=2d^2-3+\frac{2}{d^2}$. This output trace norm lower-bounds
$\|\cP_*\|_\diamond$ and matches the total weight above. Every feasible
retriever twirls to $\cP_*$ without increasing its diamond norm. This proves
global optimality, but not uniqueness outside the covariant sector.
\hfill$\square$

Theorem~\ref{thm:cptp} gives exact one-copy reconstruction for every
$d$-dimensional channel. Its norm-optimal quasi-decomposition
produces unbiased estimates of output observables with overhead that scales
quadratically with $d$. This reconstruction
is statistical and does not provide a reusable physical channel.
Channel--state duality stores all linear channel information without making a
Choi state physically executable. Interestingly, the protocol is a sampling-based one-copy
analogue of teleportation~\cite{Bennett1993teleporting}. The many-copy result
below is also closely related to the same Choi-port idea through deterministic port-based
teleportation~\cite{ishizaka2008,Christandl2021asymptotic}.

\vspace{2mm}
\emph{Many-copy trade-off.}---A $k$-copy trial consumes $k$ identical
normalized Choi programs under the resource convention above. Discarding one
program gives $\nu_{k+1}(\cS)\leq\nu_k(\cS)$ for every nonempty $\cS$.
Additional copies therefore cannot increase the overhead. The remaining
question is whether the overhead approaches its lower bound of one and at what
rate. For all $d$-dimensional channels, this approach follows the sharp
$1/k$ law stated below.

\begin{theorem}\label{thm:kcopy_main}
Fix an integer $d\geq2$ and use the product program
$\pi_\cE^{\ox k}$ in the exact model of Eq.~\eqref{eq:overhead_def}. For each
$k$, the retriever may depend on $d$ and $k$ but remains independent of the
target $\cE$. Then
\begin{equation*}
\nu_k(\Op{CPTP}_d)
=1+\frac{d^2-1}{2k}+o(k^{-1}),
\quad k\to\infty.
\end{equation*}
Here $k$ runs through the positive integers with $d$ fixed. The
remainder need not be uniform in $d$.
\end{theorem}

\emph{Achievability and converse.} Deterministic PBT with $k$ Choi programs
realizes $\cE\circ\cD_{\eta_{k,d}}$, where
$\cD_\eta(X):=\eta X+(1-\eta)\tr(X)I_d/d$ and $\eta_{k,d}$ is the PBT
shrinkage factor. At fixed $d$,
$1-\eta_{k,d}=d^2/(4k)+o(k^{-1})$~\cite{Christandl2021asymptotic}.
Before PBT, sample a physical branch from a quasi-decomposition of
$\cD_{\eta_{k,d}}^{-1}$ and reweight the measured outcome. This removes the
distortion exactly. The norm identity
$\|\cD_\eta^{-1}\|_\diamond=1+2(1-d^{-2})(\eta^{-1}-1)$ gives
$\limsup_{k\to\infty}k[\nu_k(\Op{CPTP}_d)-1]\leq(d^2-1)/2$.

For the converse, take any exact retriever $\cP$ and a norm-optimal
quasi-decomposition into physical channels
$\cP=p_+\cQ_+-p_-\cQ_-$, where $p_+-p_-=1$ and
$2p_-=\|\cP\|_\diamond-1$~\cite{Regula2021operational}. For the unitary
channel $\cU(\rho):=U\rho U^\dagger$, its program induces
$\cQ_\pm^U(\rho):=\cQ_\pm(\rho\ox\pi_\cU^{\ox k})$. Exact retrieval gives
\begin{equation*}
\cQ_+^U-\cU=p_-\bigl(\cQ_-^U-\cQ_+^U\bigr),
\quad
\|\cQ_+^U-\cU\|_\diamond\leq\|\cP\|_\diamond-1.
\end{equation*}
Let $R_k^{\rm av}$ denote the minimum Haar-averaged gate infidelity
over target-independent physical learners supplied with these $k$ programs.
The positive branch is admissible, so
$R_k^{\rm av}\leq d[\|\cP\|_\diamond-1]/[2(d+1)]$.
A fixed Schur transform and target-independent channels establish statistical
equivalence with the representation memory used in optimal unitary learning.
The corresponding local-estimation bound gives
$\liminf_{k\to\infty}kR_k^{\rm av}\geq d(d^2-1)/[4(d+1)]$
~\cite{Bisio2010optimal,BraunsteinCaves1994statistical,GillLevit1995applications}.
Combining the risk bounds and minimizing over exact retrievers gives
$\liminf_{k\to\infty}k[\nu_k(\Op{CPTP}_d)-1]\geq(d^2-1)/2$.
The Supplemental Material provides the memory conversion and
local-estimation details.

The leading coefficient has a local geometric interpretation.
Unitary channels are locally parameterized by $\Op{SU}(d)$ modulo its finite
center and therefore have $d^2-1$ identifiable directions. The coefficient is
half this local dimension. Because the converse uses only the unitary
subfamily, this local dimension already fixes the leading universal lower
bound. This identification is specific to the present programming model and is not a
general Lie-group law. For qubits, the theorem becomes
$\nu_k(\Op{CPTP}_2)=1+3/(2k)+o(k^{-1})$.

The quantity $\nu_k^2$ is the worst-case quasiprobability shot factor
for independent trials, not an end-to-end cost. It excludes program
preparation, memory lifetime, retriever implementation, measurement,
classical post-processing, and fault-tolerant costs. Since each trial consumes
$k$ programs, the product of the per-trial program count and worst-case shot
factor obeys
$k\nu_k^2=k+d^2-1+o(1)$. This relation does not establish fewer
target-channel queries or an end-to-end speedup. Equivalently, let
$d_{\rm P}^{\rm Choi}=d^{2k}$ denote the prescribed product-program
dimension. The theorem gives
$\nu_k-1=(d^2-1)\log d/\log d_{\rm P}^{\rm Choi}
+o((\log d_{\rm P}^{\rm Choi})^{-1})$. The excess overhead therefore scales
inversely with the logarithm of the memory dimension in this fixed encoding.

Probabilistic retrieval gives a one-way comparison when the memory is
held fixed. Suppose a retriever uses the same product-Choi ensemble and
succeeds uniformly with probability $q_k$, independent of the input and
target. Then $\nu_k(\cS)\leq2/q_k-1$. This bound permits a scaling comparison
when $q_k$ approaches one. Its converse does not follow, so the two frameworks
are not operationally equivalent. In particular, success probabilities
optimized over different program memories cannot be inserted into this
same-Choi bound~\cite{sedlak2019optimal}.

\vspace{2mm}
\emph{Further exact one-copy families.}---Restricted target families
show how symmetry and affine structure determine the exact one-copy
overhead. Unitary channels are the standard benchmark for programmable gate
arrays~\cite{nielsen1997programmable,Vidal2002storing,sedlak2019optimal}.
\begin{proposition}\label{prop:main_unitary_onecopy}
For $d\geq2$, the exact one-copy overhead for all $d$-dimensional unitary
channels is $\nu_1(\Op{Ad}_{\Op{SU}(d)})=d^2-1$.
\end{proposition}
Independent input--output covariance reduces the optimization to a four-sector
linear program. A matching primal--dual pair certifies the stated optimum.

Let $\cT(d)$ be the full unital family. It is the \emph{real affine} closure
of the unitary family~\cite{mendl2009unital}. The affine-invariance property
established above therefore gives $\nu_1(\cT(d))=d^2-1$. This conclusion does
not rely on a convex-mixture representation. The unital-family protocol covers
random-unitary targets encoded by their Choi programs. This family-level
statement does not identify any particular random circuit as an optimal
retriever or imply that it forms a unitary design~\cite{Brandao2016local}.

Although the restriction below is basis dependent, it is conceptually
adjacent to the operational distinction between real and complex quantum
theory~\cite{Renou2021real} and to resource theories that quantify
imaginarity~\cite{HickeyGour2018imaginarity,Wu2021operational,Wu2021resource}.
For comparison, let $\cL^\mathbb{R}$ denote channels whose Choi matrices are
real in the fixed computational basis. Orthogonal covariance reduces their
optimization to three invariant sectors. Together with the controlled-map
diamond-norm identity proved in the Supplemental Material, a KKT analysis of
the resulting minimax problem yields
$\nu_1(\cL^\mathbb{R})=\frac{2}{3}d^2+O(1)$. The restriction to real channels
therefore lowers the leading coefficient from $1$ to $2/3$ while preserving
the quadratic one-copy scaling.

Input--output covariance also reduces the finite-copy SDP. Its ambient
dimension grows exponentially, but both positive Choi variables lie in a
group commutant. Regrouping the tensor factors yields the mixed sectors
$U\ox(U^*)^{\ox k}$ and $V^{\ox k}\ox V^*$ for independent representations
$U$ and $V$. Mixed Schur--Weyl duality identifies their commutants with
represented images of $B_{1,k}(d)$ and $B_{k,1}(d)$~\cite{BCHLLS94}.
The isotypic decomposition replaces global positivity with semidefinite blocks
on multiplicity spaces. Trace preservation and programming constraints remain
linear. Finite-dimensional diagram relations can further reduce these
images~\cite{DipperDotyStoll2014}. Enforcing the programming constraints on a
basis of the finite-dimensional Choi-power span preserves equivalence with the
unreduced problem. The block reduction is used only for finite-copy numerical calculations, whose
results are reported in the Supplemental Material and do not enter the
analytic proof of either universal theorem.

\vspace{2mm}
\emph{Concluding remarks.}---Quasiprobability sampling turns the universal
no-programming obstruction into a quantitative memory--sampling trade-off.
With one normalized Choi program, the exact universal optimum is
$\nu_1(\Op{CPTP}_d)=2d^2-3+2/d^2$. For $k$ identical programs in product
form, the overhead approaches its lower bound as
$\nu_k(\Op{CPTP}_d)=1+\frac{d^2-1}{2k}+o(k^{-1})$ at fixed dimension.
The coefficient $(d^2-1)/2$ is fixed by the locally identifiable unitary
directions, connecting this resource law to the geometry of channel learning.

Operationally, each trial applies only a physical channel, while classical
reweighting reconstructs output observables exactly in expectation. The
protocol does not yield a reusable implementation of the target channel, and
$\nu_k^2$ captures only its worst-case sampling factor. A complete resource
account should also include program preparation, memory lifetime, and
retriever complexity. It remains open whether correlated program memories or
more general many-copy post-processing can improve finite-copy trade-offs. It
is also unknown whether probabilistic retrieval and exact virtual programming
can be compared sharply under a common memory constraint. These questions
point toward a resource theory that treats quantum memory, probabilistic
success, and classical sampling on the same operational footing.

\emph{Acknowledgments.}---This work was partially supported by the National Natural Science Foundation of China (Grant No.92576114, 12447107), the Guangdong Provincial Quantum Science Strategic Initiative (Grant No.~GDZX2403008, GDZX2503001), and the Guangdong Provincial Key Lab of Integrated Communication, Sensing and Computation for Ubiquitous Internet of Things (Grant No.~2023B1212010007).

\emph{Statement on AI use.}---OpenAI ChatGPT 5.6 assisted with language
editing and with developing and presenting technical arguments in the
many-copy analysis. OpenAI ChatGPT was also used to generate the
conceptual illustration in Fig.~\ref{fig:main_fig}; the authors designed its
scientific content and verified the final image. The authors supplied the problem formulation, proof
strategy, source material, and revision criteria. They independently
verified all mathematical statements, references, and final wording and
take full responsibility for the manuscript.

\emph{Data availability.}---The numerical data and custom code
supporting this work are publicly accessible in the GitHub repository of
Ref.~\cite{QuAIR2026UniversalProgrammingCodes}. The repository documents
the finite channel ensembles, solver settings, and numerical residual
checks used for the reduced-SDP calculations.

\bibliography{ref}

@misc{wildeQuantumFisherInformation2025,
  title = {Quantum {{Fisher}} Information Matrices from {{Rényi}} Relative Entropies},
  author = {Wilde, Mark M.},
  year = {2025},
  eprint = {2510.02218},
  eprinttype = {arXiv},
  eprintclass = {quant-ph},
  doi = {10.48550/arXiv.2510.02218},
  url = {http://arxiv.org/abs/2510.02218},
  urldate = {2025-11-06},
  pubstate = {prepublished}
}

@article{fujiwara1995quantum,
  title={Quantum Fisher metric and estimation for pure state models},
  author={Fujiwara, Akio and Nagaoka, Hiroshi},
  journal={Physics Letters A},
  volume={201},
  number={2-3},
  pages={119--124},
  year={1995},
  publisher={Elsevier}
}

@article{Chitambar2019quantum,
  title = {Quantum resource theories},
  author = {Chitambar, Eric and Gour, Gilad},
  journal = {Rev. Mod. Phys.},
  volume = {91},
  issue = {2},
  pages = {025001},
  numpages = {48},
  year = {2019},
  month = {Apr},
  publisher = {American Physical Society},
  doi = {10.1103/RevModPhys.91.025001},
  url = {https://link.aps.org/doi/10.1103/RevModPhys.91.025001}
}

@article{ballester2004estimation,
  title={Estimation of unitary quantum operations},
  author={Ballester, Manuel A},
  journal={Physical Review A},
  volume={69},
  number={2},
  pages={022303},
  year={2004},
  publisher={APS},
  doi={10.1103/PhysRevA.69.022303},
  eprint={quant-ph/0305104},
  archivePrefix={arXiv},
  primaryClass={quant-ph}
}

@article{wang2019resource,
  title={Resource theory of asymmetric distinguishability for quantum channels},
  author={Wang, Xin and Wilde, Mark M},
  journal={Physical Review Research},
  volume={1},
  number={3},
  pages={033169},
  year={2019},
  publisher={APS}
}

@article{Bennett1993teleporting,
  title={Teleporting an unknown quantum state via dual classical and Einstein-Podolsky-Rosen channels},
  author={Bennett, Charles H and Brassard, Gilles and Cr{\'e}peau, Claude and Jozsa, Richard and Peres, Asher and Wootters, William K},
  journal={Physical review letters},
  volume={70},
  number={13},
  pages={1895--1899},
  year={1993},
  publisher={APS},
  doi={10.1103/PhysRevLett.70.1895}
}

@article{Vidal1999robustness,
  title={Robustness of entanglement},
  author={Vidal, G. and Tarrach, R.},
  journal={Physical Review A},
  volume={59},
  number={1},
  pages={141--155},
  year={1999},
  publisher={American Physical Society}
}

@article{Steiner2003generalized,
  title={Generalized robustness of entanglement},
  author={Steiner, M.},
  journal={Physical Review A},
  volume={67},
  number={5},
  pages={054305},
  year={2003},
  publisher={American Physical Society}
}

@article{knill2000theory,
  title={Theory of quantum error correction for general noise},
  author={Knill, Emanuel and Laflamme, Raymond and Viola, Lorenza},
  journal={Physical Review Letters},
  volume={84},
  number={11},
  pages={2525--2528},
  year={2000},
  publisher={APS},
  doi={10.1103/PhysRevLett.84.2525},
  eprint={quant-ph/9908066},
  archivePrefix={arXiv},
  primaryClass={quant-ph}
}

@article{Gschwendtner2021programmabilityof,
  title = {Programmability of covariant quantum channels},
  author = {Gschwendtner, Martina and Bluhm, Andreas and Winter, Andreas},
  journal = {{Quantum}},
  issn = {2521-327X},
  publisher = {{Verein zur F{\"{o}}rderung des Open Access Publizierens in den Quantenwissenschaften}},
  volume = {5},
  pages = {488},
  month = jun,
  year = {2021},
  doi = {10.22331/q-2021-06-29-488},
  eprint = {2012.00717},
  archivePrefix = {arXiv},
  primaryClass = {quant-ph}
}

@article{Piveteau2022quasiprobability,
  title={Quasiprobability decompositions with reduced sampling overhead},
  author={Piveteau, Christophe and Sutter, David and Woerner, Stefan},
  journal={npj Quantum Information},
  volume={8},
  number={1},
  pages={12},
  year={2022},
  publisher={Nature Publishing Group UK London}
}

@article{Huggins2021virtual,
  title={Virtual Distillation for Quantum Error Mitigation},
  author={Huggins, William J. and McArdle, Sam and O'Brien, Thomas E. and Lee, Joonho and Rubin, Nicholas C. and Boixo, Sergio and Whaley, K. Birgitta and Babbush, Ryan and McClean, Jarrod R.},
  journal={Physical Review X},
  volume={11},
  number={4},
  pages={041036},
  year={2021},
  publisher={American Physical Society},
  doi={10.1103/PhysRevX.11.041036},
  url={https://doi.org/10.1103/PhysRevX.11.041036}
}

@article{Yuan2024virtual,
  title={Virtual Quantum Resource Distillation},
  author={Yuan, Xiao and Regula, Bartosz and Takagi, Ryuji and Gu, Mile},
  journal={Physical Review Letters},
  volume={132},
  number={5},
  pages={050203},
  year={2024},
  publisher={American Physical Society},
  doi={10.1103/PhysRevLett.132.050203},
  url={https://doi.org/10.1103/PhysRevLett.132.050203}
}

@article{Takagi2024general,
  title={Virtual Quantum Resource Distillation: General Framework and Applications},
  author={Takagi, Ryuji and Yuan, Xiao and Regula, Bartosz and Gu, Mile},
  journal={Physical Review A},
  volume={109},
  number={2},
  pages={022403},
  year={2024},
  publisher={American Physical Society},
  doi={10.1103/PhysRevA.109.022403},
  url={https://doi.org/10.1103/PhysRevA.109.022403}
}

@article{Zhang2024experimental,
  title={Experimental Virtual Distillation of Entanglement and Coherence},
  author={Zhang, Ting and Zhang, Yukun and Liu, Lu and Fang, Xiao-Xu and Zhang, Qian-Xi and Yuan, Xiao and Lu, He},
  journal={Physical Review Letters},
  volume={132},
  number={18},
  pages={180201},
  year={2024},
  publisher={American Physical Society},
  doi={10.1103/PhysRevLett.132.180201},
  url={https://doi.org/10.1103/PhysRevLett.132.180201}
}

@article{seddon2019quantifying,
  title = {Quantifying Magic for Multi-Qubit Operations},
  author = {Seddon, James R. and Campbell, Earl T.},
  year = 2019,
  month = jul,
  journal = {Proceedings of the Royal Society A: Mathematical, Physical and Engineering Sciences},
  volume = {475},
  number = {2227},
  pages = {20190251},
  publisher = {Royal Society}
}

@article{endo2018practical,
  title={Practical quantum error mitigation for near-future applications},
  author={Endo, Suguru and Benjamin, Simon C and Li, Ying},
  journal={Physical Review X},
  volume={8},
  number={3},
  pages={031027},
  year={2018},
  publisher={APS}
}

@article{chiribella2008quantum,
  title={Quantum circuit architecture},
  author={Chiribella, Giulio and D'Ariano, G Mauro and Perinotti, Paolo},
  journal={Physical review letters},
  volume={101},
  number={6},
  pages={060401},
  year={2008},
  publisher={APS},
  doi={10.1103/PhysRevLett.101.060401}
}

@article{zhu2024reversing,
  title={Reversing unknown quantum processes via virtual combs for channels with limited information},
  author={Zhu, Chengkai and Mo, Yin and Chen, Yu-Ao and Wang, Xin},
  journal={Physical Review Letters},
  volume={133},
  number={3},
  pages={030801},
  year={2024},
  publisher={APS}
}

@article{Gherardini2024quasiprobabilities,
  title={Quasiprobabilities in quantum thermodynamics and many-body systems},
  author={Gherardini, Stefano and De Chiara, Gabriele},
  journal={PRX Quantum},
  volume={5},
  number={3},
  pages={030201},
  year={2024},
  publisher={APS}
}

@article{Zhao2024retrieving,
  title={Retrieving nonlinear features from noisy quantum states},
  author={Zhao, Benchi and Jing, Mingrui and Zhang, Lei and Zhao, Xuanqiang and Chen, Yu-Ao and Wang, Kun and Wang, Xin},
  journal={PRX Quantum},
  volume={5},
  number={2},
  pages={020357},
  year={2024},
  publisher={APS}
}

@article{Yoshida2026quantum,
  title={Quantum advantage in storage and retrieval of isometry channels},
  author={Yoshida, Satoshi and Miyazaki, Jisho and Murao, Mio},
  journal={Physical Review Letters},
  volume={136},
  number={19},
  pages={190601},
  year={2026},
  publisher={American Physical Society},
  doi={10.1103/fdvq-9m8m},
  eprint={2507.10784},
  archivePrefix={arXiv},
  primaryClass={quant-ph}
}

@article{Yoshida2026erratum,
  title={Erratum: Quantum advantage in storage and retrieval of isometry channels [{Phys. Rev. Lett.} {136}, {190601} ({2026})]},
  author={Yoshida, Satoshi and Miyazaki, Jisho and Murao, Mio},
  journal={Physical Review Letters},
  year={2026},
  note={Accepted 21 August 2026; version of record pending publication},
  publisher={American Physical Society},
  doi={10.1103/34df-l36k}
}

@article{Jiang2021physical,
  title={Physical implementability of linear maps and its application in error mitigation},
  author={Jiang, Jiaqing and Wang, Kun and Wang, Xin},
  journal={Quantum},
  volume={5},
  pages={600},
  year={2021},
  doi={10.22331/q-2021-12-07-600},
  eprint={2012.10959},
  archivePrefix={arXiv},
  primaryClass={quant-ph},
  publisher={Verein zur F{\"o}rderung des Open Access Publizierens in den Quantenwissenschaften}
}

@article{Regula2021operational,
  title={Operational applications of the diamond norm and related measures in quantifying the non-physicality of quantum maps},
  author={Regula, Bartosz and Takagi, Ryuji and Gu, Mile},
  journal={Quantum},
  volume={5},
  pages={522},
  year={2021},
  doi={10.22331/q-2021-08-09-522},
  eprint={2102.07773},
  archivePrefix={arXiv},
  primaryClass={quant-ph},
  publisher={Verein zur F{\"o}rderung des Open Access Publizierens in den Quantenwissenschaften}
}

@article{leditzky2018approaches,
  title={Approaches for approximate additivity of the Holevo information of quantum channels},
  author={Leditzky, Felix and Kaur, Eneet and Datta, Nilanjana and Wilde, Mark M},
  journal={Physical Review A},
  volume={97},
  number={1},
  pages={012332},
  year={2018},
  publisher={APS},
  doi={10.1103/PhysRevA.97.012332}
}

@article{tomamichel2016strong,
  title={Strong converse rates for quantum communication},
  author={Tomamichel, Marco and Wilde, Mark M and Winter, Andreas},
  journal={IEEE Transactions on Information Theory},
  volume={63},
  number={1},
  pages={715--727},
  year={2017},
  publisher={IEEE},
  doi={10.1109/TIT.2016.2615847}
}

@book{boyd2004convex,
  title={Convex optimization},
  author={Boyd, Stephen and Vandenberghe, Lieven},
  year={2004},
  publisher={Cambridge university press}
}

@article{mendl2009unital,
  title={Unital quantum channels--convex structure and revivals of Birkhoff's theorem},
  author={Mendl, Christian B and Wolf, Michael M},
  journal={Communications in Mathematical Physics},
  volume={289},
  number={3},
  pages={1057--1086},
  year={2009},
  publisher={Springer}
}

@article{Brandao2016local,
  title={Local random quantum circuits are approximate polynomial-designs},
  author={Brand{\~a}o, Fernando G. S. L. and Harrow, Aram W. and Horodecki, Micha{\l}},
  journal={Communications in Mathematical Physics},
  volume={346},
  number={2},
  pages={397--434},
  year={2016},
  publisher={Springer},
  doi={10.1007/s00220-016-2706-8},
  eprint={1208.0692},
  archivePrefix={arXiv},
  primaryClass={quant-ph}
}

@book{wilde2013quantum,
  title={Quantum information theory},
  author={Wilde, Mark M},
  year={2013},
  publisher={Cambridge university press}
}

@article{nielsen1997programmable,
  title={Programmable quantum gate arrays},
  author={Nielsen, Michael A and Chuang, Isaac L},
  journal={Physical Review Letters},
  volume={79},
  number={2},
  pages={321--324},
  year={1997},
  publisher={APS},
  doi={10.1103/PhysRevLett.79.321},
  eprint={quant-ph/9703032},
  archivePrefix={arXiv},
  primaryClass={quant-ph}
}

@article{Vidal2002storing,
  title={Storing quantum dynamics in quantum states: A stochastic programmable gate},
  author={Vidal, G. and Masanes, L. and Cirac, J. I.},
  journal={Physical Review Letters},
  volume={88},
  number={4},
  pages={047905},
  year={2002},
  doi={10.1103/PhysRevLett.88.047905}
}

@article{Hillery2002probabilistic,
  title={Probabilistic implementation of universal quantum processors},
  author={Hillery, M. and Buzek, V. and Ziman, M.},
  journal={Physical Review A},
  volume={65},
  number={2},
  pages={022301},
  year={2002},
  doi={10.1103/PhysRevA.65.022301}
}

@article{Renou2021real,
  title={Quantum theory based on real numbers can be experimentally falsified},
  author={Renou, Marc-Olivier and Trillo, David and Weilenmann, Mirjam and Le, Thinh P. and Tavakoli, Armin and Gisin, Nicolas and Ac{\'i}n, Antonio and Navascu{\'e}s, Miguel},
  journal={Nature},
  volume={600},
  number={7890},
  pages={625--629},
  year={2021},
  doi={10.1038/s41586-021-04160-4}
}

@article{HickeyGour2018imaginarity,
  title={Quantifying the imaginarity of quantum mechanics},
  author={Hickey, Alexander and Gour, Gilad},
  journal={Journal of Physics A: Mathematical and Theoretical},
  volume={51},
  number={41},
  pages={414009},
  year={2018},
  doi={10.1088/1751-8121/aabe9c}
}

@article{Wu2021operational,
  title={Operational resource theory of imaginarity},
  author={Wu, Kang-Da and Kondra, Tulja Varun and Rana, Swapan and Scandolo, Carlo Maria and Xiang, Guo-Yong and Li, Chuan-Feng and Guo, Guang-Can and Streltsov, Alexander},
  journal={Physical Review Letters},
  volume={126},
  number={9},
  pages={090401},
  year={2021},
  doi={10.1103/PhysRevLett.126.090401}
}

@article{Wu2021resource,
  title={Resource theory of imaginarity: Quantification and state conversion},
  author={Wu, Kang-Da and Kondra, Tulja Varun and Rana, Swapan and Scandolo, Carlo Maria and Xiang, Guo-Yong and Li, Chuan-Feng and Guo, Guang-Can and Streltsov, Alexander},
  journal={Physical Review A},
  volume={103},
  number={3},
  pages={032401},
  year={2021},
  doi={10.1103/PhysRevA.103.032401}
}

@article{jing2025programmable,
  title={Programmable Open Quantum Systems},
  author={Jing, Mingrui and Guo, Mengbo and Zhu, Lin and Yao, Hongshun and Wang, Xin},
  journal={Physical Review Letters},
  volume={137},
  number={4},
  pages={040403},
  year={2026},
  month={jul},
  publisher={American Physical Society},
  doi={10.1103/yqlr-2dhr},
  url={https://doi.org/10.1103/yqlr-2dhr},
  eprint={2512.08279},
  archivePrefix={arXiv},
  primaryClass={quant-ph}
}

@article{ChenYu2026learning,
  title={Learning Arbitrary Lindbladians from Time Evolution},
  author={Chen, Zhili and Yu, Zhan},
  journal={arXiv preprint arXiv:2607.28610},
  year={2026},
  month={jul},
  eprint={2607.28610},
  archivePrefix={arXiv},
  primaryClass={quant-ph},
  url={https://arxiv.org/abs/2607.28610}
}

@article{Zhen2026shadow,
  title={Structure, Optimality, and Symmetry in Shadow Unitary Inversion},
  author={Zhen, Guocheng and Chen, Yu-Ao and Jing, Mingrui and Xie, Jingu and Wang, Xin and Chen, Ranyiliu},
  journal={Communications Physics},
  year={2026},
  month={may},
  doi={10.1038/s42005-026-02690-9},
  url={https://doi.org/10.1038/s42005-026-02690-9},
  eprint={2510.24880},
  archivePrefix={arXiv},
  primaryClass={quant-ph}
}

@article{He2026resource,
  title={Resource quantification for programming low-depth quantum circuits},
  author={He, Entong and Yang, Yuxiang},
  journal={Quantum},
  volume={10},
  pages={2166},
  year={2026},
  month={jul},
  doi={10.22331/q-2026-07-20-2166},
  url={https://doi.org/10.22331/q-2026-07-20-2166},
  eprint={2509.09642},
  archivePrefix={arXiv},
  primaryClass={quant-ph}
}

@article{Kubicki2019resource,
  title={Resource Quantification for the No-Programing Theorem},
  author={Kubicki, Aleksander M. and Palazuelos, Carlos and P{\'e}rez-Garc{\'i}a, David},
  journal={Physical Review Letters},
  volume={122},
  number={8},
  pages={080505},
  year={2019},
  publisher={American Physical Society},
  doi={10.1103/PhysRevLett.122.080505}
}

@article{Yang2020optimal,
  title={Optimal Universal Programming of Unitary Gates},
  author={Yang, Yuxiang and Renner, Renato and Chiribella, Giulio},
  journal={Physical Review Letters},
  volume={125},
  number={21},
  pages={210501},
  year={2020},
  publisher={American Physical Society},
  doi={10.1103/PhysRevLett.125.210501},
  eprint={2007.10363},
  archivePrefix={arXiv},
  primaryClass={quant-ph}
}

@article{Banchi2020convex,
  title={Convex optimization of programmable quantum computers},
  author={Banchi, Leonardo and Pereira, Jason and Lloyd, Seth and Pirandola, Stefano},
  journal={npj Quantum Information},
  volume={6},
  pages={42},
  year={2020},
  publisher={Springer Nature},
  doi={10.1038/s41534-020-0268-2}
}

@article{sedlak2019optimal,
  title={Optimal probabilistic storage and retrieval of unitary channels},
  author={Sedl{\'a}k, Michal and Bisio, Alessandro and Ziman, M{\'a}rio},
  journal={Physical Review Letters},
  volume={122},
  number={17},
  pages={170502},
  year={2019},
  publisher={American Physical Society},
  doi={10.1103/PhysRevLett.122.170502},
  eprint={1809.04552},
  archivePrefix={arXiv},
  primaryClass={quant-ph}
}

@article{sedlak2024storage,
  title={Storage and retrieval of two unknown unitary channels},
  author={Sedl{\'a}k, Michal and St{\'a}rek, Robert and Horov{\'a}, Nikola and Mi{\v{c}}uda, Michal and Fiur{\'a}{\v{s}}ek, Jaromir and Bisio, Alessandro},
  journal={arXiv preprint arXiv:2410.23376},
  year={2024}
}

@article{sedlak2020probabilistic,
  title={Probabilistic storage and retrieval of qubit phase gates},
  author={Sedl{\'a}k, Michal and Ziman, M{\'a}rio},
  journal={Physical Review A},
  volume={102},
  number={3},
  pages={032618},
  year={2020},
  publisher={APS}
}

@article{BCHLLS94,
  title={Tensor product representations of general linear groups and their connections with {B}rauer algebras},
  author={Benkart, Georgia and Chakrabarti, Manish and Halverson, Thomas and Leduc, Robert and Lee, Chanyoung and Stroomer, Jeffrey},
  journal={Journal of Algebra},
  volume={166},
  number={3},
  pages={529--567},
  year={1994},
  publisher={Elsevier},
  doi={10.1006/jabr.1994.1166}
}

@article{DipperDotyStoll2014,
  title={The quantized walled {B}rauer algebra and mixed tensor space},
  author={Dipper, Richard and Doty, Stephen and Stoll, Friederike},
  journal={Algebras and Representation Theory},
  volume={17},
  number={2},
  pages={675--701},
  year={2014},
  doi={10.1007/s10468-013-9414-2},
  eprint={0806.0264},
  archivePrefix={arXiv},
  primaryClass={math.QA},
  publisher={Springer}
}

@phdthesis{stroomer1991,
  title={Combinatorics and the Representation Theory of {GL(r,C)} and {Sp(2r,C)}},
  author={Stroomer, Jeffrey Dean},
  year={1991},
  school={University of Wisconsin--Madison}
}

@article{brauer1937,
  title={On algebras which are connected with the semisimple continuous groups},
  author={Brauer, Richard},
  journal={Annals of Mathematics},
  volume={38},
  number={4},
  pages={857--872},
  year={1937}
}

@article{CDDM05,
  title={On the blocks of the walled {B}rauer algebra},
  author={Cox, Anton and De Visscher, Maud and Doty, Stephen and Martin, Paul},
  journal={Journal of Algebra},
  volume={320},
  number={1},
  pages={169--212},
  year={2008},
  publisher={Elsevier}
}

@article{ishizaka2008,
  title={Asymptotic teleportation scheme as a universal programmable quantum processor},
  author={Ishizaka, Satoshi and Hiroshima, Tohya},
  journal={Physical Review Letters},
  volume={101},
  number={24},
  pages={240501},
  year={2008},
  publisher={APS},
  doi={10.1103/PhysRevLett.101.240501}
}

@article{beigi2011,
  title={Simplified instantaneous non-local quantum computation with applications to position-based cryptography},
  author={Beigi, Salman and K{\"o}nig, Robert},
  journal={New Journal of Physics},
  volume={13},
  number={9},
  pages={093036},
  year={2011},
  publisher={IOP Publishing}
}

@article{mozrzymas2018,
  title={Optimal port-based teleportation},
  author={Mozrzymas, Marek and Studzi{\'n}ski, Micha{\l} and Strelchuk, Sergii and Horodecki, Micha{\l}},
  journal={New Journal of Physics},
  volume={20},
  number={5},
  pages={053006},
  year={2018},
  publisher={IOP Publishing}
}

@article{Christandl2021asymptotic,
  title={Asymptotic performance of port-based teleportation},
  author={Christandl, Matthias and Leditzky, Felix and Majenz, Christian and Smith, Graeme and Speelman, Florian and Walter, Michael},
  journal={Communications in Mathematical Physics},
  volume={381},
  pages={379--451},
  year={2021},
  doi={10.1007/s00220-020-03884-0}
}

@article{Studzinski2017port,
  title={Port-based teleportation in arbitrary dimension},
  author={Studzi{\'n}ski, Micha{\l} and Strelchuk, Sergii and Mozrzymas, Marek and Horodecki, Micha{\l}},
  journal={Scientific Reports},
  volume={7},
  pages={10871},
  year={2017},
  doi={10.1038/s41598-017-10051-4}
}

@article{Kretschmann2008continuity,
  title={The information-disturbance tradeoff and the continuity of Stinespring's representation},
  author={Kretschmann, Dennis and Schlingemann, Dirk and Werner, Reinhard F.},
  journal={IEEE Transactions on Information Theory},
  volume={54},
  number={4},
  pages={1708--1717},
  year={2008},
  doi={10.1109/TIT.2008.917696}
}

@article{Nielsen2002average,
  title={A simple formula for the average gate fidelity of a quantum dynamical operation},
  author={Nielsen, Michael A.},
  journal={Physics Letters A},
  volume={303},
  number={4},
  pages={249--252},
  year={2002},
  doi={10.1016/S0375-9601(02)01272-0},
  eprint={quant-ph/0205035},
  archivePrefix={arXiv}
}

@article{Bisio2010optimal,
  title={Optimal quantum learning of a unitary transformation},
  author={Bisio, Alessandro and Chiribella, Giulio and D'Ariano, Giacomo Mauro and Facchini, Stefano and Perinotti, Paolo},
  journal={Physical Review A},
  volume={81},
  number={3},
  pages={032324},
  year={2010},
  doi={10.1103/PhysRevA.81.032324},
  eprint={0903.0543},
  archivePrefix={arXiv},
  primaryClass={quant-ph}
}

@article{BraunsteinCaves1994statistical,
  title={Statistical distance and the geometry of quantum states},
  author={Braunstein, Samuel L. and Caves, Carlton M.},
  journal={Physical Review Letters},
  volume={72},
  number={22},
  pages={3439--3443},
  year={1994},
  doi={10.1103/PhysRevLett.72.3439}
}

@article{GillLevit1995applications,
  title={Applications of the van {T}rees inequality: A {B}ayesian {C}ram{\'e}r--{R}ao bound},
  author={Gill, Richard D. and Levit, Boris Y.},
  journal={Bernoulli},
  volume={1},
  number={1--2},
  pages={59--79},
  year={1995},
  doi={10.2307/3318681}
}

@article{Matsumoto2002pure,
  title={A new approach to the {C}ram{\'e}r--{R}ao-type bound of the pure-state model},
  author={Matsumoto, Keiji},
  journal={Journal of Physics A: Mathematical and General},
  volume={35},
  number={13},
  pages={3111--3123},
  year={2002},
  doi={10.1088/0305-4470/35/13/307},
  eprint={quant-ph/9711008},
  archivePrefix={arXiv}
}

@article{Yang2019attaining,
  title={Attaining the ultimate precision limit in quantum state estimation},
  author={Yang, Yuxiang and Chiribella, Giulio and Hayashi, Masahito},
  journal={Communications in Mathematical Physics},
  volume={368},
  number={1},
  pages={223--293},
  year={2019},
  doi={10.1007/s00220-019-03433-4},
  eprint={1802.07587},
  archivePrefix={arXiv},
  primaryClass={quant-ph}
}

@misc{QuAIR2026UniversalProgrammingCodes,
  author={{QuAIR}},
  title={{UniversalProgramming-Codes}},
  year={2026},
  howpublished={GitHub repository},
  url={https://github.com/QuAIR/UniversalProgramming-Codes},
  note={Commit 55e15c8, accessed 8 August 2026}
}

\clearpage
\appendix
\onecolumngrid
\setcounter{secnumdepth}{1}
\makeatletter
\@removefromreset{equation}{section}
\makeatother
\setcounter{equation}{0}
\renewcommand{\theequation}{S\arabic{equation}}
\renewcommand{\theHequation}{S\arabic{equation}}
\setcounter{proposition}{0}
\renewcommand{\theproposition}{S\arabic{proposition}}
\setcounter{definition}{0}
\renewcommand{\thedefinition}{S\arabic{definition}}
\setcounter{remark}{0}
\renewcommand{\theremark}{S\arabic{remark}}
\setcounter{figure}{0}
\renewcommand{\thefigure}{S\arabic{figure}}
\setcounter{table}{0}
\renewcommand{\thetable}{S\arabic{table}}
\renewcommand{\theHproposition}{S\arabic{proposition}}
\renewcommand{\theHdefinition}{S\arabic{definition}}
\providecommand{\theHremark}{}
\renewcommand{\theHremark}{S\arabic{remark}}
\renewcommand{\theHfigure}{S\arabic{figure}}
\renewcommand{\theHtable}{S\arabic{table}}

\begin{center}
\large{\textbf{Supplemental Material for
``Exact Virtual Channel Programming with Vanishing Excess Overhead''}}
\end{center}

This Supplemental Material is organized by proof dependency rather than by
appearance in the main text. Appendix~\ref{appendix:notation} fixes notation,
normalization conventions, and the one-copy programming SDP used throughout.
Appendix~\ref{sec:covariance_reduction} collects the symmetry reductions
used by the closed-form calculations. The one-copy
evaluations for all quantum channels, unitary channels, and real channels are
proved in Appendices~\ref{sec:cptp_appendix}, \ref{sec:unitary_appendix}, and
\ref{sec:real_channel}. Appendix~\ref{sec:family_comparison} records
supporting comparisons across prior programming frameworks and restricted
target families.
Appendix~\ref{sec:property} then collects structural
properties of the overhead, including invariance, stability under processing,
affine invariance, and tensor-product behavior. Appendix~\ref{sec:kcopy}
treats the many-copy setting and its reduced SDP.

\section{Notation and the one-copy programming SDP}\label{appendix:notation}

This appendix fixes the notation, normalization conventions, and the
one-copy SDP formulation used in the proofs below. Let $\cH_d$ denote a
$d$-dimensional Hilbert space. The
signal and program systems are labeled $S$ and $P$, with Hilbert spaces
$\cH_S$ and $\cH_P$ of dimensions $d_S$ and $d_P$, and $S'$ denotes the
signal output. The program register splits as
$\cH_P=\cH_{P_1}\ox\cH_{P_2}$, where $P_1$ carries the input side of the
programmed channel and $P_2$ its output side, so that
$\cH_{P_1}\simeq\cH_S$ and $\cH_{P_2}\simeq\cH_{S'}$. The space of bounded
linear operators on $\cH_d$ is $\cB(\cH_d)$, and its positive-semidefinite
cone is denoted $\operatorname{Pos}(\cH_d)$. For $A\in\cB(\cH_d)$, $\overline{A}$ is the entrywise
complex conjugate in the computational basis, $A^T$ the transpose, and
$A^\dagger=\overline{A}^T=\overline{A^T}$ the adjoint. A density operator
$\rho\in\operatorname{Pos}(\cH_d)$ satisfies $\tr\rho=1$. The set of such operators is
$\cD(\cH_d)$, and a pure state $\psi=\ketbra{\psi}{\psi}$ is the rank-one
case with $\ket{\psi}\in\CC^d$. The trace norm is
$\|A\|_1=\tr\sqrt{A^\dagger A}$.

Maximally entangled states enter through two related objects, kept
typographically distinct to avoid normalization ambiguities. The
unnormalized maximally entangled vector is
$\dket{I_d}:=\sum_{i=0}^{d-1}\ket{ii}$, with associated rank-one
positive operator
\begin{equation*}
    \Omega_d:=\dketbra{I_d}{I_d}=\sum_{i,j=0}^{d-1}\ketbra{ii}{jj},
    \qquad \Omega_d^2=d\,\Omega_d,\quad \tr\Omega_d=d.
\end{equation*}
The normalized maximally entangled state is $\Phi_d:=\Omega_d/d$, so that
$\tr\Phi_d=1$. The operator $\Omega_d$, carrying register subscripts when
needed, is the object that appears in Choi constructions and commutant
expansions, whereas $\Phi_d$ is reserved for the normalized state. The
flip operator on $\cH_d\ox\cH_d$ is written $\mathbb{F}$ and acts as
$\mathbb{F}\ket{ij}=\ket{ji}$.

Linear maps from $\cB(\cH_A)$ to $\cB(\cH_B)$ are denoted $\cN$ or
$\cN_{A\to B}$, and $\mathscr{L}(\cH_A\to\cH_B)$ is the space of all such
maps. Its Hermitian-preserving trace-preserving and completely-positive
trace-preserving subsets are $\Op{HPTP}(\cH_A\to\cH_B)$ and
$\Op{CPTP}(\cH_A\to\cH_B)$, so that
$\Op{CPTP}(\cH_A\to\cH_B)\subseteq\Op{HPTP}(\cH_A\to\cH_B)\subseteq
\mathscr{L}(\cH_A\to\cH_B)$. The shorthand $\Op{CPTP}(A\to B)$ is used
where the spaces are clear, and $\cI$ denotes the identity map. For a
unitary $U$ on $\cH_A$, $\cU:=\Op{Ad}_U$ denotes the corresponding
unitary channel, $\cU(X)=UXU^\dagger$, and
$\cU^\dagger:=\Op{Ad}_{U^\dagger}$ denotes its inverse. For a unitary representation
$g\mapsto U_g$, we write $\cU_g:=\Op{Ad}_{U_g}$ when the underlying
Hilbert space is specified locally. The Choi
operator of $\cE\in\mathscr{L}(\cH_A\to\cH_B)$, with
$d:=\dim\cH_A$, is
\begin{equation*}
    J_\cE:=(\cI_A\ox\cE)(\Omega_d)
    =\sum_{i,j=0}^{d-1}\ketbra{i}{j}\ox\cE(\ketbra{i}{j})
    \in\cB(\cH_A\ox\cH_B),
\end{equation*}
and the normalized Choi state, which serves as the program state of the
target channel, is $\pi_\cE:=J_\cE/d$.
For a quantum channel $\Lambda$ on $\cH_d$, its entanglement fidelity
with respect to the maximally mixed input is
\begin{equation*}
F_{\rm e}(\Lambda)
:=\tr\!\left[\Phi_d(\cI\ox\Lambda)(\Phi_d)\right].
\end{equation*}
If $\cM$ is a quantum channel and $\cU=\Op{Ad}_U$ is a unitary
channel on the same system, their average gate fidelity is
\begin{equation}\label{eq:average_gate_fidelity}
\begin{aligned}
F_{\rm av}(\cM,\cU)
&:=\int
\bra{\psi}U^\dagger\cM(\ketbra{\psi}{\psi})U\ket{\psi}\,d\psi\\
&=\frac{d\,F_{\rm e}(\Op{Ad}_{U^\dagger}\circ\cM)+1}{d+1},
\end{aligned}
\end{equation}
where $d\psi$ is the normalized unitarily invariant measure on pure
states. The second equality is Nielsen's identity~\cite[Eq.~(3)]{Nielsen2002average}.
For $\cE\in\mathscr{L}(\cH_A\to\cH_B)$ and
$\cF\in\mathscr{L}(\cH_B\to\cH_C)$, composition is represented by the
\emph{link product}~\cite{chiribella2008quantum},
\begin{equation*}
J_\cE\star J_\cF := \tr_B\!\Big[\big(J_\cE^{T_B}\ox I_C\big)
\big(I_A\ox J_\cF\big)\Big] = J_{\cF\circ\cE},
\end{equation*}
where $T_B$ and $\tr_B$ act on $\cH_B$. The link product is commutative
and associative whenever the labeled contractions are compatible.
The quasi-quantum retriever is denoted $\cP$, and the HPTP recovery map in
Appendix~\ref{sec:property} is denoted $\cR$. As in the main
text, we use
$\Op{CPTP}_d:=\Op{CPTP}(\cH_d\to\cH_d)$ for the family of all
channels on $\cH_d$.

\begin{definition}[Approximate and exact programming overheads]
\label{def:koverhead}
Let
$\varnothing\ne\cS\subseteq\Op{CPTP}(\cH_A\to\cH_B)$,
$k\in\mathbb{N}$ with $k\geq1$, and
$\varepsilon\geq0$. Set $d_A:=\dim\cH_A$. For the normalized Choi programs
$\pi_\cE=J_\cE/d_A$, the uniform $\varepsilon$-approximate $k$-copy
quasi-quantum programming overhead is
\begin{equation}\label{eq:supp_approx_overhead}
\nu_{k,\varepsilon}(\cS)
:=\min_{\cP\in\Op{HPTP}}
\left\{\|\cP\|_\diamond\;\middle|\;
\sup_{\cE\in\cS}\frac{1}{2}
\left\|\cP(\,\cdot\ox\pi_\cE^{\ox k})-\cE\right\|_\diamond
\leq\varepsilon\right\},
\end{equation}
where
$\cP:\cB\!\left(\cH_A\ox(\cH_A\ox\cH_B)^{\ox k}\right)
\to\cB(\cH_B)$ is independent of
$\cE$. At the exact endpoint, the error constraint in
Eq.~\eqref{eq:supp_approx_overhead} is equivalent to equality of the
retrieved and target maps. We therefore define
\begin{equation}\label{eq:supp_exact_overhead}
\begin{aligned}
\nu_k(\cS)
&:=\nu_{k,0}(\cS)\\
&=\min_{\cP\in\Op{HPTP}}
\bigl\{\|\cP\|_\diamond:\,
\cP(\rho\ox\pi_\cE^{\ox k})=\cE(\rho),\,
\forall\rho\in\cD(\cH_A),\ \forall\cE\in\cS\bigr\},
\end{aligned}
\end{equation}
\end{definition}

For the all-channel family, these quantities are denoted
$\nu_{k,\varepsilon}(\Op{CPTP}_d)$ and
$\nu_k(\Op{CPTP}_d)=\nu_{k,0}(\Op{CPTP}_d)$.

\paragraph{Estimator induced by a quasi-decomposition.}
Fix a feasible retriever $\cP$ and a diamond-norm-optimal
quasi-decomposition into quantum channels~\cite[Theorem~3]{Regula2021operational}
\begin{equation*}
\cP=p_+\cQ_+-p_-\cQ_-,
\qquad
w:=p_++p_-=\|\cP\|_\diamond,
\end{equation*}
where $\cQ_\pm$ are quantum channels and $p_\pm\geq0$. Trace
preservation implies $p_+-p_-=1$, so the branch probabilities below sum
to one. Given any reference-assisted input state $\rho_{RA}$ and the
program $\pi_\cE^{\ox k}$, draw $s=+1$ with probability $p_+/w$ or
$s=-1$ with probability $p_-/w$. Apply $\cI_R\ox\cQ_s$ to
$\rho_{RA}\ox\pi_\cE^{\ox k}$ and measure a Hermitian observable
$O_{RB}$ with $\|O_{RB}\|_\infty\leq1$, obtaining an eigenvalue outcome
$X\in[-1,1]$. Write
$\sigma_s:=(\cI_R\ox\cQ_s)(\rho_{RA}\ox\pi_\cE^{\ox k})$. For
$Z:=swX$, these conditional output states give
\begin{equation*}
\begin{aligned}
\mathbb E[Z]
&=p_+\tr(O_{RB}\sigma_+)-p_-\tr(O_{RB}\sigma_-)\\
&=\tr\!\left[
O_{RB}(\cI_R\ox\cP)(\rho_{RA}\ox\pi_\cE^{\ox k})
\right].
\end{aligned}
\end{equation*}
Thus $Z$ is unbiased for the retrieved-map observable. At the
exact endpoint its expectation is
$\tr[O_{RB}(\cI_R\ox\cE)(\rho_{RA})]$. Pointwise,
$|Z|\leq w$, and hence
\begin{equation*}
\operatorname{Var}(Z)
\leq\mathbb E[Z^2]
\leq w^2.
\end{equation*}

For $N$ independent trials with fresh branch draws and fresh program
inputs, each trial consuming all $k$ Choi programs, Hoeffding's
inequality for variables in $[-w,w]$ yields
\begin{equation*}
\Pr\!\left\{
\left|\overline Z_N-\mathbb E[Z]\right|\geq\eta
\right\}
\leq
2\exp\!\left(-\frac{N\eta^2}{2w^2}\right).
\end{equation*}
Therefore, for $\eta>0$ and $0<\delta<1$, it is sufficient to take
\begin{equation*}
N\geq\frac{2w^2}{\eta^2}\log\frac{2}{\delta}.
\end{equation*}

For a feasible approximate retriever, define
$\Delta_\cE:=\cP(\,\cdot\ox\pi_\cE^{\ox k})-\cE$. The constraint in
Eq.~\eqref{eq:supp_approx_overhead} gives
$\|\Delta_\cE(\rho_A)\|_1\leq2\varepsilon$. Write this output as $Y$.
It is Hermitian and traceless because both maps preserve
trace. Thus a binary effect $0\leq M\leq I$ and a general Hermitian
observable $\|O\|_\infty\leq1$ obey
\begin{equation*}
|\tr(MY)|\leq\frac{1}{2}\|Y\|_1\leq\varepsilon,
\qquad
|\tr(OY)|\leq\|Y\|_1\leq2\varepsilon.
\end{equation*}
Consequently, with
probability at least $1-\delta$, the total error relative to the target is
at most $\eta+\varepsilon$ for a binary-event probability and at most
$\eta+2\varepsilon$ for a general norm-one observable. Increasing $N$
reduces $\eta$ and leaves the approximation bias $\varepsilon$ unchanged.

For completeness, exact one-copy retrieval is feasible also when
$d_A\ne d_B$. Set $d_B:=\dim\cH_B$ and, for an operator $X$ on the
signal and one program register, define
\begin{equation*}
\begin{aligned}
\cM(X)
&:=\tr_{S P_1}\!\left[
(\Omega_{S P_1}\ox I_{P_2})X\right],\\
\cP_{A\to B}(X)
&:=d_A\cM(X)
+\frac{\tr X-d_A\tr[\cM(X)]}{d_B}\,I_B .
\end{aligned}
\end{equation*}
This map is HPTP and satisfies
$\cP_{A\to B}(\rho\ox J_\cE/d_A)=\cE(\rho)$ for every
$\cE\in\Op{CPTP}(\cH_A\to\cH_B)$.
The minima in Eqs.~\eqref{eq:supp_approx_overhead} and
\eqref{eq:supp_exact_overhead} are attained in finite dimensions.
Indeed, discarding the extra program copies turns this universal
one-copy retriever into a feasible $k$-copy point, the constraints are
closed, and diamond-norm sublevel sets are compact. Every HPTP map has
diamond norm at least one. If its diamond norm equals one, its
normalized Choi operator is Hermitian with trace one and trace norm at
most one. Its trace norm is therefore exactly one, so the Choi operator
is positive and the map is a quantum channel. Hence
$\nu_{k,\varepsilon}(\cS)\geq1$, with equality if and only if a quantum
retriever meets the same tolerance. Moreover, for
$\varepsilon\geq1$, a fixed quantum retriever that ignores the program
is feasible because the half-diamond distance between any two quantum
channels is at most one. Thus $\nu_{k,\varepsilon}(\cS)=1$ in this
regime, and the nontrivial approximate range is
$0\leq\varepsilon<1$.

For any nonempty one-copy channel set
$\cS\subseteq\Op{CPTP}(\cH_S\to\cH_{S'})$, the following SDP gives the
exact overhead.
\begin{equation}\label{eq:onecopy_SDP}
\begin{aligned}
&\underline{\textbf{Primal program}}\\
    \nu_1(\cS)&=\min\;p_++p_-\\
    {\rm s.t.}\;\;& J_{\cP}:=J_+-J_-,\\
    &\tr_P[J_{\cP}(I_S\ox J_{\cE}^{T}\ox I_{S'})] = d_S J_{\cE},\, \forall \cE \in \cS,\\
    &J_+\geq 0,\,\tr_{S'}[J_+] =p_+ I_{SP},\\
    &J_-\geq 0,\,\tr_{S'}[J_-] =p_- I_{SP}.
\end{aligned}
\end{equation}
Here $p_{\pm}$ are the quasi-decomposition coefficients attached to the
positive semidefinite variables $J_{\pm}$. No separate constraint on
$p_+-p_-$ is needed. Indeed, fix any $\cE\in\cS$ and trace the
programming equality over $S'$. Using $\tr J_\cE=d_S$,
$\tr_{S'}J_\cE=I_S$, and
$\tr_{S'}J_\cP=(p_+-p_-)I_{SP}$ gives
\begin{equation*}
\begin{aligned}
d_S(p_+-p_-)I_S
&=\tr_{P,S'}\!\left[J_\cP
  (I_S\ox J_\cE^T\ox I_{S'})\right]\\
&=d_S\,\tr_{S'}J_\cE=d_S I_S.
\end{aligned}
\end{equation*}
Hence $p_+-p_-=1$, and $J_\cP$ is automatically trace preserving.
The three closed-form
one-copy proofs below use this formulation together with symmetry
reductions of $J_{\cP}$.

\section{Covariance and SDP reduction tools}\label{sec:covariance_reduction}

We use two symmetry reductions throughout the one-copy proofs. The
first evaluates the diamond norm of a covariant map from its Choi
operator. The second restricts the programming SDP to a commutant
algebra. Both reductions are independent of the particular channel
family being programmed.

The following lemma is the only diamond-norm fact needed below. It
adapts the purification and symmetrization argument for covariant
channels in Refs.~\cite{leditzky2018approaches,tomamichel2016strong} to
Hermiticity-preserving maps.
\begin{lemma}[Diamond norm reduction for covariant maps]\label{lem:covariant_diamond_reduction}
    Let $\cE_{A\rightarrow B}$ be a Hermiticity-preserving linear map
    that is jointly covariant
    with respect to unitary representations
    $g\mapsto U_g$ on $\cH_A$ and $g\mapsto V_g$ on $\cH_B$, where
    $G$ is a finite group. Write
    $\cU_g:=\Op{Ad}_{U_g}$ on $\cB(\cH_A)$ and
    $\cV_g:=\Op{Ad}_{V_g}$ on $\cB(\cH_B)$. Joint covariance means
    \begin{equation*}
        \cV_g^\dagger\circ\cE_{A\rightarrow B}\circ\cU_g
        =\cE_{A\rightarrow B},
        \qquad g\in G.
    \end{equation*}
    Set $d_A:=\dim\cH_A$. If $\{U_g\}_{g\in G}$ is a one-design, so that
    \begin{equation*}
        \frac{1}{|G|}\sum_{g\in G}U_g XU_g^\dagger
        =\frac{\tr X}{d_A}I_A
    \end{equation*}
    for every $X\in\cB(\cH_A)$, then
    \begin{equation}
        \|\cE_{A\rightarrow B}\|_\diamond
        =
        \|(\cI_R\otimes \cE_{A\rightarrow B})(\Phi_{RA})\|_1,
    \end{equation}
    where $\Phi_{RA}$ is the normalized maximally entangled state.
\end{lemma}
\begin{proof}
    For a Hermiticity-preserving map between finite-dimensional spaces,
    the diamond norm is attained on a pure state with a reference
    $R\simeq A$. To see this, Hermitian dilation first restricts the
    optimization to Hermitian inputs. Jordan decomposition and the
    triangle inequality then reduce it to a density operator. Convexity
    permits a pure optimizer, and Schmidt compression restricts its
    reference dimension to $d_A$. Thus
    \begin{equation*}
        \|\cE\|_\diamond
        =\max_{\phi_{RA}\ {\rm pure}}
        \|(\cI_R\ox\cE)(\phi_{RA})\|_1.
    \end{equation*}
    Existence follows from compactness of the pure-state set. Let
    $\phi_{RA}=\ketbra{\phi}{\phi}_{RA}$ attain the maximum and let
    $\rho_A=\tr_R\phi_{RA}$. Introduce a register $C$ with basis
    $\{\ket{g}\}_{g\in G}$ and define
    \begin{equation*}
        \ket{\bar\phi}_{CRA}
        :=\frac{1}{\sqrt{|G|}}\sum_{g\in G}
        \ket{g}_C\ox(I_R\ox U_g)\ket{\phi}_{RA}.
    \end{equation*}
    Write $\bar\phi_{CRA}:=\ketbra{\bar\phi}{\bar\phi}_{CRA}$. Its
    marginal on $A$ is
    $\bar\rho_A=|G|^{-1}\sum_g\cU_g(\rho_A)$. Dephasing $C$ after
    applying $\cI_{CR}\ox\cE$ and using trace-norm contractivity on
    Hermitian operators gives
    \begin{equation*}
        \|(\cI_{CR}\ox\cE)(\bar\phi_{CRA})\|_1
        \geq \frac{1}{|G|}\sum_{g\in G}
        \|(\cI_R\ox\cE\circ\cU_g)(\phi_{RA})\|_1.
    \end{equation*}
    Covariance gives $\cE\circ\cU_g=\cV_g\circ\cE$, so every term on
    the right equals $\|\cE\|_\diamond$. Stability of the diamond norm
    under larger reference systems bounds the left-hand side by the same
    value. Hence $\bar\phi_{CRA}$ is also optimal. The one-design condition gives
    $\bar\rho_A=I_A/d_A$. Every purification of this state is related to
    $\Phi_{RA}$ by an isometry on the reference, which preserves the
    output trace norm. The claimed identity follows.
\end{proof}

To state the symmetry reduction without referring to a particular
channel family, we use the following covariance convention.

\begin{definition}[Bipartite $G$-covariance]\label{def:G_covariance}
    Let $G$ be a compact group, and let $g\mapsto U_g$ and
    $g\mapsto V_g$ be continuous unitary representations on the input
    and output Hilbert spaces, respectively. Write
    $\cU_g:=\Op{Ad}_{U_g}$ and $\cV_g:=\Op{Ad}_{V_g}$.
    A channel set $\cS\subseteq\Op{CPTP}_d$ is
    \emph{$(G,U,V)$-covariant} if
    \begin{equation*}
        \cV_g\circ\cE\circ\cU_g^\dagger
        \;\in\;\cS,
        \qquad \forall\,g\in G,\ \forall\,\cE\in\cS.
    \end{equation*}
    The representation pair $(U,V)$ is called
    \emph{self-conjugate} if for every $g\in G$ there exists
    $h\in G$ such that $U_h=U_g^T$ and $V_h=V_g^T$
    up to phases.
    The corresponding representation on
    $\cH_{\mathrm{tot}}:=\cH_S\ox\cH_{P_1}\ox\cH_{P_2}\ox\cH_{S'}$ is
    \begin{equation}\label{eq:induced_rep}
        \varrho_G(g)
        \;:=\;
        (U_g)_S\;\ox\;(U_g)^*_{P_1}\;\ox\;
        (V_g)_{P_2}\;\ox\;(V_g)^*_{S'}.
    \end{equation}
\end{definition}

Several channel families used later fit this convention. Examples include
the full set $\Op{CPTP}_d$ with independent input and output rotations, the
unitary family $\Op{Ad}_{\Op{SU}(d)}$ with independent left and right
rotations, the unital family $\cT(d)$~\cite{mendl2009unital}, and the real-channel
family $\cL^{\mathbb{R}}$. The central structural result is that every
$G$-covariant channel set with self-conjugate representations admits an
optimal quasi-decomposition whose two positive Choi operators lie in
the commutant of~\eqref{eq:induced_rep}.

\begin{theorem}[$G$-symmetrization]\label{thm:G_symmetrisation}
    Let $\cS$ be $(G,U,V)$-covariant with self-conjugate
    representations and let $d\mu(g)$ denote
    normalized Haar measure on $G$. Define the commutant
    \begin{equation*}
        \mathfrak{C}_G
        \;:=\;
        \bigl\{X\in\cB(\cH_{\mathrm{tot}})
        \;:\;[\varrho_G(g),X]=0,\;\forall\,g\in G\bigr\}.
    \end{equation*}
    Then the following statements hold.
    \begin{enumerate}[(i)]
    \item The programming SDP~\eqref{eq:onecopy_SDP} admits an optimal
          quasi-decomposition $J_{\cP_*}=J_{+,*}-J_{-,*}$ with
          $J_{+,*},J_{-,*}\in\mathfrak{C}_G$.
    \item If $n:=\dim_{\CC}\mathfrak{C}_G$, each of the two Hermitian
          variables $J_{+,*}$ and $J_{-,*}$ is specified by $n$ real
          coefficients in the Hermitian part of the commutant.
    \end{enumerate}
\end{theorem}
\begin{proof}
    Fix an optimal quasi-decomposition $J_{\cP}=J_+-J_-$ with
    $J_\pm\geq 0$, $\tr_{S'}[J_\pm]=p_\pm I_{SP}$, and
    $p_++p_-=\nu_1(\cS)$. For each $g\in G$ set
    \begin{equation*}
        J_\pm^g \;:=\; \varrho_G(g)\,J_\pm\,\varrho_G(g)^\dagger.
    \end{equation*}
    Because $\varrho_G(g)$ is unitary, positivity and trace preservation
    carry over as follows.
    \begin{equation*}
        J_\pm^g\geq 0,\qquad
        \tr_{S'}[J_\pm^g]
        = \varrho_{\mathrm{in}}(g)\,\tr_{S'}[J_\pm]\,\varrho_{\mathrm{in}}(g)^\dagger
        = p_\pm\,I_{SP},
    \end{equation*}
    where $\varrho_{\mathrm{in}}(g):=U_g\ox U_g^*\ox V_g$
    is the restriction of $\varrho_G(g)$ to $\cH_S\ox\cH_P$, and the second
    equality uses unitarity together with
    $\tr_{S'}[(\cI_{SP}\ox (V_g^*)_{S'})X(\cI_{SP}\ox (V_g^T)_{S'})]
    =\tr_{S'}[X]$ for all $X$.

    Set $J_\cP^g:=J_+^g-J_-^g=\varrho_G(g)J_\cP\varrho_G(g)^\dagger$.
    For any $\cE\in\cS$, write
    $\varrho_G(g)=W_{SS'}\ox R_P$ with $W_{SS'}:=(U_g)_S\ox(V_g^*)_{S'}$
    and $R_P:=(U_g^*)_{P_1}\ox(V_g)_{P_2}$. Pulling
    $W_{SS'}$ through the partial trace over $P$ and using the
    cyclic property on $P$ yields
    \begin{equation*}
        \tr_P\!\bigl[J_\cP^g\,(I_S\ox J_\cE^T\ox I_{S'})\bigr]
        \;=\;
        W_{SS'}\;\tr_P\!\bigl[J_\cP\,(I_S\ox R_P^\dagger
        J_\cE^T R_P\ox I_{S'})\bigr]\;W_{SS'}^\dagger.
    \end{equation*}
    Define $\hat\cE$ by
    $J_{\hat\cE}:=(U_g^\dagger\ox V_g^T)\,J_\cE\,(U_g\ox V_g^*)$
    on $P_1\ox P_2$. Transposing confirms
    $J_{\hat\cE}^T=R_P^\dagger J_\cE^T R_P$. By the Choi
    transformation rule,
    $\hat\cE=\Op{Ad}_{V_g^T}\circ\cE\circ\Op{Ad}_{U_g^*}$.
    Self-conjugacy of the representations guarantees
    $\hat\cE\in\cS$ because for every $g$ there exists
    $h\in G$ with $U_h=U_g^T$ and $V_h=V_g^T$
    (up to phases absorbed by $\Op{Ad}$), so
    $\hat\cE=\cV_h\circ\cE\circ\cU_h^\dagger\in\cS$,
    since $\cU_h^\dagger=\Op{Ad}_{U_h^\dagger}
    =\Op{Ad}_{(U_g^T)^\dagger}=\Op{Ad}_{U_g^*}$.
    The feasibility of $J_\cP$ for $\hat\cE$ then gives
    \begin{equation*}
        \tr_P\!\bigl[J_\cP^g\,(I_S\ox J_\cE^T\ox I_{S'})\bigr]
        \;=\;
        d\;W_{SS'}\,(J_{\hat\cE})_{SS'}\,W_{SS'}^\dagger.
    \end{equation*}
    The outer conjugation collapses to the following expression.
    $W_{SS'}\,(J_{\hat\cE})_{SS'}\,W_{SS'}^\dagger
    =(U_gU_g^\dagger)_S\ox(V_g^*V_g^T)_{S'}\;J_\cE\;
    (U_gU_g^\dagger)_S\ox(V_g^*V_g^T)_{S'}=J_\cE$,
    since $U_gU_g^\dagger=I$ and
    $V_g^*V_g^T=V_g^*(V_g^*)^\dagger=I$.
    Hence $J_\cP^g$ is feasible for $\cS$ with overhead
    $p_++p_-=\nu_1(\cS)$.

    Averaging over $G$ gives the Haar-symmetrized operators
    \begin{equation*}
        \bar{J}_\pm
        \;:=\;\int_G J_\pm^g\,d\mu(g)
        \;\geq 0,\qquad
        \bar{J}_\cP:=\bar{J}_+-\bar{J}_-,
    \end{equation*}
    with $\tr_{S'}[\bar{J}_\pm]=p_\pm I_{SP}$ (linearity of
    integration and the trace-preservation identity above) and
    $[\varrho_G(g),\bar{J}_\pm]=0$ for all $g$ (by construction).
    The feasibility constraint is linear in $J_\cP$ and holds for every
    $J_\cP^g$. It therefore survives Haar averaging, making $\bar{J}_\cP$
    feasible for $\cS$ with $\bar{p}_++\bar{p}_-=p_++p_-=\nu_1(\cS)$.
    Setting $J_{\pm,*}:=\bar{J}_\pm$ proves~(i).

    The Hermitian part $\mathfrak C_G^{\rm sa}$ has real dimension
    $n=\dim_{\CC}\mathfrak C_G$. Choose a real orthonormal Hermitian
    basis $(e_1,\dots,e_n)$ of $\mathfrak C_G^{\rm sa}$ and write
    $J_{\pm,*}=\sum_{\alpha=1}^n x_{\pm,\alpha}e_\alpha$ with
    $x_{\pm,\alpha}\in\RR$. This proves~(ii), while retaining the two
    positive variables required by the SDP.
\end{proof}

The closed-form appendices below apply this symmetrization directly in
each channel family, keeping the representation-theoretic reduction local
to the corresponding proof.


\section{One-copy programming protocol for all quantum channels}\label{sec:cptp_appendix}

This appendix proves the all-channel programming theorem. The general
one-copy SDP is Eq.~\eqref{eq:onecopy_SDP}. In the present case, the
$\Op{SU}(d)\times\Op{SU}(d)$ symmetry fixes the feasible retriever
essentially uniquely. The invariance needed for the covariance reduction
is the following elementary Choi-state fact.

\begin{lemma}\label{lem:unitary_invariant_in_cptp}
    Let $\cE\in\Op{CPTP}(\cH_{S} \rightarrow \cH_{S'})$ with $\cH_{S}\simeq \cH_{S'}$ and $\dim[\cH_{S}] = d$. Then for any $U,V\in \Op{SU}(d)$, the twisted Choi operator $(U\ox V)J_{\cE}(U\ox V)^{\dagger}$ is the Choi operator of some channel in $\Op{CPTP}(\cH_{S} \rightarrow \cH_{S'})$.
\end{lemma}
\begin{proof}
    Unitary conjugation preserves positivity, so $J_{\cE}\ge 0$ yields $(U\ox V)J_{\cE}(U\ox V)^{\dagger}\ge 0$. Trace preservation follows directly from
    \begin{equation*}
        \tr_{S'}\bigl((U\ox V)J_{\cE}(U\ox V)^{\dagger}\bigr) = U\tr_{S'}(J_{\cE})U^{\dagger} = U I_S U^{\dagger} = I_S.
    \end{equation*}
\end{proof}

\begin{theorem}\label{thm:hptp_protocol_channels}
    Let $\cS = \Op{CPTP}(\cH_d \rightarrow \cH_d)$ and set $\pi_{\cE} = J_{\cE} / d$. The HPTP map
    \begin{equation*}
        \cP(X) = \left(\frac{\tr X}{d} - \tr(\cM(X)) \right) I_{d} + d\,\cM(X),
        \qquad X\in \cB(\cH_{S} \ox \cH_P),
    \end{equation*}
    is an exact universal retriever and satisfies $\cP(\rho_S\ox\pi_{\cE})=\cE(\rho_S)$ for every $\cE\in\cS$. Here $\cM(X):=\tr_{SP_1}((\Omega_{SP_1} \ox I_{P_2})X)$ and $\Omega_{SP_1}$ is the unnormalized maximally entangled projector on $\cH_{SP_1}$. Moreover, $\cP$ is optimal and
    \begin{equation*}
        \nu_1(\cS) = 2d^2 - 3 + \tfrac{2}{d^2}.
    \end{equation*}
\end{theorem}
\begin{proof}
    Notice that Lemma~\ref{lem:unitary_invariant_in_cptp} and
    Theorem~\ref{thm:G_symmetrisation}, applied with
    $G=\Op{SU}(d)\times\Op{SU}(d)$, allow every feasible
    quasi-decomposition in Eq.~\eqref{eq:onecopy_SDP} to be Haar-averaged
    into the commutant of $U\ox U^*\ox V\ox V^*$. Since
    $U\ox U^*$ decomposes into the trivial and adjoint sectors,
    Schur's lemma gives
    \begin{equation*}
        J_{\cP} = \a\, I_{d^2}\ox I_{d^2} + \b\, I_{d^2}\ox \Omega_d + \eta\, \Omega_d\ox I_{d^2} + \delta\, \Omega_d\ox \Omega_d,
    \end{equation*}
    for scalars $\a,\b,\eta,\delta$, where $\Omega_d=\dketbra{I_d}{I_d}$.

    To prove the uniqueness, we re-derive the linear constraints using the above expression.
    Trace preservation $\tr_{S'}J_{\cP}=I_{SP}$ is equivalent to
    \begin{equation}\label{eq:trace_preserving_condition_from_decomp}
        (d\alpha + \beta) I_{d^2}\ox I_d + (d\eta + \delta)\Omega_d \ox I_d = I_{d^2}\ox I_d \;\Rightarrow\;
        \begin{cases}
            d\alpha + \beta = 1,\\
            d\eta + \delta = 0.
        \end{cases}
    \end{equation}
    For every quantum channel $\cE$, the four basis terms satisfy
    \begin{equation*}
    \begin{cases}
        \tr_P[(I_{d^2}\ox I_{d^2})(I_S\ox J_{\cE}^T \ox I_{S'})] = d\,I_{SS'},\\
        \tr_P[(I_{d^2}\ox \Omega_d)(I_S\ox J_{\cE}^T \ox I_{S'})] = I_{S}\ox \tr_{P_1}(J_{\cE}),\\
        \tr_P[(\Omega_d\ox I_{d^2})(I_S\ox J_{\cE}^T \ox I_{S'})] = I_{SS'},\\
        \tr_P[(\Omega_d\ox \Omega_d)(I_S\ox J_{\cE}^T \ox I_{S'})] = J_{\cE}.
    \end{cases}
    \end{equation*}
    Substitution into the exact-programming constraint gives
    \begin{equation*}
        (d\alpha + \eta) I_{SS'} + \beta\, I_S \ox \tr_{P_1}(J_{\cE}) + \delta\, J_{\cE} = d\,J_{\cE},\qquad \forall\,\cE\in\Op{CPTP}_d.
    \end{equation*}
    Restricting this identity to unital channels, for which
    $\tr_{P_1}(J_{\cE})=I_{S'}$, and using the fact that their Choi
    operators are not confined to the identity direction, yields
    \begin{equation*}
        \delta=d,\qquad d\alpha+\beta+\eta=0.
    \end{equation*}
    Combining with~\eqref{eq:trace_preserving_condition_from_decomp}
    gives $\eta=-1$ and $\delta=d$. The remaining condition is
    \begin{equation*}
        \beta\bigl(I_S\ox\tr_{P_1}(J_{\cE})-I_{SS'}\bigr)=0,
        \qquad \forall\,\cE\in\Op{CPTP}_d.
    \end{equation*}
    Since the all-channel family contains non-unital channels,
    $\beta=0$, and therefore $\alpha=1/d$. Thus the commutant
    constraints have a unique solution,
    \begin{equation*}
        J_{\cP} = \tfrac{1}{d}\,I_{d^2} \ox I_{d^2} - \Omega_d\ox I_{d^2} + d\,\Omega_d \ox \Omega_d.
    \end{equation*}

    Now we can derive the retriever action acting on the principal system.
    Inserting this $J_\cP$ into
    $\cP(X)=\tr_{SP}[J_\cP(X^T\ox I_{S'})]$ gives, for every
    $X\in\cB(\cH_S\ox\cH_P)$,
    \begin{equation}\label{eq:map_form_retrieval_arbitrary_channel}
        \cP(X) = \Bigl(\tfrac{\tr X}{d} - \tr(\cM(X))\Bigr) I_{d} + d\,\cM(X),
    \end{equation}
    where
    $\cM(X):=\tr_{SP_1}[(\Omega_{SP_1}\ox I_{P_2})X]$.
    This expression is linear. It is Hermiticity preserving because
    partial-trace cyclicity on $SP_1$ gives
    $\cM(X)^\dagger=\cM(X)$ whenever $X=X^\dagger$, and it is trace
    preserving since
    \begin{equation*}
        \tr[\cP(X)]
        =d\left(\frac{\tr X}{d}-\tr[\cM(X)]\right)
        +d\tr[\cM(X)]
        =\tr X.
    \end{equation*}
    For a density operator, $\tr X=1$. The defining constraint in
    Eq.~\eqref{eq:onecopy_SDP} then implies
    $\cP(\rho_S\ox J_{\cE}/d)=\cE(\rho_S)$ for every quantum channel
    $\cE$.

    \emph{Lower bound.}
    Evaluate the diamond norm on
    $\sigma=\Omega_{SP_1}/d\ox\Omega_{P_2E}/d$. Then
    \begin{equation*}
        (\cP\ox\cI_E)(\sigma) = \Bigl(\tfrac{1}{d}-d\Bigr)\frac{I_{d^2}}{d} + d\,\Omega_{P_2E},
    \end{equation*}
    and hence
    \begin{equation*}
        \|\cP\|_\diamond
        \ge
        \left|\tfrac{1}{d}\bigl(\tfrac{1}{d}-d\bigr)+d^2\right|
        +(d^2-1)\left|\tfrac{1}{d}\bigl(\tfrac{1}{d}-d\bigr)\right|
        =
        2d^2-3+\tfrac{2}{d^2}.
    \end{equation*}

    \emph{Upper bound and optimality.}
    Theorem~3 of Ref.~\cite{Regula2021operational} gives, for any
    trace-preserving Hermiticity-preserving map,
    \begin{equation}
        \|\cP\|_\diamond = \min\bigl\{\mu_+ + \mu_- \;\big|\; J_{\cP}=\mu_+J_{\cQ_+}-\mu_-J_{\cQ_-},\;
        \cQ_\pm\in\Op{CPTP}(\cH_S\ox\cH_P\to\cH_{S'}),\; \mu_\pm\ge 0\bigr\}.
    \end{equation}
    Set
    \begin{equation*}
        \mu_+=d^2-1+\frac{1}{d^2},\qquad
        \mu_-=d^2-2+\frac{1}{d^2},
    \end{equation*}
    and
    \begin{equation}
        \begin{cases}
            J_{\cQ_+} = \tfrac{1}{d} I_{d^2}\ox I_{d^2} - \tfrac{1}{d^2}\Omega_d\ox I_{d^2} + \tfrac{1}{d}\Omega_d\ox\Omega_d,\\[2pt]
            J_{\cQ_-} = \tfrac{1}{d} I_{d^2}\ox I_{d^2} + \tfrac{1}{d^2(d^2-1)}\Omega_d\ox I_{d^2} - \tfrac{1}{d(d^2-1)}\Omega_d\ox\Omega_d.
        \end{cases}
    \end{equation}
    Direct substitution gives
    $J_\cP=\mu_+J_{\cQ_+}-\mu_-J_{\cQ_-}$. Let
    $P_0:=\Omega_d/d$ and $P_1:=I_{d^2}-P_0$. In the sector basis
    $\{P_i\ox P_j\}_{i,j=0,1}$, the nonzero coefficients of
    $J_{\cQ_+}$ are
    \begin{equation*}
        d\;P_0\ox P_0+\frac{1}{d}P_1\ox P_0+\frac{1}{d}P_1\ox P_1,
    \end{equation*}
    and those of $J_{\cQ_-}$ are
    \begin{equation*}
        \frac{d}{d^2-1}P_0\ox P_1
        +\frac{1}{d}P_1\ox P_0
        +\frac{1}{d}P_1\ox P_1.
    \end{equation*}
    Thus $J_{\cQ_\pm}\ge0$. Their partial traces satisfy
    $\tr_{S'}J_{\cQ_\pm}=I_{SP}$, so
    $\cQ_\pm\in\Op{CPTP}(\cH_S\ox\cH_P\to\cH_{S'})$. Therefore
    \begin{equation*}
        \|\cP\|_\diamond \le \mu_++\mu_- = 2d^2-3+\tfrac{2}{d^2}.
    \end{equation*}
    The lower bound matches this quasi-decomposition, so
    $\|\cP\|_\diamond=2d^2-3+2/d^2$. Finally, take any feasible
    quasi-decomposition of a retriever for all quantum channels. Haar averaging as
    above produces a commutant-feasible quasi-decomposition with the same
    overhead, and the commutant constraints force its retriever Choi
    operator to be the unique $J_\cP$ derived above. Hence no feasible
    exact programming protocol can have overhead below
    $\|\cP\|_\diamond$, proving optimality and the stated value of
    $\nu_1(\cS)$.
\end{proof}

\section{One-copy programming protocols for unitary channels}\label{sec:unitary_appendix}

Programming a symmetry-group action is a central instance of
programmable quantum computing. This appendix treats
$\cS=\Op{Ad}_{\{U\}}$, where $\Op{Ad}$ denotes adjoint action by a
unitary set. The exact one-copy overhead
$\nu_1(\Op{Ad}_{\Op{SU}(d)})=d^2-1$ follows by reducing the
symmetrized programming SDP to a finite linear program.

Let $\cS=\Op{SU}(d)$ and let $d\ge2$. After
$G$-symmetrization, with $G=\Op{SU}(d)\times\Op{SU}(d)$, both
positive parts $J_+$ and $J_-$ may be chosen in the commutant of
$U\otimes U^*\otimes V\otimes V^*$. The performance constraints for
unitary programming are expressed through
\begin{equation}\label{sdp:primal_prob_vir_sym_J1J2}
\begin{aligned}
&\underline{\textbf{Primal Program}}\\
    \nu_1(\Op{Ad}_{\cS})=\min\;&p_++p_-\\
    {\rm s.t.}\;\;& J_{\cP}:=J_+-J_-,\\
    &\tr[J_{\cP}W_1]=d, \; \tr[J_{\cP}W_2]=d,\\
    &[J_{\pm},U \otimes U^*\otimes V\otimes V^*]=0,\,\forall U,V\in\Op{SU}(d),\\
    &J_+\geq 0,\,\tr_{S'}[J_+] =p_+ I_{SP},\\
    &J_-\geq 0,\,\tr_{S'}[J_-] =p_- I_{SP}.
\end{aligned}
\end{equation}
Here
\begin{equation*}
    W_1:=\int_{\Op{SU}(d)}
    \frac{\dketbra{U}{U}_{P}}{d}\otimes I_{SS'}\,dU,
    \qquad
    W_2:=\int_{\Op{SU}(d)}
    \frac{\dketbra{U\ox U^*}{U\ox U^*}_{SPS'}}{d^2}\,dU,
\end{equation*}
where $\dket{U}:=(I\ox U)\dket{I_d}$ is unnormalized.

\begin{lemma}\label{lem:LP_pga_SU_one_copy}
    For $d>1$ and $\cS = \Op{SU}(d)$, the minimal overhead for programming arbitrary unitary operations in $\Op{Ad}_{\Op{SU}(d)}$ is computed by the following linear program. In standard form, the primal LP is
    \[
    \begin{aligned}
    &\underline{\textbf{Linear Program}}\\
        \nu_1(\Op{Ad}_{\cS})=\min_{z\in\RR^8, z\geq 0}\;&c^T z \\
        {\rm s.t.}\;\;& Az = b,
    \end{aligned}
    \]
    and the dual LP reads,
    \[
    \begin{aligned}
    &\underline{\textbf{Dual Program}}\\
        \nu_1(\Op{Ad}_{\cS})=\max_{\lambda\in\RR^4}\;& -b^T \lambda \\
    {\rm s.t.}\;\;& A^T\lambda + c \geq 0.
    \end{aligned}
    \]
    The fixed constraint matrix $A$ and vectors $b$ and $c$ are
    \[
    A =
    \begin{pmatrix}
        \tfrac{1}{d^2} & \tfrac{d^2-1}{d^2} & \tfrac{d^2-1}{d^2} & \tfrac{(d^2-1)^2}{d^2} & -\tfrac{1}{d^2} & -\tfrac{d^2-1}{d^2} & -\tfrac{d^2-1}{d^2} & -\tfrac{(d^2-1)^2}{d^2}\\[2pt]
        \tfrac{1}{d^2} & 0 & 0 & \tfrac{d^2-1}{d^2} & -\tfrac{1}{d^2} & 0 & 0 & -\tfrac{d^2-1}{d^2}\\[2pt]
        \tfrac{1}{d} & \tfrac{d^2-1}{d} & -\tfrac{1}{d} & -\tfrac{d^2-1}{d} & 0 & 0 & 0 & 0\\[2pt]
        0 & 0 & 0 & 0 & \tfrac{1}{d} & \tfrac{d^2-1}{d} & -\tfrac{1}{d} & -\tfrac{d^2-1}{d}
    \end{pmatrix}, \;
    b =
    \begin{pmatrix}
        d \\
        d \\
        0 \\
        0
    \end{pmatrix}, \;
    c =
    \begin{pmatrix}
        0\\
        0\\
        1/d\\
        (d^2-1)/d\\
        0\\
        0\\
        1/d\\
        (d^2-1)/d
    \end{pmatrix}.
    \]
\end{lemma}
\begin{proof}
    Put $Q_d:=I_{d^2}-\Phi_d$. With the normalizations in
    \eqref{sdp:primal_prob_vir_sym_J1J2}, Haar integration gives
    \begin{equation*}
        W_1=\frac{I_{SPS'}}{d^2},\qquad
        W_2=\frac{Q_d\ox Q_d}{d^2(d^2-1)}+
        \frac{\Phi_d\ox\Phi_d}{d^2},
    \end{equation*}
    where the first tensor factor is on $SP_1$ and the second one on
    $P_2S'$. The first identity follows from
    $\int\dketbra{U}{U}\,dU=I_P/d$. For the second, regroup the four
    registers as $(SP_1):(P_2S')$ and write $R(U):=U\ox U^*$. With
    the convention $\dket{A}=(I\ox A)\dket{I}$, its vectorization in
    the physical order $[S,P_1,P_2,S']$ is
    \begin{equation*}
        \dket{R(U)}
        =\sum_{i,j,k,\ell}U_{ki}\,\overline{U_{\ell j}}\,
        \ket{i}_S\ket{j}_{P_1}\ket{k}_{P_2}\ket{\ell}_{S'}.
    \end{equation*}
    Thus this vector is precisely the regrouped
    $\dket{U\ox U^*}$ appearing in $W_2$. The representation is
    multiplicity-free and decomposes as
    \begin{equation*}
        R\simeq \mathbf{1}\oplus\mathrm{Ad},\qquad
        P_0=\Phi_d,\quad P_1=Q_d,\qquad
        m_0=1,\quad m_1=d^2-1,
    \end{equation*}
    where $P_\mu$ projects onto the corresponding irreducible sector
    and $m_\mu=\tr P_\mu$. In orthonormal bases adapted to these two
    sectors, Schur orthogonality reads
    \begin{equation*}
        \int_{\Op{SU}(d)}
        [R_\mu(U)]_{ab}\,\overline{[R_\nu(U)]_{cd}}\,dU
        =\delta_{\mu\nu}\frac{\delta_{ac}\delta_{bd}}{m_\mu},
        \qquad \mu,\nu\in\{0,1\}.
    \end{equation*}
    Vectorizing this identity therefore gives
    \begin{equation*}
        \int_{\Op{SU}(d)}\dketbra{R(U)}{R(U)}\,dU
        =(P_0^T)_{SP_1}\ox(P_0)_{P_2S'}
        +\frac{(P_1^T)_{SP_1}\ox(P_1)_{P_2S'}}{d^2-1}.
    \end{equation*}
    The transpose on the input projector comes from the vectorization
    convention. In the computational basis, $P_0=\Phi_d$ and
    $P_1=Q_d$ are real symmetric, so it may be dropped. Dividing by the
    factor $d^2$ in the definition of $W_2$ yields the displayed formula.

    The commutant condition restricts the two positive parts to
    \begin{equation*}
    \begin{aligned}
        J_+&=v_1I\ox I+v_2I\ox\Phi_d+v_3\Phi_d\ox I+v_4\Phi_d\ox\Phi_d,\\
        J_-&=w_1I\ox I+w_2I\ox\Phi_d+w_3\Phi_d\ox I+w_4\Phi_d\ox\Phi_d.
    \end{aligned}
    \end{equation*}
    In the sector order
    $\Phi_d\ox\Phi_d,\Phi_d\ox Q_d,Q_d\ox\Phi_d,Q_d\ox Q_d$,
    positivity is equivalent to $\mathbb{S}v\ge0$ and
    $\mathbb{S}w\ge0$. The two contractions with $W_1,W_2$ give
    $c_1^T(v-w)=d$ and $c_2^T(v-w)=d$, the partial-trace condition for
    $J_\pm$ gives the objective $c_3^T(v+w)$, and vanishing of the
    $\Phi_d\ox I$ partial-trace component gives $c_4^Tv=c_4^Tw=0$.
    Hence the SDP reduces to
\begin{equation}
\begin{aligned}
&\underline{\textbf{Linear Program}}\\
    \nu_1(\Op{Ad}_{\Op{SU}(d)})=\min_{v,w\in\RR^4}\;&c_3^T (v+w)\\
    {\rm s.t.}\;\;& c_1^T(v - w) = d,\\
    &c_2^T (v - w) = d,\\
    & \mathbb{S}v \geq 0, \, \mathbb{S}w \geq 0,\\
    & c_4^T v = 0, \; c_4^T w = 0
\end{aligned}
\end{equation}
where the fixed vectors $\bm{c}_{1,2,3,4}$ and matrix $\mathbb{S}$ are
\begin{equation}
\begin{aligned}
    c_1 &= (d^2, 1, 1, 1/d^2)^T, \\ c_2 &= (1, 1/d^2, 1/d^2, 1/d^2)^T,\\
    c_3 &= (d, 1/d, 0, 0)^T, \\
    c_4 &= (0, 0, d, 1/d)^T, \\
    \mathbb{S} &=
    \begin{pmatrix}
        1 & 1 & 1 & 1\\
        1 & 0 & 1 & 0\\
        1 & 1 & 0 & 0\\
        1 & 0 & 0 & 0
    \end{pmatrix}.
\end{aligned}
\end{equation}
To see that these scalar constraints are exactly the unitary-programming
constraints, set $q:=v-w$. They imply
$c_4^Tq=0$, $c_1^Tq=d$, and $c_2^Tq=d$, hence
\begin{equation*}
    q_3=-d,\qquad q_4=d^3,\qquad d^2q_1+q_2=d,\qquad
    dq_1+\frac{q_2+q_3}{d}=0.
\end{equation*}
For every unitary $U$, a direct contraction of the Choi action gives
\begin{equation*}
    \tr_P[J_\cP(I_S\ox J_U^T\ox I_{S'})]
    =\Bigl(dq_1+\frac{q_2+q_3}{d}\Bigr)I_{SS'}
    +\frac{q_4}{d^2}J_U
    =dJ_U,
\end{equation*}
which is precisely exact programming on $\Op{Ad}_{\Op{SU}(d)}$.

Introduce the change of variables $x := \mathbb{S}v$,
$y := \mathbb{S}w$, and $z := (x,y)^T \in \RR^8$. Since
$c_j^T v = c_j^T \mathbb{S}^{-1} x = \tilde c_j^T x$ with
$\tilde c_j := \mathbb{S}^{-T} c_j$, the linear program takes the
\textit{standard form}
\begin{equation}
\begin{aligned}
&\underline{\textbf{Linear Program}}\\
    \nu_1(\Op{Ad}_{\Op{SU}(d)})=\min_{z\in\RR^8}\;&c^T z \\
    {\rm s.t.}\;\;& Az = b,\\
    & z \geq 0,
\end{aligned}
\end{equation}
where
\begin{equation}
    A =
    \begin{pmatrix}
        \tilde c_1^T & -\tilde c_1^T\\
        \tilde c_2^T & -\tilde c_2^T\\
        \tilde c_4^T & 0\\
        0 & \tilde c_4^T
    \end{pmatrix}, \;
    b =
    \begin{pmatrix}
        d \\
        d \\
        0 \\
        0
    \end{pmatrix}, \;
    c =
    \begin{pmatrix}
        \tilde c_3\\
        \tilde c_3
    \end{pmatrix}.
\end{equation}
Because $\mathbb{S}$ is symmetric, $\tilde c_j = \mathbb{S}^{-1} c_j$, and direct computation yields
\begin{equation}
\begin{aligned}
    \tilde c_1 &= \bigl(\tfrac{1}{d^2},\, \tfrac{d^2-1}{d^2},\, \tfrac{d^2-1}{d^2},\, \tfrac{(d^2-1)^2}{d^2}\bigr)^T, \\
    \tilde c_2 &= \bigl(\tfrac{1}{d^2},\, 0,\, 0,\, \tfrac{d^2-1}{d^2}\bigr)^T,\\
    \tilde c_3 &= \bigl(0,\, 0,\, \tfrac{1}{d},\, \tfrac{d^2-1}{d}\bigr)^T, \\
    \tilde c_4 &= \bigl(\tfrac{1}{d},\, \tfrac{d^2-1}{d},\, -\tfrac{1}{d},\, -\tfrac{d^2-1}{d}\bigr)^T,
\end{aligned}
\end{equation}
yielding the displayed $A$, $b$, $c$.
\end{proof}

\begin{theorem}\label{thm:unitary_appendix}
    Let $d\ge2$ and $\cS:=\Op{SU}(d)$. Then
    $\nu_1(\Op{Ad}_{\cS})=d^2-1$.
\end{theorem}
\begin{proof}
    Lemma~\ref{lem:LP_pga_SU_one_copy} gives a primal LP and its dual.
    For any primal-feasible $z$ and dual-feasible $\lambda$, weak duality
    gives $-b^T\lambda\le c^Tz$. It is therefore enough to exhibit a
    feasible pair with common value $d^2-1$.

    The dual program reads
    \begin{equation}
        \max_{\lambda\in\RR^4}\; -b^T\lambda \quad \text{s.t. }\; A^T\lambda + c \geq 0.
    \end{equation}
    \emph{Dual feasibility (lower bound).} Set $\lambda' = \bigl(\tfrac{1}{d},\; -d,\; \tfrac{d^2-1}{d^2},\; \tfrac{1}{d^2}\bigr)^T$. For this choice,
    \begin{equation*}
        A^T\lambda' + c = \Bigl(0,\; \tfrac{d^2-1}{d},\; \tfrac{1}{d},\; 0,\; \tfrac{1}{d},\; 0,\; 0,\; \tfrac{d^2-1}{d}\Bigr)^T \geq 0,
    \end{equation*}
    confirming dual feasibility, and the objective
    $-b^T\lambda' = d^2-1$ gives a lower bound on the primal optimum.

    \emph{Primal feasibility (upper bound).} Set
    \begin{equation*}
        z' = \Bigl(\tfrac{d^3}{2},\; 0,\; 0,\; \tfrac{d^3}{2(d^2-1)},\; 0,\; \tfrac{d(d^2-2)}{2(d^2-1)},\; \tfrac{d(d^2-2)}{2},\; 0\Bigr)^T \geq 0.
    \end{equation*}
    Substitution gives $Az' = (d,d,0,0)^T = b$, so $z'$ is primal-feasible, with objective $c^Tz' = d^2-1$.

    The primal certificate has the same value, $c^Tz'=d^2-1$. Hence the
    primal and dual bounds coincide, proving
    $\nu_1(\Op{Ad}_{\cS}) = d^2-1$.
\end{proof}

\section{One-copy programming protocols for real channels}\label{sec:real_channel}

Write $\cL^\mathbb{R}$ for the set of all real channels, i.e., channels
whose Choi matrices are real in the fixed computational basis. The proof
uses orthogonal covariance under $\Op{O}(d)\times\Op{O}(d)$ for all
$d\ge2$. The corresponding two-copy commutant is larger than the unitary
commutant and is generated, on each paired tensor factor, by
$I$, the swap $\mathbb{F}$, and the unnormalized maximally entangled
operator $\Omega_d=\dketbra{I_d}{I_d}$.

\begin{lemma}\label{lem:O_invariant_in_real_channels}
    Let $\cE\in\cL^\mathbb{R}(\cH_{S} \rightarrow \cH_{S'})$ with $\cH_{S}\simeq \cH_{S'}$ and $\dim[\cH_{S}] = d$. Then for any $U,V\in \Op{O}(d)$, the twisted Choi operator $(U\ox V)J_{\cE}(U\ox V)^{T}$ is the Choi matrix of some channel in $\cL^\mathbb{R}(\cH_{S}\rightarrow \cH_{S'})$.
\end{lemma}
\begin{proof}
    Positivity is preserved because $J_\cE\ge 0$ forces $(U\ox V)J_\cE(U\ox V)^T\ge 0$. The entries remain real because $U,V$ are real orthogonal matrices. Trace preservation follows from
    \begin{equation*}
        \tr_{S'}\bigl((U\ox V)J_{\cE}(U\ox V)^{T}\bigr) = U\tr_{S'}(J_{\cE})U^{T} = UU^{T} = I_S.
    \end{equation*}
\end{proof}

By the $G$-symmetrization of Theorem~\ref{thm:G_symmetrisation}, applied
with $G=\Op{O}(d)\times\Op{O}(d)$ through
Lemma~\ref{lem:O_invariant_in_real_channels}, an optimal $J_{\cP}$ may be chosen
to commute with every $U\ox U\ox V\ox V$ for
$U,V\in\Op{O}(d)$. The Brauer commutant of the defining orthogonal
representation on two copies is $\Op{span}\{I,\mathbb{F},\Omega_d\}$.
Assigning the first factor to $SP_1$ and the second to $P_2S'$, we obtain
the following form for the retriever Choi operator.
\begin{equation*}
    J_{\cP} = \sum_{a,b\in\{I,\mathbb{F},\Omega_d\}} \a_{ab}\, a\ox b
    = \begin{aligned}[t]
        &\a_{11}\,I\ox I + \a_{12}\,I\ox \mathbb{F} + \a_{13}\,I\ox\Omega_d\\
        +\;&\a_{21}\,\mathbb{F}\ox I + \a_{22}\,\mathbb{F}\ox \mathbb{F} + \a_{23}\,\mathbb{F}\ox\Omega_d\\
        +\;&\a_{31}\,\Omega_d\ox I + \a_{32}\,\Omega_d\ox \mathbb{F} + \a_{33}\,\Omega_d\ox\Omega_d,
    \end{aligned}
\end{equation*}
for real coefficients $\a_{ij}$. Trace preservation gives, by comparing
the coefficients of $I\ox I$, $\mathbb{F}\ox I$, and $\Omega_d\ox I$
after tracing out $S'$,
\begin{equation}\label{eq:trace_preserving_condition_from_decomp_real}
\begin{cases}
    d\a_{11} + \a_{12}+\a_{13} = 1,\\
    d\a_{21} + \a_{22}+\a_{23} = 0, \\
    d\a_{31} + \a_{32}+\a_{33} = 0.
\end{cases}
\end{equation}

The feasibility constraint is next imposed on $J_{\cP}$. For $\cE\in\cL^\mathbb{R}$, the Choi matrix $J_\cE$ is Hermitian and real, hence real symmetric. In particular, $J_\cE^T = J_\cE$, and $\tr_{P_1} J_\cE$ is again real symmetric. These facts are used freely in what follows. Partial traces of $a\ox b$ against $I_S\ox J_\cE^T\ox I_{S'}$ collapse to five distinct combinations,
\begin{equation*}
\begin{cases}
    \tr_P[(I\ox I)(I_S\ox J_{\cE}^T\ox I_{S'})] = dI_{SS'},\\
    \tr_P[(I\ox \Omega_d)(I_S\ox J_{\cE}^T\ox I_{S'})] = \tr_P[(I\ox \mathbb{F})(I_S\ox J_{\cE}^T\ox I_{S'})] = I_S\ox\tr_{P_1} J_\cE,\\
    \tr_P[(\Omega_d\ox I)(I_S\ox J_{\cE}^T\ox I_{S'})] = \tr_P[(\mathbb{F}\ox I)(I_S\ox J_{\cE}^T\ox I_{S'})] = I_{SS'},\\
    \tr_P[(\Omega_d\ox\Omega_d)(I_S\ox J_{\cE}^T\ox I_{S'})] = \tr_P[(\mathbb{F}\ox \mathbb{F})(I_S\ox J_{\cE}^T\ox I_{S'})] = J_\cE,\\
    \tr_P[(\mathbb{F}\ox\Omega_d)(I_S\ox J_{\cE}^T\ox I_{S'})] = \tr_P[(\Omega_d\ox \mathbb{F})(I_S\ox J_{\cE}^T\ox I_{S'})] = J_\cE^{T_S}.
\end{cases}
\end{equation*}
and substitution into the feasibility constraint of~\eqref{eq:onecopy_SDP} gives
\begin{equation*}
    (d\a_{11} + \a_{21}+\a_{31}) I_{SS'} + (\a_{12}+\a_{13})I_S \ox \tr_{P_1}(J_{\cE}) + (\a_{22}+\a_{33}) J_{\cE}+(\a_{23}+\a_{32})J_{\cE}^{T_S} = dJ_{\cE},
\end{equation*}
for all $\cE\in\cL^\mathbb{R}$.

Testing this identity on the subclass of real unital channels, where $\tr_{P_1} J_\cE = I_{S'}$, simplifies it to
\begin{equation*}
    (d\a_{11}+\a_{21}+\a_{31}+\a_{12}+\a_{13})I_{SS'} + (\a_{22}+\a_{33}) J_\cE + (\a_{23}+\a_{32}) J_\cE^{T_S} = dJ_\cE.
\end{equation*}
The coefficients of $J_\cE$ and $J_\cE^{T_S}$ can be separated by
local perturbations around the completely depolarizing real unital
channel $J_0=I_{P_1P_2}/d$. If $X$ is real symmetric and satisfies
$\tr_{P_1}X=\tr_{P_2}X=0$, then $J_0+\epsilon X$ is the Choi matrix of a
real unital channel for sufficiently small real $\epsilon$. Taking
$X_+=R\ox R$ with $R$ real symmetric traceless gives
$X_+^{T_{P_1}}=X_+$, whereas $X_-=K\ox K$ with $K$ real antisymmetric
gives $X_-^{T_{P_1}}=-X_-$. The linear term in $\epsilon$ therefore
forces
\begin{equation*}
\begin{cases}
    d\a_{11}+\a_{21}+\a_{31}+\a_{12}+\a_{13} = 0,\\
    \a_{22}+\a_{33} = d,\\
    \a_{23}+\a_{32} = 0.
\end{cases}
\end{equation*}
where the scalar equation follows by substituting $J_0$. Combining this
with~\eqref{eq:trace_preserving_condition_from_decomp_real} yields
$\a_{21}+\a_{31}=-1$.

It remains to remove the coefficient multiplying
$I_S\ox\tr_{P_1}J_\cE$. The reset channel
$J_\cE=I_{P_1}\ox\ketbra{0}{0}_{P_2}$ is real and satisfies
$\tr_{P_1}J_\cE=d\ketbra{0}{0}\ne I_{S'}$. Hence the residual identity
forces $\a_{12}+\a_{13}=0$. Equation
\eqref{eq:trace_preserving_condition_from_decomp_real} then gives
$\a_{11}=1/d$.

All constraints assembled, three free parameters $\a_{12}, \a_{21}, \a_{22}$ remain, and
\begin{equation*}
\begin{aligned}
    J_{\cP} =\;& \tfrac{1}{d}\,I\ox I + \a_{12}\,I\ox \mathbb{F} - \a_{12}\,I\ox\Omega_d\\
    +\;& \a_{21}\,\mathbb{F}\ox I + \a_{22}\,\mathbb{F}\ox \mathbb{F} - (d\a_{21}+\a_{22})\,\mathbb{F}\ox\Omega_d\\
    -\;& (1+\a_{21})\,\Omega_d\ox I + (d\a_{21}+\a_{22})\,\Omega_d\ox \mathbb{F} + (d-\a_{22})\,\Omega_d\ox\Omega_d.
\end{aligned}
\end{equation*}
Thus the remaining optimization is $\min_{\a_{12},\a_{21},\a_{22}}\|\cP\|_\diamond$.

\begin{theorem}\label{thm:real_appendix}
    The optimal overhead for programming all real channels in dimension $d\ge 2$ is
    $\nu_1(\cL^\mathbb{R})=\frac{2d^4-d^2-2d+4}{3d^2}$.
\end{theorem}

\begin{proof}
    The three-parameter family $J_{\cP}(\a_{12},\a_{21},\a_{22})$ derived above is built from tensor products of the commuting operators $I$, $\mathbb{F}$ (SWAP) and $\Omega_d$ (the unnormalized maximally entangled operator). These operators share a common eigenbasis. The programming problem therefore decomposes along that basis, and the diamond norm reduces to three trace-norm contributions. The resulting optimization is a convex three-variable problem amenable to KKT analysis.

    We first diagonalize $I$, $\mathbb{F}$ and $\Omega_d$ on
    $\cH_d\ox\cH_d$. They commute pairwise and admit a common
    \emph{real} orthonormal eigenbasis. Grouping basis vectors by their
    joint eigenvalues produces three mutually orthogonal real subspaces.

    The first is the one-dimensional span of the maximally entangled state,
    \begin{equation*}
        E_0 := \Op{span}\{\dket{I_d}\},\qquad \dim E_0 = 1.
    \end{equation*}
    The second, the symmetric traceless subspace $E_1$, is spanned by the diagonal traceless vectors
    \begin{equation*}
        \ket{e_k} := \tfrac{1}{\sqrt{k(k-1)}}\Bigl(\textstyle\sum_{i=1}^{k-1}\ket{ii} - (k-1)\ket{kk}\Bigr),\qquad k=2,\dots,d,
    \end{equation*}
    together with the symmetrized off-diagonal states $\ket{e_{ij}^+} := \tfrac{1}{\sqrt{2}}(\ket{ij}+\ket{ji})$, where $1\le i<j\le d$. A direct count gives $\dim E_1 = (d+2)(d-1)/2$.
    The third subspace is the antisymmetric one,
    \begin{equation*}
        E_2 := \Op{span}\bigl\{\ket{e_{ij}^-} := \tfrac{1}{\sqrt{2}}(\ket{ij}-\ket{ji}) : 1\le i<j\le d\bigr\},\qquad \dim E_2 = d(d-1)/2.
    \end{equation*}
    These three subspaces together exhaust $\cH_d\ox\cH_d$, since $1 + (d+2)(d-1)/2 + d(d-1)/2 = d^2$.

    On each $E_\mu$, the operators $I$, $\mathbb{F}$ and $\Omega_d$ act as scalars, and the nine resulting eigenvalues appear in Table~\ref{tab:eigen}. They follow immediately from the definitions. The operator $I$ is the identity throughout. The operator $\mathbb{F}$ has eigenvalue $+1$ on symmetric states ($E_0\oplus E_1$) and $-1$ on antisymmetric states ($E_2$). The operator $\Omega_d=\dketbra{I_d}{I_d}$ projects onto $\dket{I_d}$, vanishes on its orthogonal complement, and has eigenvalue $d$ on $E_0$.

    \begin{table}[h]
    \centering
    \renewcommand{\arraystretch}{1.15}
    \begin{tabular}{lccc}
    \toprule
           & $E_0$ & $E_1$              & $E_2$            \\
           & (MES) & (sym.\ traceless)  & (antisymmetric)  \\
    \midrule
    $\dim$ & $1$   & $(d+2)(d-1)/2$     & $d(d-1)/2$       \\
    \midrule
    $I$    & $1$   & $1$                & $1$              \\
    $\mathbb{F}$ & $1$   & $1$                & $-1$             \\
    $\Omega_d$ & $d$   & $0$                & $0$              \\
    \bottomrule
    \end{tabular}
    \caption{Dimensions and joint eigenvalues of the identity $I$, the swap $\mathbb{F}$, and the unnormalized maximally entangled operator $\Omega_d=\dketbra{I_d}{I_d}$ on the three common eigenspaces of $\cH_d\ox\cH_d$.}
    \label{tab:eigen}
    \end{table}

    In this eigenbasis, the programming map $\cP$
    decomposes into a classical mixture of three reduced
    Hermitian-preserving maps, one for each eigenspace. These reduced
    maps need not be completely positive and are therefore not physical
    channels. Set $O_1,O_2,O_3 := I,\mathbb{F},\Omega_d$ and let $\lambda_{s\mu}$ denote the scalar by which $O_s$ acts on $E_\mu$, as read off from Table~\ref{tab:eigen}. Expand the Choi operator as $J_{\cP} = \sum_{s,t} r_{st}\,O_s\ox O_t$ and a generic input as $\rho_{SP_1P_2}=\sum_{ij}\alpha_{ij}\ketbra{e_i}{e_j}\ox\tau_{ij}$ in the joint eigenbasis $\{\ket{e_i}\}$ of the $SP_1$ system. The Choi formula then gives
    \begin{equation}
    \begin{aligned}
        \cP(\rho_{SP_1P_2})
        &= \tr_{SP_1P_2}\bigl[J_{\cP}(\rho_{SP_1P_2}^T\ox I_{S'})\bigr]\\
        &= \sum_{ij}\sum_{st} r_{st}\,\tr[O_s\,\alpha_{ij}\ketbra{e_i}{e_j}]\,\tr_{P_2}\bigl[O_t(\tau_{ij}^T\ox I_{S'})\bigr]\\
        &\overset{(*)}{=} \sum_{\mu=0}^2 \sum_{\ket{e_i}\in E_\mu}\alpha_{ii}\,\tr_{P_2}\Bigl[\sum_{st} r_{st}\,\lambda_{s\mu}\,O_t(\tau_{ii}^T\ox I_{S'})\Bigr].
    \end{aligned}
    \end{equation}
    Step $(*)$ is pivotal because the $\ket{e_i}$ are simultaneous eigenvectors of every $O_s$. This gives $\tr(O_s\ketbra{e_i}{e_j}) = \lambda_{si}\delta_{ij}$, so only the diagonal $i=j$ contributions survive. Moreover, $\lambda_{si}$ depends on $\ket{e_i}$ only through the subspace $E_\mu$ that contains it, which reduces the inner sum to a function of $\mu$.

    Thus $\cP$ is a classically controlled HPTP map. Regarding the
    diagonal weights $\{\alpha_{ii}\}$ as a probability distribution on
    $\{0,1,\dots,d^2-1\}$ grouped by eigenspace, the programming map acts as
    \begin{equation*}
        \cP(\rho_{SP_1P_2})
        \;=\; \sum_{\mu=0}^{2}\,
               \underbrace{\Bigl(\sum_{\ket{e_i}\in E_\mu}\alpha_{ii}\Bigr)}_{\text{prob.\ of control outcome }\mu}\;\cdot\;
               \cQ_\mu\bigl(\bar\tau_\mu\bigr),
    \end{equation*}
    where $\bar\tau_\mu$ is the conditional average of $\tau_{ii}$ over $\ket{e_i}\in E_\mu$, and $\cQ_\mu$ is the reduced HPTP map with Choi operator
    \begin{equation*}
        J_{\cQ_\mu} \;=\; \sum_{s,t=1}^3 r_{st}\,\lambda_{s\mu}\,O_t,\qquad \mu=0,1,2.
    \end{equation*}
    In effect, $SP_1$ plays the role of a classical controller that detects the eigenspace containing the input and then routes the remaining qudit $P_2$ through the corresponding reduced HPTP map. Evaluating the sums against Table~\ref{tab:eigen} produces the three Choi operators
    \begin{equation}
        \begin{aligned}
            J_{\cQ_0}=& \bigl(\tfrac{1}{d}+\a_{21}-d-d\a_{21}\bigr)I+\bigl(\a_{12}+\a_{22}+d^2\a_{21}+d\a_{22}\bigr)\mathbb{F} +\bigl(-\a_{12}-d\a_{21}-\a_{22}+d^2-d\a_{22}\bigr)\Omega_d,\\
            J_{\cQ_1}=& \bigl(\tfrac{1}{d}+\a_{21}\bigr)I+\bigl(\a_{12}+\a_{22}\bigr)\mathbb{F}+\bigl(-\a_{12}-d\a_{21}-\a_{22}\bigr)\Omega_d,   \\
            J_{\cQ_2}=& \bigl(\tfrac{1}{d}-\a_{21}\bigr)I+ \bigl(\a_{12}-\a_{22}\bigr)\mathbb{F}+\bigl(-\a_{12}+d\a_{21}+\a_{22}\bigr)\Omega_d.
        \end{aligned}
    \end{equation}

    Let $H_d\subset\Op{O}(d)$ be the finite signed-permutation subgroup
    \begin{equation*}
        H_d:=\left\{DP\ \middle|\ D=\operatorname{diag}(s_1,\ldots,s_d),
        \ s_i\in\{\pm1\},\ P\text{ is a permutation matrix}\right\}.
    \end{equation*}
    Averaging over the diagonal signs removes every off-diagonal matrix
    element, and averaging over permutations equalizes the diagonal
    elements. Hence
    \begin{equation*}
        \frac{1}{|H_d|}\sum_{W\in H_d}WXW^\dagger
        =\frac{\tr X}{d}I_d
    \end{equation*}
    for every $X\in\cB(\cH_d)$, so the defining representation is an
    exact one-design. Each reduced map $\cQ_\mu$ inherits the orthogonal
    covariance of $\cP$ and is therefore covariant under $H_d$.
    Lemma~\ref{lem:covariant_diamond_reduction} expresses its diamond
    norm as the following trace norm of its Choi operator.
    \begin{equation}
        \|\cQ_\mu\|_\diamond = \|(\cI_R\ox\cQ_\mu)(\Phi_{RP_2})\|_1 = \frac{1}{d}\|J_{\cQ_\mu}\|_1.
    \end{equation}
    The trace norm admits a very explicit form. Because $J_{\cQ_\mu}$ is a linear combination of the commuting operators $I,\mathbb{F},\Omega_d$, it is diagonal in the joint eigenbasis, and its trace norm reduces to a weighted sum of absolute eigenvalues, with weights given by the subspace dimensions. Denoting by $\Lambda_{\mu\nu}$ the eigenvalue of $J_{\cQ_\mu}$ on $E_\nu$,
    \begin{equation}\label{eq:Nmu_definition}
        N_\mu \;:=\; \|J_{\cQ_\mu}\|_1
        \;=\; \dim E_0\,|\Lambda_{\mu 0}| \;+\; \dim E_1\,|\Lambda_{\mu 1}| \;+\; \dim E_2\,|\Lambda_{\mu 2}|.
    \end{equation}
    Reading the coefficients off the three reduced Choi operators and multiplying by the eigenvalues of Table~\ref{tab:eigen} produces the nine entries
    \begin{equation}
        \begin{aligned}
            \Lambda_{00}=&\bigl(\tfrac{1}{d}-d+d^3\bigr)+(1-d)\a_{12}+(1-d)\a_{21}+(1-d^2)\a_{22}, \\
            \Lambda_{01}=&\bigl(\tfrac{1}{d}-d\bigr)+\a_{12}+(1-d+d^2)\a_{21}+(1+d)\a_{22}, \\
            \Lambda_{02}=&\bigl(\tfrac{1}{d}-d\bigr)-\a_{12}+(1-d-d^2)\a_{21}-(1+d)\a_{22}, \\
            \Lambda_{10}=&\tfrac{1}{d}+(1-d)\a_{12}+(1-d^2)\a_{21}+(1-d)\a_{22}, \\
            \Lambda_{11}=&\tfrac{1}{d}+\a_{12}+\a_{21}+\a_{22}, \\
            \Lambda_{12}=&\tfrac{1}{d}-\a_{12}+\a_{21}-\a_{22}, \\
            \Lambda_{20}=&\tfrac{1}{d}+(1-d)\a_{12}+(d^2-1)\a_{21}+(d-1)\a_{22}, \\
            \Lambda_{21}=&\tfrac{1}{d}+\a_{12}-\a_{21}-\a_{22}, \\
            \Lambda_{22}=&\tfrac{1}{d}-\a_{12}-\a_{21}+\a_{22}.
        \end{aligned}
    \end{equation}

    It remains to pass from the three reduced maps to the full
    programming map. Let $X$ be an arbitrary operator on
    $RSP_1P_2$ with $\|X\|_1\le1$, and expand it in the common
    eigenbasis of $SP_1$ as
    $X=\sum_{ij}\ketbra{e_i}{e_j}_{SP_1}\ox X_{ij}$. The action of
    $\cP$ only sees the diagonal control blocks shown below.
    \begin{equation}
    \begin{aligned}
        (I_R\ox\cP)(X)
        =\sum_{\mu=0}^2\sum_{\ket{e_i}\in E_\mu}
        (I_R\ox\cQ_\mu)(X_{ii}).
    \end{aligned}
    \end{equation}
    Therefore
    \begin{equation*}
    \begin{aligned}
        \|(I_R\ox\cP)(X)\|_1
        &\le \sum_{\mu=0}^2\sum_{\ket{e_i}\in E_\mu}
        \|\cQ_\mu\|_\diamond\,\|X_{ii}\|_1\\
        &\le \max_\mu\|\cQ_\mu\|_\diamond
        \sum_i\|X_{ii}\|_1.
    \end{aligned}
    \end{equation*}
    If $\Delta$ denotes pinching in the basis $\{\ket{e_i}\}$ of the
    $SP_1$ control system, then
    $\sum_i\|X_{ii}\|_1=\|\Delta(X)\|_1\le\|X\|_1\le1$. Hence
    $\|\cP\|_\diamond\le\max_\mu\|\cQ_\mu\|_\diamond$. The reverse
    inequality is obtained by choosing the input supported on a single
    eigenspace $E_\mu$ and on a state attaining the diamond norm of
    $\cQ_\mu$. Thus
    \begin{equation*}
        \|\cP\|_\diamond=\max_\mu\|\cQ_\mu\|_\diamond
        =\frac{1}{d}\max\{N_0,N_1,N_2\}.
    \end{equation*}

    Combining the two reductions, the programming overhead becomes the finite-dimensional minimax
    \begin{equation}
        \nu_1(\cL^\mathbb{R}) \;=\; \frac{1}{d}\inf_{\a_{12},\a_{21},\a_{22}}\,\max\{N_0,N_1,N_2\}.
    \end{equation}
    Each $N_\mu$ is a sum of absolute values of affine functions of
    $(\a_{12},\a_{21},\a_{22})$, so the problem is convex after
    introducing an epigraph variable $t$ with constraints
    $N_\mu(\a)\le t$. The KKT conditions~\cite{boyd2004convex} are
    sufficient for optimality, and take the form
    \begin{equation}\label{eq:core_cond2real}
        \omega_\mu\ge0,\qquad
        \sum_{\mu=0}^2\omega_\mu=1,\qquad
        \omega_\mu(N_\mu(\a)-t)=0,\qquad
        0\in\sum_{\mu=0}^2\omega_\mu\,\partial N_\mu(\a),
    \end{equation}
    where the subgradient is taken with respect to
    $\a=(\a_{12},\a_{21},\a_{22})$.

    Consider the point
    \begin{equation}\label{eq:real_opt_point}
        \a_{12}^\star=-\frac{1}{d(d-1)},\quad
        \a_{21}^\star=\frac{-2d^3+d^2+2d-2}{3d^2(d-1)},\quad
        \a_{22}^\star=\frac{2d^4-2d^3+2}{3d^2(d-1)},
    \end{equation}
    for which direct substitution gives
    $N_0(\a^\star)=N_1(\a^\star)=N_2(\a^\star)$. The corresponding
    overhead is
    \begin{equation*}
        \|\cP\|_\diamond \;=\; \frac{N_1}{d} \;=\; \frac{2d^4-d^2-2d+4}{3d^2}.
    \end{equation*}
    It remains only to certify the subgradient condition in
    \eqref{eq:core_cond2real}. Inserting~\eqref{eq:real_opt_point} into
    the nine eigenvalues $\Lambda_{\mu\nu}$ gives
    \begin{equation}
        \begin{aligned}
            \Lambda_{00}=&\tfrac{2+d(d^3+d-1)}{3d}\quad \Lambda_{01}=0\quad\Lambda_{02}=\tfrac{4-2d(d^2+d-1)}{3d^2}
            \\
            \Lambda_{10}=&\tfrac{2}{d}+d-1 \quad
            \Lambda_{11}=\tfrac{2d(d(d-2)+2)-4}{3d(d-1)} \quad
            \Lambda_{12}=\tfrac{4-2d(d^2+d-1)}{3d^2} \\
            \Lambda_{20}=&\tfrac{2}{d}-d+1 \quad
            \Lambda_{21}=\tfrac{2}{3}\bigl(\tfrac{4}{d}-\tfrac{2}{d-1}-d+1\bigr)\quad
            \Lambda_{22}=\tfrac{2(d^4+d^2-d+2)}{3d^2(d-1)}
        \end{aligned}
    \end{equation}
    For $d\ge 3$ all signs are fixed. The eigenvalues $\Lambda_{00}$, $\Lambda_{10}$, $\Lambda_{11}$, and $\Lambda_{22}$ are positive. The eigenvalues $\Lambda_{02}$, $\Lambda_{12}$, $\Lambda_{20}$, and $\Lambda_{21}$ are negative, while $\Lambda_{01}$ is the only eigenvalue that vanishes at the stationary point. Away from that point each $|\Lambda_{\mu\nu}|$ is smooth, and the gradients of $N_\mu$ are linear combinations of the rows of the coefficient matrix of the $\Lambda$'s with signs determined by those inequalities. At the stationary point itself, the kink in $|\Lambda_{01}|$ contributes a subgradient component $s\in[-1,1]$ times the affine coefficient of $\Lambda_{01}$, where $s$ interpolates between left- and right-derivatives. Evaluating explicitly,
    \begin{equation}
    \begin{aligned}
        \nabla N_0 &= \tfrac{d-1}{2}\begin{pmatrix} s(d+2)+d-2\\ s(d^3+d^2-d+2)+(d^3+d^2-d-2)\\ (d+1)\bigl[\,s(d+2)+d-2\,\bigr]\end{pmatrix},\\
        \nabla N_1 &= d(d-1)(1,-1,1)^T,\qquad
        \nabla N_2 = d(d-1)(-1,-1,1)^T.
    \end{aligned}
    \end{equation}
    The stationarity equation
    $\omega_0\nabla N_0 + \omega_1\nabla N_1 + \omega_2\nabla N_2 = 0$ is
    solved by
    \begin{equation}
        s=-\frac{d^2-2}{d^2+2},\qquad
        \omega_0=\frac{d^2+2}{3d^2},\qquad
        \omega_1=\frac{(d+2)(d-1)}{3d^2},\qquad
        \omega_2=\frac{d(d-1)}{3d^2}.
    \end{equation}
    For every $d\ge 3$, these values satisfy
    $\omega_\mu\ge0$, $\sum_\mu\omega_\mu=1$, and $s\in[-1,1]$. Since all
    three $N_\mu$ are active at $\a^\star$, all KKT conditions hold.

    In the exceptional dimension $d=2$, both $\Lambda_{01}$ and $\Lambda_{20}$ vanish at the stationary point, producing two kinks. Let $r\in[-1,1]$ denote the subgradient parameter attached to $|\Lambda_{20}|$. One admissible certificate is
    \begin{equation}
        s=-\frac{3}{7},\qquad r=0,\qquad
        \omega_0=\frac{7}{15},\qquad
        \omega_1=\frac{2}{5},\qquad
        \omega_2=\frac{2}{15}.
    \end{equation}
    In this case one may take
    $g_0=\frac{1}{7}(-6,10,-18)^T$, $g_1=(2,-2,2)^T$, and
    $g_2=(-3,1,3)^T$ as subgradients of $N_0,N_1,N_2$, respectively.
    The weighted sum
    $\omega_0g_0+\omega_1g_1+\omega_2g_2$ vanishes, and the multipliers
    are nonnegative and sum to one. Thus the KKT certificate also applies
    at $d=2$, giving
    $(2d^4-d^2-2d+4)/(3d^2)\big|_{d=2} = 7/3$.
\end{proof}


\section{Supporting resource comparisons}\label{sec:family_comparison}

Main-text Table~\ref{tab:resource_axes} contrasts program dimension
with exact sampling overhead. Table~\ref{tab:prior_art_supp} separately surveys
representative physical, probabilistic, and virtual programming approaches.
The comparison uses diamond-distance error; its control of physical
channel dilations is quantified by continuity bounds for Stinespring
representations~\cite{Kretschmann2008continuity}.

\begin{table*}[t]
    \centering
    \caption{Separation from representative prior approaches.
    ``Exact'' denotes physical channel retrieval in the physical rows and
    exact observable reconstruction in the virtual rows. The listed
    optimality statements concern different resource variables and should
    not be compared numerically.}
    \label{tab:prior_art_supp}
    \vspace{2pt}
    \footnotesize
    \setlength{\tabcolsep}{2pt}
    \renewcommand{\arraystretch}{1.10}
    \begin{tabular*}{\textwidth}{@{\extracolsep{\fill}}llllll@{}}
        \toprule
        \priorcell{0.135\textwidth}{Approach} &
        \priorcell{0.115\textwidth}{Task} &
        \priorcell{0.175\textwidth}{Program encoding} &
        \priorcell{0.125\textwidth}{Retriever} &
        \priorcell{0.155\textwidth}{Guarantee} &
        \priorcell{0.215\textwidth}{Proven optimality} \\
        \midrule
        \priorcell{0.135\textwidth}{Universal physical gates~\cite{nielsen1997programmable,Kubicki2019resource,Yang2020optimal}} &
        \priorcell{0.115\textwidth}{All unitaries} &
        \priorcell{0.175\textwidth}{General optimized program} &
        \priorcell{0.125\textwidth}{Deterministic CPTP} &
        \priorcell{0.155\textwidth}{Exact impossible at finite dimension; $\varepsilon$-approximate} &
        \priorcell{0.215\textwidth}{Program-size bounds and an asymptotically optimal universal protocol} \\
        \priorcell{0.135\textwidth}{Probabilistic storage and retrieval~\cite{Vidal2002storing,sedlak2019optimal}} &
        \priorcell{0.115\textwidth}{All unitaries} &
        \priorcell{0.175\textwidth}{Memory generated from target uses} &
        \priorcell{0.125\textwidth}{Probabilistic physical comb} &
        \priorcell{0.155\textwidth}{Exact on success} &
        \priorcell{0.215\textwidth}{Optimal retrieval success probability} \\
        \priorcell{0.135\textwidth}{PBT and optimized programs~\cite{ishizaka2008,Christandl2021asymptotic,Banchi2020convex,Yoshida2026quantum,Yoshida2026erratum}} &
        \priorcell{0.115\textwidth}{Channels or isometries} &
        \priorcell{0.175\textwidth}{Choi/PBT resources or query-generated memory} &
        \priorcell{0.125\textwidth}{Deterministic CPTP} &
        \priorcell{0.155\textwidth}{Approximate physical retrieval} &
        \priorcell{0.215\textwidth}{PBT bounds, fixed-processor optimization, or asymptotic program cost} \\
        \priorcell{0.135\textwidth}{Covariant programming~\cite{Gschwendtner2021programmabilityof}} &
        \priorcell{0.115\textwidth}{Group-covariant channels} &
        \priorcell{0.175\textwidth}{Symmetry-compressed Choi program} &
        \priorcell{0.125\textwidth}{Deterministic CPTP} &
        \priorcell{0.155\textwidth}{Exact for input-irreducible actions} &
        \priorcell{0.215\textwidth}{Minimum program dimension} \\
        \priorcell{0.135\textwidth}{Virtual-map frameworks~\cite{Regula2021operational,Jiang2021physical,Takagi2024general}} &
        \priorcell{0.115\textwidth}{Specified map or resource task} &
        \priorcell{0.175\textwidth}{Task dependent} &
        \priorcell{0.125\textwidth}{Quasi-decomposition over physical operations} &
        \priorcell{0.155\textwidth}{Exact statistics or approximate resource conversion} &
        \priorcell{0.215\textwidth}{Optimal cost for the specified map or task} \\
        \priorcell{0.135\textwidth}{Programmable open systems~\cite{jing2025programmable}} &
        \priorcell{0.115\textwidth}{Lindbladian semigroups} &
        \priorcell{0.175\textwidth}{Time-varying $\pi_t$} &
        \priorcell{0.125\textwidth}{CPTP or HPTP} &
        \priorcell{0.155\textwidth}{Family dependent} &
        \priorcell{0.215\textwidth}{Structural laws and family-specific constructions} \\
        \priorcell{0.135\textwidth}{This work} &
        \priorcell{0.115\textwidth}{All $\Op{CPTP}_d$} &
        \priorcell{0.175\textwidth}{Fixed $\pi_\cE^{\ox k}$} &
        \priorcell{0.125\textwidth}{Target-independent HPTP} &
        \priorcell{0.155\textwidth}{Exact observable reconstruction} &
        \priorcell{0.215\textwidth}{Exact $\nu_1$ and sharp fixed-$d$ asymptotic law for $\nu_k$} \\
        \bottomrule
    \end{tabular*}
\end{table*}

For completeness, Table~\ref{tab:family_overheads_supp} compares the
scalings obtained here with representative restricted open-system
constructions from Ref.~\cite{jing2025programmable}. This supporting table
concerns how target restrictions change the prescribed program and quasiprobability
overhead, rather than the optimized physical-memory dimension used in
Table~\ref{tab:resource_axes}.

\begin{table}[h]
    \centering
    \caption{Program states and quasiprobability overheads for the channel
    families studied here and representative open-system dynamics from
    Ref.~\cite{jing2025programmable}. Here $\cD_p$ denotes the
    $d$-dimensional depolarising family with $p\in[0,1]$, and $K$ counts
    vertices of the invariant Pauli-channel polytope. The Hamiltonian entry
    is an achievable upper bound, not an optimum.}
    \label{tab:family_overheads_supp}
    \vspace{2pt}
    \setlength{\tabcolsep}{3pt}
    \renewcommand{\arraystretch}{1.18}
    \begin{tabular*}{\columnwidth}{@{\extracolsep{\fill}}lll@{}}
        \toprule
        Family $\cS$ & Program state & Overhead \\
        \midrule
        All $d$-dim channels & Choi program $\pi_\cE$ & $\Theta(d^2)$ \\[2pt]
        Unitary/unital channels & Choi program $\pi_\cE$ & $\Theta(d^2)$ \\[2pt]
        Real channels & Choi program $\pi_\cE$ & $\Theta(d^2)$ \\[2pt]
        Hamiltonian dynamics & $\cO(d)$-dim $\pi_t$ & $\cO(d)$ \\[2pt]
        Fully dissipative Pauli & $\mathcal{O}(K)$-dim $\pi_t$ & $\cO(1)$ \\[2pt]
        Depolarising family $\cD_p$ & $\mathcal{O}(1)$-dim $\pi_p$ & $\cO(1)$ \\
        \bottomrule
    \end{tabular*}
\end{table}

For the last three rows, the reported overhead is $2^\gamma$, where
$\gamma$ is the logarithmic quasiprobability cost of
Ref.~\cite{jing2025programmable}. The first three rows are
the exact results proved here, the Hamiltonian row is only an achievable
bound, and the two listed dissipative constructions have unit overhead.

\section{Properties of the programming overhead}\label{sec:property}
This appendix establishes structural properties of the one-copy programming
overhead $\nu_1$. Throughout,
$\varnothing\ne\cS\subseteq\Op{CPTP}(\cH_A\to\cH_B)$. We identify the signal systems
$S$ and $S'$ with $A$ and $B$, respectively. The program systems satisfy
$\cH_{P_1}\simeq\cH_A$ and $\cH_{P_2}\simeq\cH_B$, and
$d_A:=\dim\cH_A$. The arguments use the normalized Choi program state
$\pi_\cE=J_\cE/d_A$ and the quasi-quantum retriever convention fixed in
Appendix~\ref{appendix:notation}.
In particular, Lemmas~\ref{lem:set_monotonicity},
\ref{lem:aff_invariance}, and \ref{lem:tensor_submultiplicativity}
prove the structural laws stated in the main text.

We first record set monotonicity. Enlarging the target channel set cannot
reduce the optimal overhead.

\begin{lemma}[Set monotonicity]\label{lem:set_monotonicity}
    If $\varnothing\ne\cS_1\subseteq\cS_2\subseteq\Op{CPTP}(A\to B)$,
    then $\nu_1(\cS_1)\le\nu_1(\cS_2)$.
\end{lemma}
\begin{proof}
    Write $\mathfrak{F}(\cS)$ for the set of HPTP retrievers that program
    every channel in $\cS$. From $\cS_1\subseteq\cS_2$ it follows that
    $\mathfrak{F}(\cS_2)\subseteq\mathfrak{F}(\cS_1)$, since a retriever
    feasible for every $\cE\in\cS_2$ is feasible for every
    $\cE\in\cS_1$.
    Minimizing the diamond norm over a larger feasible set can only
    decrease or preserve the optimum, giving
    $\nu_1(\cS_1)\le\nu_1(\cS_2)$.
\end{proof}

Beyond monotonicity, $\nu_1$ is invariant under complex conjugation or
transposition of every target channel in a fixed computational basis. For
$\cE(\rho)=\sum_j K_j\rho K_j^\dagger$, define the conjugate channel by
$\bar{\cE}(\rho):=\sum_j\overline{K_j}\rho\overline{K_j}^\dagger$ and the
transpose channel by $\cE^T(\rho):=(\cE(\rho^T))^T$. Both are physical
quantum channels. We also set $\bar{\cS}:=\{\bar{\cE}:\cE\in\cS\}$ and
$\cS^T:=\{\cE^T:\cE\in\cS\}$.

\begin{lemma}[Conjugation-transpose symmetry]\label{lem:conj_transpose}
    $\nu_1(\cS)=\nu_1(\bar{\cS})=\nu_1(\cS^T)$.
\end{lemma}
\begin{proof}
    Conjugation and transposition of a channel coincide. The following
    Kraus-level computation establishes this equality.
    \begin{equation*}
    \cE^T(\rho) = \Big(\sum_j K_j\rho^T K_j^\dagger\Big)^T
    = \sum_j \overline{K_j}\rho\overline{K_j}^\dagger = \bar{\cE}(\rho),
    \end{equation*}
    using $K_j^T=\overline{K_j}^\dagger$ and
    $(K_j^\dagger)^T=\overline{K_j}$. Hence $\cS^T=\bar{\cS}$, and it suffices
    to prove $\nu_1(\bar{\cS})=\nu_1(\cS)$. Conjugating every Kraus
    operator of $\cE$ conjugates its Choi operator as well,
    $J_{\bar{\cE}}=\sum_{ij}\ketbra{i}{j}\ox\bar{\cE}(\ketbra{i}{j})=\overline{J_\cE}$,
    so that $\pi_{\bar{\cE}}=\overline{\pi_\cE}$. For any feasible HPTP
    retriever $\cP$ for $\cS$, set
    $\bar{\cP}(\cdot):=\overline{\cP(\overline{\cdot})}$, which is again
    HPTP. For every $\cE\in\cS$,
    \begin{equation*}
    \bar{\cP}(\rho\ox\overline{\pi_\cE})
    = \overline{\cP(\overline{\rho}\ox\pi_\cE)}
    = \overline{\cE(\overline{\rho})}
    = \bar{\cE}(\rho),
    \end{equation*}
    making $\bar{\cP}$ feasible for $\bar{\cS}$. Since complex conjugation
    is an isometry in trace norm,
    \begin{equation*}
    \|(\bar{\cP}\ox\cI_R)(X_{SPR})\|_1
    = \|\overline{(\cP\ox\cI_R)(\overline{X}_{SPR})}\|_1
    = \|(\cP\ox\cI_R)(\overline{X}_{SPR})\|_1,
    \end{equation*}
    whence $\|\bar{\cP}\|_\diamond=\|\cP\|_\diamond$ and
    $\nu_1(\bar{\cS})\le\nu_1(\cS)$. Applying the same construction to
    $\bar{\cS}$ yields the reverse inequality and establishes equality.
\end{proof}

We next quantify the effect of a common post-processing channel. A map
$\cF\in\Op{CPTP}(B\to B)$ is
\emph{invertible} when its inverse exists as a linear map on
$\cB(\cH_B)$. In that situation, $\cF^{-1}$ is automatically Hermiticity-
and trace-preserving, though typically not completely positive. Write
$\cF\circ\cS:=\{\cF\circ\cE:\cE\in\cS\}$ and
$\cS\circ\cV:=\{\cE\circ\cV:\cE\in\cS\}$.

\begin{theorem}[Post-processing stability]\label{thm:post_stability}
    Let $\cF\in\Op{CPTP}(B\to B)$ be invertible. Then
    \begin{equation*}
        \frac{\nu_1(\cS)}{\|\cF^{-1}\|_\diamond}
        \;\leq\; \nu_1(\cF\circ\cS)
        \;\leq\; \nu_1(\cS)\cdot\|\cF^{-1}\|_\diamond.
    \end{equation*}
    When $\cF=\cU$ is a unitary channel, $\|\cF^{-1}\|_\diamond=1$ and
    the two bounds collapse to $\nu_1(\cU\circ\cS)=\nu_1(\cS)$.
\end{theorem}
\begin{proof}
    Post-composition with $\cF$ acts on the Choi state as
    $\pi_{\cF\circ\cE}=(\cI_{P_1}\ox\cF_{P_2})(\pi_\cE)$. This identity
    gives feasible retrievers in both directions.

    \emph{Upper bound.} Let $\cP^*$ be an optimal retriever for $\cS$, so
    that $\|\cP^*\|_\diamond=\nu_1(\cS)$. Define
    \begin{equation*}
        \cP' \;:=\; \cF_{S'}\circ\cP^*_{SP\to S'}\circ
        (\cI_S\ox\cI_{P_1}\ox\cF^{-1}_{P_2}),
    \end{equation*}
    where $\cF^{-1}_{P_2}$ recovers the original program state and
    $\cF_{S'}$ applies the required post-processing. For any
    $\cE\in\cS$,
    \begin{equation*}
    \cP'(\rho\ox\pi_{\cF\circ\cE})
    = \cF\!\left(\cP^*\!\left(\rho\ox(\cI\ox\cF^{-1})(\cI\ox\cF)(\pi_\cE)\right)\right)
    = \cF(\cP^*(\rho\ox\pi_\cE))
    = \cF(\cE(\rho)),
    \end{equation*}
    so $\cP'$ is feasible for $\cF\circ\cS$. Combining
    $\|\cF\|_\diamond=1$ with multiplicativity of the diamond norm under
    tensor product and sub-multiplicativity under
    composition~\cite{Jiang2021physical},
    \begin{equation*}
    \|\cP'\|_\diamond
    \;\leq\; \|\cF\|_\diamond\,\|\cP^*\|_\diamond\,\|\cF^{-1}\|_\diamond
    = \nu_1(\cS)\cdot\|\cF^{-1}\|_\diamond.
    \end{equation*}

    \emph{Lower bound.} For any feasible retriever $\cP'$ for
    $\cF\circ\cS$, the mirror construction
    \begin{equation*}
        \widetilde{\cP} \;:=\; \cF^{-1}_{S'}\circ\cP'_{SP\to S'}\circ
        (\cI_S\ox\cI_{P_1}\ox\cF_{P_2})
    \end{equation*}
    satisfies
    $\widetilde{\cP}(\rho\ox\pi_\cE)=\cF^{-1}(\cF(\cE(\rho)))=\cE(\rho)$
    and is therefore feasible for $\cS$. The same diamond-norm estimate
    gives $\nu_1(\cS)\leq\|\widetilde{\cP}\|_\diamond
    \leq\|\cF^{-1}\|_\diamond\,\|\cP'\|_\diamond$, and minimizing over
    $\cP'$ yields
    $\nu_1(\cS)\leq\|\cF^{-1}\|_\diamond\,\nu_1(\cF\circ\cS)$. For $\cF=\cU$,
    $\cF^{-1}=\Op{Ad}_{U^\dagger}\in\Op{CPTP}(B\to B)$ has unit
    diamond norm.
\end{proof}

Pre-processing acts instead on the input half of the Choi state. For a
unitary channel $\cV=\Op{Ad}_V$, this action remains unitary and gives an
exact invariance.

\begin{theorem}[Unitary pre-processing invariance]\label{thm:pre_stability}
    Let $\cV=\Op{Ad}_V$ be a unitary channel on $A$. Then
    $\nu_1(\cS\circ\cV)=\nu_1(\cS)$.
\end{theorem}
\begin{proof}
    The program states obey
    \begin{equation*}
        \pi_{\cE\circ\cV}
        =\bigl((\Op{Ad}_{V^T})_{P_1}\ox\cI_{P_2}\bigr)(\pi_\cE).
    \end{equation*}
    Let $\cP^*$ be an optimal retriever for $\cS$ and define
    \begin{equation*}
        \cP':=\cP^*\circ
        \bigl(\cV_S\ox(\Op{Ad}_{\overline V})_{P_1}\ox\cI_{P_2}\bigr).
    \end{equation*}
    Since $\Op{Ad}_{\overline V}$ is the inverse of
    $\Op{Ad}_{V^T}$, for every $\cE\in\cS$,
    \begin{equation*}
        \cP'(\rho\ox\pi_{\cE\circ\cV})
        =\cP^*(\cV(\rho)\ox\pi_\cE)
        =(\cE\circ\cV)(\rho).
    \end{equation*}
    Thus $\cP'$ is feasible for $\cS\circ\cV$. Composition with the
    unitary input channel preserves the diamond norm, so
    $\nu_1(\cS\circ\cV)\leq\nu_1(\cS)$. Applying the same construction to
    $\Op{Ad}_{V^\dagger}$ gives the reverse inequality.
\end{proof}

Combining unitary pre-processing invariance with the unitary case of
Theorem~\ref{thm:post_stability} gives invariance under fixed unitary
channels on both sides.
\begin{corollary}[Unitary invariance]\label{cor:unitary_invariance}
    Let $\cU,\cV$ be unitary channels. Then
    $\nu_1(\cU\circ\cS\circ\cV)=\nu_1(\cS)$.
\end{corollary}

\begin{remark}
    Quasi-decompositions of HPTP
    maps~\cite{Jiang2021physical,Regula2021operational} give
    $\|\cF^{-1}\|_\diamond$ a direct operational meaning,
    \begin{equation*}
    \|\cF^{-1}\|_\diamond
    = \min\{\alpha_++\alpha_- \,:\,
    \cF^{-1}=\alpha_+\cQ_+-\alpha_-\cQ_-,\;
    \cQ_\pm\in\Op{CPTP}(B\to B),\;\alpha_\pm\geq 0\},
    \end{equation*}
    which is the minimum sampling overhead for simulating the non-physical
    map $\cF^{-1}$ on a physical device.
    Theorem~\ref{thm:post_stability} therefore bounds the change in
    programming overhead by the same inverse-map sampling cost.
\end{remark}

The processing bounds compare related target sets. A complementary lower
bound measures how much a feasible retriever must amplify the
distinguishability of Choi program states.

\begin{theorem}[Distinguishability amplification lower bound]\label{thm:distinguishability_lb}
    For any $\cS\subseteq\Op{CPTP}(A\to B)$,
    \begin{equation*}
    \nu_1(\cS) \;\geq\;
    \max\left\{1,\;
    \sup_{\substack{\cE_1,\cE_2\in\cS\\ \cE_1\neq\cE_2}}
    \frac{\|\cE_1-\cE_2\|_\diamond}
    {\|\pi_{\cE_1}-\pi_{\cE_2}\|_1}\right\},
    \end{equation*}
    where the supremum over an empty set is defined as zero.
\end{theorem}
\begin{proof}
    Fix distinct $\cE_1,\cE_2\in\cS$ and any feasible HPTP retriever
    $\cP$ for $\cS$,
    and set $\Delta:=\pi_{\cE_1}-\pi_{\cE_2}$. Linearity of $\cP$ extends
    the programming condition from the individual $\pi_{\cE_i}$ to
    $\Delta$. After tensoring with an auxiliary register $R$, every
    $X_{SR}\in\cB(\cH_S\ox\cH_R)$ satisfies
    \begin{equation*}
    (\cP_{SP\to S'}\ox\cI_R)(X_{SR}\ox\Delta_P)
    = ((\cE_1-\cE_2)\ox\cI_R)(X_{SR}).
    \end{equation*}
    Trace norms then combine with $\|A\ox B\|_1=\|A\|_1\|B\|_1$ and
    $\|\cP\ox\cI_R\|_{1\to 1}\leq\|\cP\|_\diamond$ to yield
    \begin{equation*}
    \|((\cE_1-\cE_2)\ox\cI_R)(X_{SR})\|_1
    \;\leq\;
    \|\cP\|_\diamond\cdot\|X_{SR}\|_1\cdot\|\Delta\|_1.
    \end{equation*}
    Taking the supremum over unit-trace-norm $X_{SR}$ and over all
    $\cH_R$ on the left-hand side gives
    $\|\cE_1-\cE_2\|_\diamond\leq\|\cP\|_\diamond\cdot\|\Delta\|_1$, and
    minimizing over $\cP$ followed by the stated supremum gives the
    ratio bound. Every HPTP retriever has diamond norm at least one,
    which completes the proof.
\end{proof}

\begin{remark}
    (\textbf{Faithfulness}) Each ratio in the supremum is at least one.
    Evaluating the diamond norm on the normalized maximally entangled
    state on $R\ox A$ gives
    $\|\cE_1-\cE_2\|_\diamond
    \geq\|(\cI_R\ox(\cE_1-\cE_2))(\Phi_{d_A})\|_1
    =\|\pi_{\cE_1}-\pi_{\cE_2}\|_1$. The ratio bound becomes stronger
    than $\nu_1(\cS)\geq 1$ precisely when the
    diamond distance between channels strictly exceeds the trace distance
    between their Choi program states. In this sense any feasible
    retriever must amplify distinguishability from the program register
    back to the corresponding operational distinguishability of the
    target channels. The closed-form lower bounds used elsewhere are proved
    by symmetry reduction and explicit primal--dual certificates, rather
    than by this distinguishability estimate alone.
\end{remark}

Theorem~\ref{thm:post_stability} assumes invertibility on the full output
operator space. A useful substitute requires a common HPTP recovery map
only for the transformed program states.

\begin{definition}[Program-state recovery overhead]\label{def:recovery_overhead}
    For $\cF\in\Op{CPTP}(B\to B)$ and $\cS\subseteq\Op{CPTP}(A\to B)$, the
    \emph{program-state recovery overhead} is
    \begin{equation*}
    \kappa_\cS(\cF) := \inf\left\{\|\cR\|_\diamond \,:\,
    \cR\in\Op{HPTP}(B\to B),\;
    (\cI_{P_1}\ox\cR_{P_2})\circ(\cI_{P_1}\ox\cF_{P_2})(\pi_\cE)=\pi_\cE
    \;\forall\,\cE\in\cS\right\}.
    \end{equation*}
    We set $\kappa_\cS(\cF):=+\infty$ if no such $\cR$ exists.
\end{definition}

\begin{proposition}\label{prop:recovery_bound}
    For any $\cF\in\Op{CPTP}(B\to B)$, not necessarily invertible,
    \begin{equation*}
    \nu_1(\cF\circ\cS) \;\leq\; \nu_1(\cS)\cdot\kappa_\cS(\cF).
    \end{equation*}
    When $\cF$ is invertible, $\cR=\cF^{-1}$ is feasible, so
    $\kappa_\cS(\cF)\leq\|\cF^{-1}\|_\diamond$ and the upper bound in
    Theorem~\ref{thm:post_stability} is recovered.
\end{proposition}
\begin{proof}
    The claim is immediate if $\kappa_\cS(\cF)=+\infty$. Otherwise, fix
    $\varepsilon>0$ and choose a feasible $\cR_\varepsilon$ such that
    $\|\cR_\varepsilon\|_\diamond\leq\kappa_\cS(\cF)+\varepsilon$.
    Let $\cP^*$ be an optimal retriever for $\cS$. Replacing $\cF^{-1}$
    by $\cR_\varepsilon$ in the construction of
    Theorem~\ref{thm:post_stability} yields
    \begin{equation*}
    \cP_\varepsilon := \cF_{S'}\circ\cP^*_{SP\to S'}
    \circ(\cI_S\ox\cI_{P_1}\ox(\cR_\varepsilon)_{P_2}).
    \end{equation*}
    The defining property of $\cR_\varepsilon$,
    $(\cI\ox\cR_\varepsilon)(\cI\ox\cF)(\pi_\cE)=\pi_\cE$ for every
    $\cE\in\cS$,
    gives
    \begin{equation*}
    \cP_\varepsilon(\rho\ox\pi_{\cF\circ\cE})
    = \cF\!\left(\cP^*(\rho\ox\pi_\cE)\right)
    = \cF(\cE(\rho)),
    \end{equation*}
    so $\cP_\varepsilon$ is feasible for $\cF\circ\cS$. The same
    diamond-norm estimate gives
    \begin{equation*}
    \nu_1(\cF\circ\cS)
    \leq\nu_1(\cS)\bigl(\kappa_\cS(\cF)+\varepsilon\bigr).
    \end{equation*}
    Letting $\varepsilon\to0$ proves the claim.
\end{proof}

\begin{remark}
    The feasibility condition
    $(\cI\ox\cR)(\cI\ox\cF)(\pi_\cE)=\pi_\cE$ requires recovery on the
    specified program states and hence on their linear span. It does not
    require a global inverse for $\cF$. This restricted recovery condition
    is analogous in form, but not equivalent, to exact code-state
    recovery in quantum error correction~\cite{knill2000theory}.
\end{remark}

The recovery result concerns transformations of a fixed target set. We
next ask whether adding physical affine combinations changes its
programming overhead.

\begin{definition}[Affine extension]\label{def:affine_extension}
    For $\varnothing\ne\cS\subseteq\Op{CPTP}(A\to B)$, the \emph{affine extension}
    inside $\Op{CPTP}(A\to B)$ is
    \begin{equation*}
    \cA[\cS] = \Op{CPTP}(A\to B)\cap
    \Big\{\cE' = \sum_{i=1}^n\lambda_i\cE_i
    \;\Big|\; \lambda_i\in\mathbb{R},\;\sum_i\lambda_i=1,\;
    \cE_i\in\cS,\; n\in\mathbb{N}\Big\}.
    \end{equation*}
\end{definition}

\begin{lemma}[Affine invariance]\label{lem:aff_invariance}
    $\nu_1(\cA[\cS])=\nu_1(\cS)$.
\end{lemma}
\begin{proof}
    Write $\cE'=\sum_i\lambda_i\cE_i\in\cA[\cS]$ with $\cE_i\in\cS$ and
    $\sum_i\lambda_i=1$. Linearity of the Choi map gives
    $J_{\cE'}=\sum_i\lambda_iJ_{\cE_i}$. Membership in the channel set
    ensures $J_{\cE'}\geq 0$.

    For any optimal retriever $\cP$ for $\cS$, linearity gives
    \begin{equation*}
    \cP\!\left(\rho\ox\pi_{\cE'}\right)
    = \sum_i\lambda_i\cP\!\left(\rho\ox\pi_{\cE_i}\right)
    = \sum_i\lambda_i\cE_i(\rho) = \cE'(\rho),
    \end{equation*}
    so $\cP$ programs $\cE'$ as well, at unchanged overhead
    $\nu_1(\cS)$. Hence $\nu_1(\cA[\cS])\leq\nu_1(\cS)$, and the
    reverse inequality is a consequence of $\cS\subseteq\cA[\cS]$ and
    Lemma~\ref{lem:set_monotonicity}.
\end{proof}

Affine invariance extends the overhead from a target set to its physical
affine closure without increasing it. Every unital channel is a real
affine combination of unitary channels, giving the following consequence.

\begin{corollary}[Unital channels from affine invariance]\label{cor:unital_affine_invariance}
    Let $\cT(d)$ be the set of unital channels acting on $\cH_d$. Then
    \begin{equation*}
        \nu_1(\cT(d))=\nu_1(\Op{Ad}_{\Op{SU}(d)}).
    \end{equation*}
\end{corollary}
\begin{proof}
    By Theorem 1 of Ref.~\cite{mendl2009unital}, every unital channel
    $\cE\in\cT(d)$ admits a real affine decomposition
    into unitary channels. Global phases do not change adjoint channels,
    so the unitary representatives may be chosen in $\Op{SU}(d)$. Thus
    $\cE=\sum_i\lambda_i\cU_i$ with
    $\cU_i\in\Op{Ad}_{\Op{SU}(d)}$,
    $\lambda_i\in\mathbb{R}$, and $\sum_i\lambda_i=1$, so that
    $\cT(d)=\cA[\Op{Ad}_{\Op{SU}(d)}]$. The claim follows from
    Lemma~\ref{lem:aff_invariance}.
\end{proof}

The physical affine closure therefore has the same programming overhead
as its generating target set. We finally consider two target channel sets
that are programmed in parallel.

\begin{definition}[Tensor-product channel sets]\label{def:tensor_product_sets}
    Given nonempty $\cS_i\subseteq\Op{CPTP}(A_i\to B_i)$ for
    $i\in\{1,2\}$, their
    tensor product is
    \begin{equation*}
        \cS_1\ox\cS_2:=\{\cE_1\ox\cE_2 \;|\; \cE_1\in \cS_1,\; \cE_2\in\cS_2\}.
    \end{equation*}
\end{definition}

\begin{lemma}[Submultiplicativity under tensor product]\label{lem:tensor_submultiplicativity}
    For nonempty $\cS_i\subseteq\Op{CPTP}(A_i\to B_i)$ with
    $i\in\{1,2\}$,
    \begin{equation*}
        \nu_1(\cS_1\ox\cS_2)\le\nu_1(\cS_1)\,\nu_1(\cS_2).
    \end{equation*}
\end{lemma}
\begin{proof}
    Up to the canonical permutation between the $A_1B_1A_2B_2$ and
    $A_1A_2B_1B_2$ register orders, the program state factorizes as
    $\pi_{\cE_1\ox\cE_2}=\pi_{\cE_1}\ox\pi_{\cE_2}$.
    Let $\cP_1$ and $\cP_2$ be optimal HPTP retrievers for $\cS_1$ and
    $\cS_2$. Their tensor product, composed with the fixed input-register
    permutation, defines a joint retriever $\cP$. For arbitrary
    $X_i\in\cB(\cH_{A_i})$,
    \begin{equation*}
        (\cP_1\ox\cP_2)
        \bigl((X_1\ox\pi_{\cE_1})\ox(X_2\ox\pi_{\cE_2})\bigr)
        =\cE_1(X_1)\ox\cE_2(X_2).
    \end{equation*}
    Product operators span
    $\cB(\cH_{A_1}\ox\cH_{A_2})$, so linearity extends this identity to
    arbitrary signal inputs. Hence $\cP$ is feasible for
    $\cS_1\ox\cS_2$. The fixed permutation is unitary and has unit
    diamond norm. Multiplicativity under tensor products for HPTP
    maps~\cite{Jiang2021physical} therefore gives
    \begin{equation*}
        \|\cP\|_\diamond \;=\; \|\cP_1\|_\diamond\,\|\cP_2\|_\diamond \;=\; \nu_1(\cS_1)\,\nu_1(\cS_2),
    \end{equation*}
    which upper-bounds $\nu_1(\cS_1\ox\cS_2)$.
\end{proof}

Numerical evidence indicates that the submultiplicative inequality can be
strict. A two-channel qubit instance, its retained rational correction data,
and code reproducing the numerical strict-gap check are publicly
available in the accompanying GitHub
repository~\cite{QuAIR2026UniversalProgrammingCodes}.

\section{Many-copy overhead for universal channel programming}\label{sec:kcopy}

Here the retriever receives several identical Choi program states. Let
$\cH_S\cong\cH_{S'}\cong\CC^d$ be the signal input and
output spaces. Each program copy has space
$\cH_P:=\cH_A\ox\cH_B$, where $A\simeq S$ and $B\simeq S'$.
These are the registers denoted by $P_1$ and $P_2$ in
Appendix~\ref{appendix:notation}. With the Choi convention and
normalization fixed there, we write
\[
J_\cE
:=(\cI_A\ox\cE_{S\to B})(\Omega_{AS})
\in\cB(\cH_A\ox\cH_B),
\qquad
\pi_\cE:=J_\cE/d.
\]

For Theorem~\ref{thm:kcopy_scaling} and
Propositions~\ref{prop:kcopy_upper}--\ref{prop:kcopy_lower}, $d\geq2$
is a fixed integer and $k$ ranges over the positive integers. For each
pair $(d,k)$, $\nu_k(\Op{CPTP}_d)$ optimizes over one HPTP map
$\cP_{k,d}$, which may depend on $d$ and $k$ but not on the target
$\cE$. It receives exactly the product memory $\pi_\cE^{\ox k}$ and
must reproduce $\cE$ for every $\cE\in\Op{CPTP}_d$ and every signal
input; by linearity, this is equality of the retrieved and target maps
on all input operators. No alternative supplied program encoding---whether
correlated, compressed, or otherwise---is optimized over; arbitrary fixed
preprocessing of the prescribed product memory may still be included in
$\cP_{k,d}$. The programming tolerance is fixed at $\varepsilon=0$, and
every $k\to\infty$ limit below is taken at fixed $d$, with constants and
remainders allowed to depend on $d$.

The exact overhead used throughout this appendix is the
$\varepsilon=0$ specialization in
Definition~\ref{def:koverhead}. In particular, its feasibility
condition is
$\cP(\rho\ox\pi_\cE^{\ox k})=\cE(\rho)$ for every input state
$\rho$ and every target channel $\cE$.

The preceding appendices treat the single-copy case $\nu_1(\cS)$. We first
establish the sharp fixed-dimension
$1/k$ law for universal programming and then reduce the finite-$k$
problem using the mixed-tensor commutant. The Choi representation of
$\cP$ gives a direct starting point. As in
Appendix~\ref{appendix:notation},
for $J_\cP\in\cB(\cH_S\ox\cH_P^{\ox k}\ox\cH_{S'})$, the link
product formula~\cite{chiribella2008quantum} gives, for any $\sigma$
on $\cH_S\ox\cH_P^{\ox k}$,
\begin{equation}
\cP(\sigma)=\tr_{S,P^{\ox k}}\bigl[J_\cP\,(\sigma^T\ox I_{S'})\bigr].
\end{equation}
Inserting $\sigma=\rho\ox\pi_\cE^{\ox k}$ and imposing
$\cP(\sigma)=\cE(\rho)$ for every $\rho$ yields the equivalent
condition on $J_\cP$,
\begin{equation}\label{eq:kconstraint}
\tr_{P^{\ox k}}\bigl[J_\cP\,\bigl(I_S\ox(J_\cE^T)^{\ox k}
\ox I_{S'}\bigr)\bigr]
\;=\;d^k\,J_\cE,\qquad\forall\cE\in\cS.
\end{equation}

Decomposing $J_\cP=J_+-J_-$ into a difference of two
positive-semidefinite operators $J_\pm\ge 0$, with
$\tr_{S'}[J_\pm]=p_\pm I_{S,P^{\ox k}}$, casts the
$k$-copy overhead as a semidefinite program on the total space
$\cH_{\mathrm{tot}}:=\cH_S\ox\cH_P^{\ox k}\ox\cH_{S'}$, whose dimension
is $D:=d^{2(k+1)}$.
\begin{equation}\label{eq:kcopy_SDP_unreduced}
\begin{aligned}
    &\underline{\textbf{Primal program ($k$ copies)}}\\
    \nu_k(\cS)&=\min\;p_++p_-\\
    {\rm s.t.}\;\;& J_{\cP}:=J_+-J_-,\\
    &\tr_{P^{\ox k}}[J_{\cP}(I_S\ox(J_{\cE}^{T})^{\ox k}\ox I_{S'})]
      = d^k J_{\cE},\, \forall \cE \in \cS,\\
    &J_+\geq 0,\,\tr_{S'}[J_+] =p_+ I_{S,P^{\ox k}},\\
    &J_-\geq 0,\,\tr_{S'}[J_-] =p_- I_{S,P^{\ox k}}.
\end{aligned}
\end{equation}
The programming equality is the Choi condition~\eqref{eq:kconstraint},
whereas the remaining lines impose positivity and the partial-trace
conditions for the two scaled channel Choi operators. These constraints
already enforce the trace preservation of $J_\cP$. For any $\cE\in\cS$,
tracing the programming equality over $S'$ yields
\begin{equation*}
\begin{aligned}
d^k(p_+-p_-)I_S
&=\tr_{P^{\ox k},S'}\!\left[J_\cP
  \bigl(I_S\ox(J_\cE^T)^{\ox k}\ox I_{S'}\bigr)\right]\\
&=d^k\,\tr_{S'}J_\cE=d^kI_S.
\end{aligned}
\end{equation*}
Thus $p_+-p_-=1$ follows from the displayed SDP and is not an
independent restriction. The identification of the objective $p_++p_-$ with
$\|\cP\|_\diamond$ is the base-norm characterization of the diamond
norm for Hermitian-preserving trace-preserving
maps~\cite[Theorem~3]{Regula2021operational}, invoked again in the
proof of Proposition~\ref{prop:kcopy_lower}. Because $J_\pm$ are $D\times D$ matrices
with $D$ growing exponentially in $k$, even moderate values of $d$ and
$k$ render direct solution infeasible.
The symmetry analysis below reduces~\eqref{eq:kcopy_SDP_unreduced} to an
equivalent program whose variable count and constraint sizes are
determined by the representation theory of the walled Brauer algebra,
rather than by~$D$.


Two structural observations prepare the ground. The programming
constraint is linear and $S_k$-symmetric. The overhead is monotone in the
number of copies. A further group average, deferred until after
the scaling theorem, then completes the reduction of the search space
from the full operator algebra $\cB(\cH_{\mathrm{tot}})$ to the
commutant of $G\times S_k$.

\begin{lemma}[$S_k$-symmetry of the programming SDP]
\label{lem:kcopy_symmetric}
Let $J_\cP=J_+-J_-$ be feasible in
Eq.~\eqref{eq:kcopy_SDP_unreduced}. Then the operators
\[
J_\pm^{\mathrm{sym}}:=\frac{1}{k!}\sum_{\sigma\in S_k}
(I\ox U_\sigma\ox I)\,J_\pm\,(I\ox U_\sigma^\dagger\ox I),
\]
where $U_\sigma$ permutes the $k$ copies of $\cH_P$, form a feasible
quasi-decomposition
$J_{\cP^{\mathrm{sym}}}=J_+^{\mathrm{sym}}-J_-^{\mathrm{sym}}$
with the same $p_\pm$. In particular, an optimal quasi-decomposition may be
chosen $S_k$-invariant.
\end{lemma}
\begin{proof}
Unitary conjugation and averaging preserve positivity, and
$\tr_{S'}J_\pm^{\mathrm{sym}}=p_\pm I$. The left-hand side
of~\eqref{eq:kconstraint} is a contraction of $J_\cP$ against
$(J_\cE^T)^{\ox k}$, which is $S_k$-invariant. Averaging $J_\cP$
over the $S_k$-conjugation on the program tensor factors therefore
leaves the programming equality unchanged. The objective remains
$p_++p_-$, which proves the claim.
\end{proof}

Whereas symmetrization operates at a fixed number of copies, the
second observation compares the overhead across different values of
$k$.

\begin{lemma}[Copy monotonicity]\label{lem:kcopy_mono}
For every nonempty channel set $\cS$ and $k\ge 1$,
$\nu_{k+1}(\cS)\le\nu_k(\cS)$. Consequently
$\nu_k(\cS)$ is a non-increasing, bounded-below sequence, so the
limit
$\nu_\infty(\cS):=\lim_{k\to\infty}\nu_k(\cS)$ exists and
satisfies $\nu_\infty(\cS)\ge 1$.
\end{lemma}
\begin{proof}
Let $\cP_k$ be any feasible $k$-copy retriever and define
$\widehat\cP_{k+1}:\cB(\cH_S\ox\cH_P^{\ox(k+1)})
\to\cB(\cH_{S'})$ by
\[
\widehat\cP_{k+1}(Z)
:=\cP_k\bigl(\tr_{A_{k+1}B_{k+1}}Z\bigr).
\]
As the composition of a physical partial trace and an HPTP map,
$\widehat\cP_{k+1}$ is HPTP, with
$\widehat\cP_{k+1}(\rho\ox\pi_\cE^{\ox(k+1)})
= \cP_k(\rho\ox\pi_\cE^{\ox k})\cdot\tr(\pi_\cE) = \cE(\rho)$.
Submultiplicativity of the diamond norm gives
$\|\widehat\cP_{k+1}\|_\diamond\le\|\cP_k\|_\diamond$.
Taking the infimum over feasible $\cP_k$ proves the claim.
\end{proof}

The following conversion places the comparison with probabilistic
retrieval on the same Choi-program ensemble.

\begin{lemma}[Uniform-success probabilistic retrieval]
\label{lem:probabilistic_to_quasi}
Let $\varnothing\ne\cS\subseteq\Op{CPTP}_d$, let $k\geq1$, and let
\begin{equation*}
\cN_{\rm succ}\colon
\cB(\cH_S\ox\cH_P^{\ox k})\longrightarrow\cB(\cH_{S'})
\end{equation*}
be completely positive and trace nonincreasing. Suppose that a fixed
$q\in(0,1]$, independent of the signal input and target channel,
satisfies
\begin{equation*}
\cN_{\rm succ}(\rho\ox\pi_\cE^{\ox k})=q\,\cE(\rho)
\end{equation*}
for every $\rho\in\cD(\cH_S)$ and every $\cE\in\cS$. Then
$\nu_k(\cS)\leq 2/q-1$.
\end{lemma}
\begin{proof}
Set $\cH_{\rm in}:=\cH_S\ox\cH_P^{\ox k}$ and fix
$\tau\in\cD(\cH_{S'})$. Trace nonincrease gives
\begin{equation*}
F_{\rm fail}
:=I_{\rm in}-\cN_{\rm succ}^\dagger(I_{S'})\geq0.
\end{equation*}
Define
\begin{equation*}
\cN_{\rm fill}(X):=\tr(F_{\rm fail}X)\tau,
\qquad
\cQ_+:=\cN_{\rm succ}+\cN_{\rm fill},
\qquad
\cQ_-(X):=\tr(X)\tau.
\end{equation*}
The map $\cN_{\rm fill}$ is completely positive and
$\cQ_+^\dagger(I_{S'})=\cN_{\rm succ}^\dagger(I_{S'})+F_{\rm fail}
=I_{\rm in}$. Hence $\cQ_+$ and $\cQ_-$ are quantum channels. Define
\begin{equation*}
\cP_q:=\frac{1}{q}\cQ_+-\frac{1-q}{q}\cQ_-.
\end{equation*}
The two coefficients differ by one, so $\cP_q$ is trace preserving and
Hermiticity preserving. On every valid program input
$\sigma=\rho\ox\pi_\cE^{\ox k}$, the assumption gives
$\tr\cN_{\rm succ}(\sigma)=q$ and therefore
$\tr(F_{\rm fail}\sigma)=1-q$. It follows that
$\cQ_+(\sigma)=q\cE(\rho)+(1-q)\tau$ and
$\cQ_-(\sigma)=\tau$, which yields $\cP_q(\sigma)=\cE(\rho)$. This is an
exact quasi-quantum retriever with
\begin{equation*}
\|\cP_q\|_\diamond
\leq\frac{1}{q}+\frac{1-q}{q}
=\frac{2}{q}-1.
\end{equation*}
Minimizing over exact retrievers proves the claim.
\end{proof}

Lemma~\ref{lem:kcopy_mono} guarantees that the universal overhead
converges as $k\to\infty$. The theorem that follows identifies the
sharp rate of that convergence.

\begin{theorem}[Sharp fixed-dimension $k$-copy overhead]\label{thm:kcopy_scaling}
For every fixed integer $d\ge2$, under the scope specified at
the beginning of this section,
\begin{equation}\label{eq:kcopy_sharp_limit}
\lim_{k\to\infty}
k\bigl(\nu_k(\Op{CPTP}_d)-1\bigr)
=\frac{d^2-1}{2}.
\end{equation}
Equivalently,
\begin{equation}\label{eq:kcopy_sharp_scaling}
\nu_k(\Op{CPTP}_d)
=1+\frac{d^2-1}{2k}+o(1/k),
\qquad k\to\infty .
\end{equation}
The little-$o$ term is understood pointwise in the fixed
dimension $d$; no joint or uniform $(d,k)$ limit is claimed.
\end{theorem}

The two directions are proved separately because they rest on
independent arguments. The upper bound is constructive, whereas the
lower bound applies to every exact HPTP retriever.

\begin{proposition}[Upper bound from standard PBT]\label{prop:kcopy_upper}
For every fixed $d\ge2$ and every $0<\zeta<1/2$, as
$k\to\infty$,
\begin{equation}\label{eq:kcopy_upper_statement}
\nu_k(\Op{CPTP}_d)
\le
1+\frac{d^2-1}{2k}
+O_{d,\zeta}(k^{-3/2+\zeta}).
\end{equation}
Here the asymptotic notation means that there exist
$C_{d,\zeta}>0$ and $k_0(d,\zeta)\in\mathbb{N}$ such that
\begin{equation*}
\nu_k(\Op{CPTP}_d)
\le 1+\frac{d^2-1}{2k}
+C_{d,\zeta}k^{-3/2+\zeta},
\qquad k\ge k_0(d,\zeta).
\end{equation*}
No uniformity in $d$ or $\zeta$ is asserted.
In particular,
\begin{equation}\label{eq:kcopy_upper_limsup}
\limsup_{k\to\infty}
k\bigl(\nu_k(\Op{CPTP}_d)-1\bigr)
\le \frac{d^2-1}{2}.
\end{equation}
\end{proposition}

\begin{proof}
Fix the standard deterministic port-based teleportation (PBT) protocol
with $k$ maximally entangled ports and the complete pretty-good
measurement (PGM). Let $A_i$ and $B_i$ be,
respectively, Alice's and Bob's halves of the $i$-th port. Write
\[
\sigma_i
:=
(\Phi_d)_{S A_i}\ox\frac{I_{A_{\ne i}}}{d^{k-1}},
\qquad
\Sigma_{\rm PGM}:=\sum_{i=1}^k\sigma_i ,
\]
where $\Phi_d=\Omega_d/d$ is the normalized maximally entangled state
and $A_{\ne i}$ denotes all of Alice's port registers except $A_i$.
If $\Pi_\Sigma$ is the support projector of $\Sigma_{\rm PGM}$, the
complete PGM is
\[
M_i
:=
\Sigma_{\rm PGM}^{-1/2}\sigma_i\Sigma_{\rm PGM}^{-1/2}
+\frac{I-\Pi_\Sigma}{k},
\qquad i=1,\ldots,k,
\]
where the inverse is taken on the support of $\Sigma_{\rm PGM}$. Indeed,
$\sum_i\Sigma_{\rm PGM}^{-1/2}\sigma_i
\Sigma_{\rm PGM}^{-1/2}=\Pi_\Sigma$, and hence
$\sum_iM_i=I$. The added term is supported on
$\ker\Sigma_{\rm PGM}$, so it is
orthogonal to every $\sigma_i$ and has zero contribution to the PGM
state-discrimination score. Let $F_d^{\rm std}(k)$ denote the
entanglement fidelity of the resulting standard-PBT channel. In the
notation of Christandl \emph{et al.},
their discrimination formulation~\cite[Sec.~3.1, Eq.~(3.1) and the
following text]{Christandl2021asymptotic} gives
\begin{equation}
\frac{1}{k}\sum_{i=1}^k\tr(M_i\sigma_i)
=\frac{d^2}{k}F^{\rm std}_d(k).
\end{equation}
The completion therefore leaves the standard-PBT entanglement fidelity
unchanged. It also preserves the deterministic protocol because every outcome
selects one port, after which Bob retains that port, discards the others,
and applies no correction~\cite[Sec.~3]{Christandl2021asymptotic}.

Let $\cT_{k,d}$ be this quantum teleportation channel, including the
classical port relabeling. The channel induced on the signal by
maximally entangled ports is
\begin{equation}
\Lambda_{k,d}(\rho):=\cT_{k,d}(\rho\ox\Phi_d^{\ox k}),
\end{equation}
and it is $\Op{SU}(d)$-covariant. To see this directly, fix
$V\in\Op{SU}(d)$. Every $\sigma_i$, and therefore
$\Sigma_{\rm PGM}$, $\Pi_\Sigma$, and $M_i$, is invariant under $V$
on $S$ and $\overline V$ on every
$A$-register, where $\overline V$ denotes entrywise complex
conjugation in the basis defining $\Phi_d$. Moreover, each port state
$(\Phi_d)_{A_jB_j}$ is invariant under $\overline V$ on $A_j$ and $V$ on
$B_j$. Moving these conjugations through the measurement and the port
relabeling gives
\[
\Lambda_{k,d}(V\rho V^\dagger)
=V\Lambda_{k,d}(\rho)V^\dagger .
\]
Since
$\cB(\CC^d)=\CC I_d\oplus\{X:\tr X=0\}$ and the traceless summand is
irreducible over $\CC$ under conjugation by $\Op{SU}(d)$, Schur's lemma
applies separately to these two inequivalent summands. Covariance gives
$V\Lambda_{k,d}(I_d)V^\dagger=\Lambda_{k,d}(I_d)$ for every
$V\in\Op{SU}(d)$, so $\Lambda_{k,d}(I_d)$ is proportional to $I_d$.
Trace preservation fixes $\Lambda_{k,d}(I_d)=I_d$. The restriction to
the traceless summand is multiplication by a scalar, and consequently
\begin{equation}
\Lambda_{k,d}=\cD_{\eta_{k,d}},
\qquad
\cD_\eta(X):=\eta X+(1-\eta)\tr(X)\frac{I_d}{d},
\end{equation}
for some scalar $\eta_{k,d}$. Since $\Lambda_{k,d}$ is completely
positive, it is Hermitian-preserving. Applying
$\Lambda_{k,d}(X)=\eta_{k,d}X$ to a nonzero traceless Hermitian $X$
shows that $\eta_{k,d}$ is real. The standard-PBT asymptotic theorem of
Christandl \emph{et al.} applies to this PGM protocol with maximally
entangled resources~\cite[Theorem~1.2]{Christandl2021asymptotic}.
In their convention, $F^{\rm std}_d(k)$ is the entanglement fidelity
obtained by applying the induced channel to one half of the normalized
maximally entangled state. Their theorem states that, for fixed $d$ and
for every $\zeta>0$,
\begin{equation}\label{eq:std_pbt_fidelity}
F^{\rm std}_d(k)
=1-\frac{d^2-1}{4k}
+O_{d,\zeta}(k^{-3/2+\zeta}).
\end{equation}
With the entanglement fidelity defined above,
$F^{\rm std}_d(k)=F_{\rm e}(\Lambda_{k,d})
=F_{\rm e}(\cD_{\eta_{k,d}})$. Since
$F_{\rm e}(\cD_\eta)=\eta+(1-\eta)/d^2$, Eq.~\eqref{eq:std_pbt_fidelity}
gives
\begin{equation}\label{eq:eta_asymptotic}
1-\eta_{k,d}
=\frac{d^2}{4k}
+O_{d,\zeta}(k^{-3/2+\zeta}),
\end{equation}
for every fixed $d$ and every $\zeta>0$.

Now replace each maximally entangled port by the Choi program
state $\pi_\cE=(\cI\ox\cE)(\Phi_d)$. Write
$\cT_{k,d}=\sum_i\cT_i$ for the trace-nonincreasing maps associated
with the $k$ port outcomes specified above. Because there is no
branch-dependent correction, the $i$-th branch only keeps the selected
output port $B_i$ and traces out the other $B$-registers. Trace
preservation implies the following linear identity. If $R$ is an
arbitrary auxiliary register and $B'_j$ is the output of $\cE$ acting
on $B_j$, then every
$Y\in\cB(\cH_R\ox\cH_{B_1}\ox\cdots\ox\cH_{B_k})$ satisfies
\[
\tr_{B'_{\ne i}}
\left[
\left(\cI_R\ox\bigotimes_{j=1}^k\cE_{B_j\to B'_j}\right)(Y)
\right]
=
\left(\cI_R\ox\cE_{B_i\to B'_i}\right)
\!\left(\tr_{B_{\ne i}}Y\right).
\]
The measurement acts only on $S A_1\cdots A_k$, so the channel actions
on the $B$-registers may be commuted through the measurement. After
relabeling $B'_i$ as the output, the preceding identity gives, for
every $X\in\cB(\cH_S)$,
\begin{equation}
\cT_i(X\ox\pi_\cE^{\ox k})
=\cE\bigl(\cT_i(X\ox\Phi_d^{\ox k})\bigr),
\end{equation}
and summing over $i$ yields
\begin{equation}
\cT_{k,d}(X\ox\pi_\cE^{\ox k})
=\cE\!\left(\cD_{\eta_{k,d}}(X)\right).
\end{equation}
Because $0<\zeta<1/2$, the remainder in
Eq.~\eqref{eq:eta_asymptotic} is $o(1/k)$, and hence
\[
1-\eta_{k,d}
=\frac{d^2}{4k}\bigl(1+o(1)\bigr).
\]
It follows that $0<\eta_{k,d}<1$ for all sufficiently large $k$.
Restricting to such $k$ suffices for this asymptotic proposition. Thus
the inverse $\cD_{\eta_{k,d}}^{-1}=\cD_{\eta_{k,d}^{-1}}$ is HPTP, and
\begin{equation}
\widetilde{\cP}_{k,d}
:=
\cT_{k,d}\circ
(\cD_{\eta_{k,d}}^{-1}\ox\cI_{P^{\ox k}})
\end{equation}
is HPTP. For every $X\in\cB(\cH_S)$ and every
$\cE\in\Op{CPTP}_d$,
\begin{align}
\widetilde{\cP}_{k,d}(X\ox\pi_\cE^{\ox k})
&=\cT_{k,d}\!\left(\cD_{\eta_{k,d}}^{-1}(X)
\ox\pi_\cE^{\ox k}\right)\\
&=\cE\!\left((\cD_{\eta_{k,d}}
\circ\cD_{\eta_{k,d}}^{-1})(X)\right)
=\cE(X),
\end{align}
which proves exact programming. Since $\cT_{k,d}$ is a quantum channel,
$\|\cT_{k,d}\|_\diamond=1$. Submultiplicativity and stability of the
diamond norm under tensoring with an identity map give
\begin{equation}
\|\widetilde{\cP}_{k,d}\|_\diamond
\le \|\cT_{k,d}\|_\diamond
\|\cD_{\eta_{k,d}}^{-1}\ox\cI_{P^{\ox k}}\|_\diamond
=\|\cD_{\eta_{k,d}}^{-1}\|_\diamond.
\end{equation}
Therefore,
\begin{equation}\label{eq:kcopy_upper_pre_norm}
\nu_k(\Op{CPTP}_d)
\le
\|\cD_{\eta_{k,d}}^{-1}\|_\diamond .
\end{equation}

It remains to evaluate the norm in
Eq.~\eqref{eq:kcopy_upper_pre_norm} to first order. Put
$t=\eta_{k,d}^{-1}>1$. The inverse depolarizing map is
$\cD_t$. Its diamond norm is exactly
\begin{equation}\label{eq:inverse_depol_norm}
\|\cD_t\|_\diamond
=1+2(1-d^{-2})(t-1).
\end{equation}
Indeed, the lower bound follows by applying $\cI\ox\cD_t$ to
$\Phi_d$. The resulting normalized Choi operator has one eigenvalue
$1+(t-1)(1-d^{-2})$ and $d^2-1$ negative eigenvalues
$-(t-1)/d^2$, so its trace norm is the right-hand side of
\eqref{eq:inverse_depol_norm}. For the reverse inequality, set
$b=(t-1)(d^2-1)/d^2$. Since both $\cD_1$ and
$\cD_{-1/(d^2-1)}$ are depolarizing quantum channels,
\begin{equation}
\cD_t=(1+b)\cD_1-b\,\cD_{-1/(d^2-1)}
\end{equation}
implies $\|\cD_t\|_\diamond\le 1+2b$, which is
Eq.~\eqref{eq:inverse_depol_norm}. Since $\eta_{k,d}\to1$,
Eq.~\eqref{eq:eta_asymptotic} gives
\begin{equation}
\begin{aligned}
t-1
&=\frac{1-\eta_{k,d}}{\eta_{k,d}}\\
&=(1-\eta_{k,d})
+\frac{(1-\eta_{k,d})^2}{\eta_{k,d}}\\
&=\frac{d^2}{4k}
+O_{d,\zeta}(k^{-3/2+\zeta}).
\end{aligned}
\end{equation}
The quadratic term is $O_{d,\zeta}(k^{-2})$ and is absorbed by the
displayed remainder. Substitution into Eq.~\eqref{eq:inverse_depol_norm}
gives
\begin{equation}\label{eq:kcopy_upper_final}
\nu_k(\Op{CPTP}_d)
\le
1+\frac{d^2-1}{2k}
+O_{d,\zeta}(k^{-3/2+\zeta}),
\end{equation}
which proves both claims.
\end{proof}

The lower bound uses the following fixed-memory consequence of the
retrieval theorem of Bisio \emph{et al.}~\cite{Bisio2010optimal}.

\begin{lemma}[Fixed-memory form of the retrieval theorem]
\label{lem:bisio_fixed_memory_reduction}
For $U,\widehat U\in\Op{SU}(d)$, let
$\cU:=\Op{Ad}_U$ and
$\widehat{\cU}:=\Op{Ad}_{\widehat U}$. Define
\begin{equation}\label{eq:kcopy_unitary_loss}
\mathcal L(U,\widehat U)
:=
1-F_{\rm av}(\widehat {\cU},\cU)
=
\frac{d^2-|\tr(U^\dagger\widehat U)|^2}{d(d+1)}.
\end{equation}
Fix $k$ and let $U_\lambda$ denote the irreducible blocks of the
representation $U\mapsto U^{\ox k}$, acting on carrier spaces
$\cH_\lambda$ of dimensions $d_\lambda$. For any probability distribution
$(p_\lambda)_\lambda$, define the canonical memory state on
$\cH_M:=\bigoplus_\lambda(\cH_\lambda\ox\cH_\lambda)$ by
\begin{equation*}
\ket{\phi_U}
:=
\bigoplus_\lambda
\sqrt{\frac{p_\lambda}{d_\lambda}}\,
\dket{U_\lambda},
\qquad
\dket{U_\lambda}:=(U_\lambda\ox I_{d_\lambda})\dket{I_{d_\lambda}}.
\end{equation*}
The conclusion below holds pointwise for every fixed
$(p_\lambda)_\lambda$; no optimization over the memory weights is
taken.
Let $\cG\in\Op{CPTP}(\cH_S\ox\cH_M\to\cH_{S'})$ be any physical
learning channel and set
$\cG^U(\rho):=\cG(\rho\ox\ketbra{\phi_U}{\phi_U})$.
There exists a POVM $M(d\widehat U)$ on $\cH_M$, with outcome
$\widehat U\in\Op{SU}(d)$, such that
\begin{equation}\label{eq:bisio_fixed_memory_reduction}
\int_{\Op{SU}(d)}
\left[1-F_{\rm av}(\cG^U,\cU)\right]\,dU
\ge
\int_{\Op{SU}(d)}r_M(U)\,dU,
\end{equation}
where
\begin{equation*}
r_M(U)
:=
\int \mathcal L(U,\widehat U)\,
\tr\!\left[
M(d\widehat U)\ketbra{\phi_U}{\phi_U}
\right].
\end{equation*}
All Haar measures are normalized.
\end{lemma}

\begin{proof}
Fix $(p_\lambda)_\lambda$. Let $U_{\rm f}$ and $V_{\rm f}$ denote two
independent copies of the defining representation. For each $\lambda$,
decompose
\begin{equation*}
U_{\rm f}\ox U_\lambda^*
\simeq\bigoplus_K U_K\ox I_{\cM_K^{(\lambda)}},
\qquad
V_{\rm f}^*\ox V_\lambda
\simeq\bigoplus_L V_L^*\ox I_{\cM_L^{(\lambda)}},
\end{equation*}
where $U_K$ and $V_L^*$ act on carrier spaces $\cH_K$ and $\cH_L$,
respectively. Write $d_K:=\dim\cH_K$ and
$m_K^{(\lambda)}:=\dim\cM_K^{(\lambda)}$. The same multiplicities occur
in the second decomposition because it is the complex-conjugate
counterpart of the first. Set
\begin{equation*}
\mathcal P_{KL}:=
\{\lambda\mid m_K^{(\lambda)}m_L^{(\lambda)}>0\}.
\end{equation*}
Twirling a learning channel under the independent input and output
group actions preserves its average entanglement fidelity.
With the memory weights held fixed, in particular,
\begin{equation*}
\int F_{\rm e}(\Op{Ad}_{U^\dagger}\circ\cG^U)\,dU
=\int F_{\rm e}(\Op{Ad}_{U^\dagger}\circ\cG_{\rm tw}^U)\,dU.
\end{equation*}
Schur's
lemma then decomposes the twirled Choi operator as
\begin{equation*}
J_{\cG_{\rm tw}}\simeq\bigoplus_{K,L} I_K\ox I_L\ox R_{KL},
\end{equation*}
where $R_{KL}\geq0$ acts on
$\bigoplus_{\lambda\in\mathcal P_{KL}}
(\cM_K^{(\lambda)}\ox\cM_L^{(\lambda)})$. Let
$R_{KL}^{(\lambda)}$ be its
compression to the summand indexed by $\lambda$. Trace preservation is
equivalent to the block identities
\begin{equation}\label{eq:bisio_fixed_p_trace_identity}
I_{\cM_L^{(\lambda)}}
=\sum_{K:\,\lambda\in\mathcal P_{KL}}
\frac{d_K}{d_\lambda}
\tr_{\cM_K^{(\lambda)}}
R_{KL}^{(\lambda)}
\end{equation}
for every $L$ and every $\lambda$ with $m_L^{(\lambda)}>0$. Every
summand is positive. Fixing the second index to $K$, taking the trace,
and retaining the term whose first index is also $K$ gives
\begin{equation*}
m_K^{(\lambda)}
=\sum_{K':\,\lambda\in\mathcal P_{K'K}}
\frac{d_{K'}}{d_\lambda}\tr R_{K'K}^{(\lambda)}
\geq\frac{d_K}{d_\lambda}\tr R_{KK}^{(\lambda)}.
\end{equation*}
Hence
\begin{equation}\label{eq:bisio_fixed_p_trace_bound}
\tr R_{KK}^{(\lambda)}
\leq\frac{d_\lambda m_K^{(\lambda)}}{d_K}.
\end{equation}

Writing $\dket{I_m}$ for the unnormalized maximally entangled vector on
a multiplicity space, define
\begin{equation*}
\ket{\alpha_K}
:=\bigoplus_{\lambda\in\mathcal P_{KK}}
\sqrt{\frac{p_\lambda}{d_\lambda}}\,
\dket{I_{m_K^{(\lambda)}}}.
\end{equation*}
Only sectors with equal coupled labels $K=L$ contribute to the
averaged entanglement fidelity; the coherent off-diagonal
$\lambda,\lambda'$ blocks inside each $R_{KK}$ remain included.
The block decomposition therefore gives the exact identity
\begin{equation*}
\overline F_{\rm e}(\cG;p)
:=\int F_{\rm e}(\Op{Ad}_{U^\dagger}\circ\cG^U)\,dU
=\frac{1}{d^2}\sum_K d_K
\bra{\alpha_K}R_{KK}\ket{\alpha_K}.
\end{equation*}
Positivity of $R_{KK}$ gives a Cauchy--Schwarz bound on its
off-diagonal compressions. Together with
$\langle\!\langle I_m|X|I_m\rangle\!\rangle\leq m\tr X$ for $X\geq0$
and Eq.~\eqref{eq:bisio_fixed_p_trace_bound}, this yields
\begin{equation}\label{eq:bisio_fixed_p_block_bound}
\begin{aligned}
\bra{\alpha_K}R_{KK}\ket{\alpha_K}
&\leq
\left(\sum_{\lambda:\,m_K^{(\lambda)}>0}
\sqrt{\frac{p_\lambda}{d_\lambda}}
\sqrt{\langle\!\langle I_{m_K^{(\lambda)}}|
R_{KK}^{(\lambda)}
|I_{m_K^{(\lambda)}}\rangle\!\rangle}
\right)^2\\
\overline F_{\rm e}(\cG;p)
&\leq\frac{1}{d^2}\sum_K
\left(\sum_{\lambda\in\mathcal P_{KK}}
m_K^{(\lambda)}\sqrt{p_\lambda}\right)^2
=:F_{\rm est}(p).
\end{aligned}
\end{equation}

The final expression is attained without changing the probabilities.
For $\widehat U\in\Op{SU}(d)$, define
\begin{equation*}
\ket{\xi_{\widehat U}}
:=\bigoplus_\lambda\sqrt{d_\lambda}\,
\dket{\widehat U_\lambda},
\qquad
M(d\widehat U):=\ketbra{\xi_{\widehat U}}{\xi_{\widehat U}}\,d\widehat U.
\end{equation*}
Schur orthogonality gives $\int M(d\widehat U)=I_{\cH_M}$. Measuring this
POVM and applying $\widehat U$ to the signal defines the physical
measure-and-rotate channel
\begin{equation*}
\cG_{\rm est}(X)
:=\int \widehat U\,
\tr_M\!\left[(I_S\ox
\ketbra{\xi_{\widehat U}}{\xi_{\widehat U}})X\right]
\widehat U^\dagger\,d\widehat U.
\end{equation*}
A second use of Schur orthogonality gives its average entanglement
fidelity as $F_{\rm est}(p)$. Hence every physical learning channel
satisfies
\begin{equation}\label{eq:bisio_fixed_fidelity_bound}
\int F_{\rm e}(\Op{Ad}_{U^\dagger}\circ\cG^U)\,dU
\le
\iint
\frac{|\tr(U^\dagger\widehat U)|^2}{d^2}
\tr\!\left[M(d\widehat U)\ketbra{\phi_U}{\phi_U}\right]dU.
\end{equation}
No supremum or optimization over $(p_\lambda)_\lambda$ has been taken.
The block decomposition, trace bound, and attaining POVM are the
fixed-memory steps underlying Eqs.~(12)--(25) of
Ref.~\cite{Bisio2010optimal}.

Equation~\eqref{eq:average_gate_fidelity} gives
\begin{equation*}
1-F_{\rm av}(\cG^U,\cU)
=\frac{d}{d+1}
\left[1-F_{\rm e}(\Op{Ad}_{U^\dagger}\circ\cG^U)\right]
\end{equation*}
and
\begin{equation*}
\mathcal L(U,\widehat U)
=\frac{d}{d+1}
\left[1-\frac{|\tr(U^\dagger\widehat U)|^2}{d^2}\right].
\end{equation*}
Taking complements in Eq.~\eqref{eq:bisio_fixed_fidelity_bound} and
using $\int M(d\widehat U)=I_{\cH_M}$ proves
Eq.~\eqref{eq:bisio_fixed_memory_reduction}.
\end{proof}

Combining this fixed-memory reduction with a local Bayesian estimate
gives the matching lower bound.

\begin{proposition}[Lower bound from fixed-Choi learning]
\label{prop:kcopy_lower}
For every fixed $d\ge2$,
\begin{equation*}
\liminf_{k\to\infty}
k\bigl(\nu_k(\Op{CPTP}_d)-1\bigr)
\ge \frac{d^2-1}{2}.
\end{equation*}
\end{proposition}

\begin{proof}
Fix a positive integer $k$ and let $\cP$ be any feasible exact
$k$-copy HPTP retriever. The trace-preserving case of Theorem~3 of
Regula \emph{et al.} provides quantum channels $\cQ_\pm$ and
coefficients $p_\pm\ge0$ such that~\cite[Theorem~3 and
Eq.~(22)]{Regula2021operational}
\begin{equation*}
\cP=p_+\cQ_+-p_-\cQ_-,
\qquad
p_++p_-=\|\cP\|_\diamond .
\end{equation*}
Since all three maps are trace preserving, taking the trace gives
$p_+-p_-=1$ and hence
$2p_-=\|\cP\|_\diamond-1$. For $U\in\Op{SU}(d)$, define the
normalized Choi vector
$\ket{p_U}:=d^{-1/2}(I\ox U)\dket{I_d}$, so that
$\pi_{\cU}=\ketbra{p_U}{p_U}$. The channels induced by $\cQ_\pm$ at
this program are
\begin{equation*}
\cQ_{\pm}^U(\rho):=\cQ_{\pm}(\rho\ox\ketbra{p_U}{p_U}^{\ox k}).
\end{equation*}
The state-insertion map
$\rho\mapsto\rho\ox\ketbra{p_U}{p_U}^{\ox k}$ is CPTP. Hence each
$\cQ_\pm^U$ is a channel induced by the same target-independent
physical learner $\cQ_\pm$.
Exact programmability of $\cP$ gives
$\cU=p_+\cQ_+^U-p_-\cQ_-^U$.
Using $p_+-p_-=1$, this is equivalently
\begin{equation*}
\cQ_+^U-\cU
=p_-\bigl(\cQ_-^U-\cQ_+^U\bigr).
\end{equation*}
Since both induced maps are quantum channels,
\begin{equation}\label{eq:kcopy_positive_part_delta}
\|\cQ_+^U-\cU\|_\diamond
\le p_-\bigl(\|\cQ_-^U\|_\diamond+\|\cQ_+^U\|_\diamond\bigr)
=2p_-= \|\cP\|_\diamond - 1.
\end{equation}

The entanglement fidelity of
$\Op{Ad}_{U^\dagger}\circ\cQ_+^U$ is
$\langle p_U|(\cI\ox\cQ_+^U)(\Phi_d)|p_U\rangle$. The pure-state
trace-distance inequality and
Eq.~\eqref{eq:kcopy_positive_part_delta} imply
\begin{equation*}
1-F_{\rm e}(\Op{Ad}_{U^\dagger}\circ\cQ_+^U)
\le
\frac{1}{2}
\|\bigl(\cI\ox(\cQ_+^U-\cU)\bigr)(\Phi_d)\|_1
\le
\frac{1}{2}\|\cQ_+^U-\cU\|_\diamond
\le \frac{1}{2}(\|\cP\|_\diamond - 1).
\end{equation*}
Equation~\eqref{eq:average_gate_fidelity} then gives
\begin{equation*}
1-F_{\rm av}(\cQ_+^U,\cU)
\le
\frac{d}{2(d+1)}(\|\cP\|_\diamond - 1).
\end{equation*}

For a physical learning channel
$\cG\in\Op{CPTP}(\cH_S\ox\cH_P^{\ox k}\to\cH_{S'})$, write
\begin{equation*}
\cG^U(\rho)
:=
\cG(\rho\ox\ketbra{p_U}{p_U}^{\ox k}).
\end{equation*}
Define the optimal Haar-averaged learning risk from this fixed Choi
memory by
\begin{equation}\label{eq:Rav_le_delta}
R_k^{\rm av}
:=
\inf_{\cG\in\Op{CPTP}(\cH_S\ox\cH_P^{\ox k}\to\cH_{S'})}
\int_{\Op{SU}(d)}
\left[1-F_{\rm av}(\cG^U,\cU)\right]\,dU
\le
\frac{d}{2(d+1)}(\|\cP\|_\diamond-1).
\end{equation}

To compare this learning problem with
Lemma~\ref{lem:bisio_fixed_memory_reduction}, choose a fixed
Schur unitary
\begin{equation*}
S_k:(\CC^d)^{\ox k}\longrightarrow
\cK_k:=\bigoplus_\lambda(\cH_\lambda\ox\cM_\lambda)
\end{equation*}
such that
\begin{equation}\label{eq:product_choi_schur_decomposition}
S_kU^{\ox k}S_k^\dagger
=\bigoplus_\lambda U_\lambda\ox I_{m_\lambda},
\qquad m_\lambda:=\dim\cM_\lambda,
\end{equation}
where $U_\lambda$ acts on $\cH_\lambda$ and
$d_\lambda:=\dim\cH_\lambda$. Let $\Pi_k$ flip the two halves of every
Choi pair and reorder the registers as $B^{\ox k}A^{\ox k}$. With the
double-ket convention
$\dket{X}:=(X\ox I)\dket{I}$ used in
Lemma~\ref{lem:bisio_fixed_memory_reduction},
\begin{equation*}
\Pi_k\ket{p_U}^{\ox k}
=d^{-k/2}\dket{U^{\ox k}}_{B^{\ox k}:A^{\ox k}}.
\end{equation*}
On these ordered registers apply
$S_{k,B}\ox\overline{S}_{k,A}$, followed by a fixed regrouping of the
carrier and multiplicity factors in each block; call the resulting
unitary $W_k$. The identity
$(A\ox B)\dket{X}=\dket{AXB^T}$ shows both that
$(S_k\ox\overline{S}_k)\dket{I}=\dket{I}$ and that the transformed
state is supported on the $\lambda=\mu$ diagonal sector of
$\cK_k^B\ox\cK_k^A$. Consequently,
\begin{equation}
W_k\ket{p_U}^{\ox k}
=
d^{-k/2}\bigoplus_\lambda
\dket{U_\lambda}\ox\dket{I_{m_\lambda}}
=
\bigoplus_\lambda
\sqrt{\frac{p_\lambda}{d_\lambda}}\,
\dket{U_\lambda}\ox\ket{\eta_\lambda},
\qquad
p_\lambda:=\frac{d_\lambda m_\lambda}{d^k},
\qquad
\ket{\eta_\lambda}:=\frac{\dket{I_{m_\lambda}}}{\sqrt{m_\lambda}}.
\end{equation}
The dimension identity $\sum_\lambda d_\lambda m_\lambda=d^k$
shows that $\sum_\lambda p_\lambda=1$.
Let $\cH_M$ be the canonical memory space in
Lemma~\ref{lem:bisio_fixed_memory_reduction}, and define the isometry
$V_k$ by
\begin{equation*}
V_k\!\left(\bigoplus_\lambda\ket{\psi_\lambda}\right)
:=
\bigoplus_\lambda\ket{\psi_\lambda}\ox\ket{\eta_\lambda},
\qquad
\ket{\psi_\lambda}\in\cH_\lambda\ox\cH_\lambda .
\end{equation*}
If $\ket{\phi_U}:=\bigoplus_\lambda
\sqrt{p_\lambda/d_\lambda}\,\dket{U_\lambda}$, then
$W_k\ket{p_U}^{\ox k}=V_k\ket{\phi_U}$.
Define channels on the Schur-output space by
\begin{equation*}
\begin{aligned}
\mathsf E_k(X)
&:=V_kXV_k^\dagger,\\
\mathsf D_k(Y)
&:=V_k^\dagger YV_k
+\tr\!\left[(I-V_kV_k^\dagger)Y\right]\tau_0,
\end{aligned}
\end{equation*}
where $\tau_0$ is any fixed state on $\cH_M$. Both maps are
CPTP and $\mathsf D_k\circ\mathsf E_k=\cI_M$. Including
$W_k$, set
\begin{equation*}
\mathsf T_{M\to P}
:=\Op{Ad}_{W_k^\dagger}\circ\mathsf E_k,
\qquad
\mathsf T_{P\to M}
:=\mathsf D_k\circ\Op{Ad}_{W_k}.
\end{equation*}
These target-independent channels satisfy, for every $U$,
\begin{equation*}
\begin{aligned}
\mathsf T_{M\to P}(\ketbra{\phi_U}{\phi_U})
&=\ketbra{p_U}{p_U}^{\ox k},\\
\mathsf T_{P\to M}(\ketbra{p_U}{p_U}^{\ox k})
&=\ketbra{\phi_U}{\phi_U},\\
\mathsf T_{P\to M}\circ\mathsf T_{M\to P}
&=\cI_M.
\end{aligned}
\end{equation*}
For any product-memory learning channel $\cG$, the channel
$\cG\circ(\cI_S\ox\mathsf T_{M\to P})$ is therefore a
canonical-memory learner with identical output channels on every valid
program. Lemma~\ref{lem:bisio_fixed_memory_reduction} applies with the
fixed probabilities $p_\lambda$ above. If $M_{\rm can}(d\widehat U)$
is its POVM on $\cH_M$, then
\begin{equation*}
M_{\rm Choi}(d\widehat U)
:=\mathsf T_{P\to M}^{*}
\bigl(M_{\rm can}(d\widehat U)\bigr)
\end{equation*}
is a POVM on the original product memory $\cH_P^{\ox k}$ with
the same outcome probabilities. Conversely,
$\mathsf T_{M\to P}^{*}$ transports every product-memory POVM to the
canonical memory with the same model statistics. Thus the two POVM
infima agree. For a product-memory POVM, define
\begin{equation*}
r_{M_{\rm Choi}}(U)
=
\int \mathcal L(U,\widehat U)\,
\tr\!\left[
M_{\rm Choi}(d\widehat U)\ketbra{p_U}{p_U}^{\ox k}
\right].
\end{equation*}
Taking the infimum over physical learning channels gives
\begin{equation}\label{eq:Rav_estimation}
R_k^{\rm av}
\ge
\inf_{M_{\rm Choi}}\int_{\Op{SU}(d)}r_{M_{\rm Choi}}(U)\,dU,
\end{equation}
where the infimum is over all POVMs on the original product
Choi memory. Hereafter, we abbreviate $M_{\rm Choi}$ and
$r_{M_{\rm Choi}}$ by $M$ and $r_M$, respectively.

The infimum in Eq.~\eqref{eq:Rav_estimation} is unchanged when
restricted to covariant POVMs. For any POVM $M$ and Borel set
$B\subseteq\Op{SU}(d)$, define its Haar twirl by
\[
\overline M(B)
:=
\int_{\Op{SU}(d)}
(I\ox V^\dagger)^{\ox k}M(VB)(I\ox V)^{\ox k}\,dV.
\]
Here $VB:=\{V\widehat U:\widehat U\in B\}$.
Finite-dimensional Haar integration preserves positivity and
normalization, so $\overline M$ is a POVM. Its pointwise risk is
\begin{equation}\label{eq:kcopy_twirl_risk}
r_{\overline M}(U)
=\int r_M(VU)\,dV
=\int r_M(U')\,dU'.
\end{equation}
The equality follows from
$\mathcal L(U,V^{-1}\widehat U)=\mathcal L(VU,\widehat U)$ and Haar
invariance. Hence $r_{\overline M}(U)$ is independent of $U$ and equals
the Haar-averaged risk of $M$.

A local chart around the identity channel provides the required lower
bound. Let $n_{\rm par}:=d^2-1$, and choose
traceless Hermitian matrices $T_1,\ldots,T_{n_{\rm par}}$ satisfying
$\tr(T_aT_b)=\delta_{ab}$. For
$x=(x_1,\ldots,x_{n_{\rm par}})\in\RR^{n_{\rm par}}$ near $0$, define
\begin{equation*}
H_x:=\sum_{a=1}^{n_{\rm par}} x_aT_a,
\qquad
U_x:=e^{iH_x},
\qquad
\cU_x:=\Op{Ad}_{U_x},
\qquad
\ket{p_x}:=\ket{p_{U_x}}.
\end{equation*}
Here $\partial_a:=\partial/\partial x_a$. We use $\|\cdot\|_{\rm HS}$
for the Hilbert--Schmidt norm of an operator and $\|\cdot\|_2$ for the
Euclidean norm of a coordinate vector.
Let $\mathcal J(x)$ be the one-copy symmetric-logarithmic-derivative
quantum Fisher information matrix in these coordinates. Related
formulations of pure-state quantum Fisher geometry, multiparameter bounds,
and R\'enyi-based QFI matrices are given in
Refs.~\cite{fujiwara1995quantum,Matsumoto2002pure,Yang2019attaining,wildeQuantumFisherInformation2025}.
Its diagonal
entries for this pure-state model are~\cite[Sec.~II.A]{ballester2004estimation}
\begin{equation*}
\mathcal J_{aa}(x)
=4\left(
\langle\partial_a p_x|\partial_a p_x\rangle
-|\langle p_x|\partial_a p_x\rangle|^2
\right).
\end{equation*}
The Choi-state trace identity and $\det U_x=1$ give $\langle p_x|\partial_a p_x\rangle
=(1/d)\tr(U_x^\dagger\partial_aU_x)=0$, whereas
$\langle\partial_a p_x|\partial_a p_x\rangle
=\frac{1}{d}\|U_x^\dagger\partial_aU_x\|_{\rm HS}^2$.
Differentiating the power series term by term with respect to $x_a$ gives
\begin{equation*}
\partial_aU_x
=
i\int_0^1e^{i(1-s)H_x}T_ae^{isH_x}\,ds,
\qquad
U_x^\dagger\partial_aU_x
=i\int_0^1e^{-isH_x}T_ae^{isH_x}\,ds.
\end{equation*}
Unitary conjugation preserves the Hilbert--Schmidt norm and
$\|T_a\|_{\rm HS}=1$. Hence, for every $x$,
\begin{equation}\label{eq:kcopy_uniform_qfi}
\mathcal J_{aa}(x)
=\frac{4}{d}\|U_x^\dagger\partial_aU_x\|_{\rm HS}^2
\le
\frac{4}{d}
\left(\int_0^1\|e^{-isH_x}T_ae^{isH_x}\|_{\rm HS}\,ds\right)^2
=\frac{4}{d}.
\end{equation}
Let $\mathcal J_{aa}^{(k)}(x)$ denote the corresponding diagonal
quantum Fisher information for $\ket{p_x}^{\ox k}$. Additivity on
product states gives
\begin{equation*}
\mathcal J_{aa}^{(k)}(x)=k\mathcal J_{aa}(x)\le\frac{4k}{d}.
\end{equation*}

For fixed $x$, define
$\ell_x(z):=\mathcal L(U_x,U_{x+z})$ for coordinate increments
$z\in\RR^{n_{\rm par}}$. The loss is nonnegative and
$\ell_x(0)=0$, so $\nabla\ell_x(0)=0$. At $x=0$, writing
$H_z:=\sum_a z_aT_a$ gives
\begin{equation*}
\ell_0(z)
=\frac{d^2-|\tr(e^{iH_z})|^2}{d(d+1)}
=\frac{\|z\|_2^2}{d+1}+O(\|z\|_2^3),
\end{equation*}
where tracelessness of the $T_a$ and
$\tr(T_aT_b)=\delta_{ab}$ were used in the expansion. Therefore
\begin{equation}\label{eq:kcopy_loss_hessian}
\nabla^2\ell_0(0)=\frac{2}{d+1}I_{n_{\rm par}},
\end{equation}
where $I_{n_{\rm par}}$ is the identity matrix on
$\RR^{n_{\rm par}}$.

Fix a local slack $0<\alpha<1$, unrelated to the programming
tolerance $\varepsilon$, which is fixed at zero throughout this
section. Continuity of $\nabla^2\ell_x(z)$ and compactness
of the unit sphere give $r_0>0$ such that
\begin{equation}\label{eq:kcopy_loss_hessian_uniform}
v^T\nabla^2\ell_x(z)v
\ge
\frac{2(1-\alpha)}{d+1}\|v\|_2^2
\end{equation}
for all $v\in\RR^{n_{\rm par}}$ whenever
$\|x\|_2,\|z\|_2<r_0$.

The set of unitary channels is the smooth quotient of $\Op{SU}(d)$ by
its finite center, and is therefore an $n_{\rm par}$-dimensional
manifold. The
differential at $x=0$ of
$x\mapsto\cU_x$ sends $h$ to the tangent map
$X\mapsto i[H_h,X]$. Its kernel is zero. If $H_h$ commutes with every
$X$, then $H_h$ is scalar, and tracelessness forces $H_h=0$. The
differential is thus an isomorphism between two
$n_{\rm par}$-dimensional tangent
spaces. The inverse function theorem therefore allows
$0<\rho<r_0/2$ to be chosen so that
$x\mapsto\cU_x$ is a diffeomorphism from $B_\rho$ onto a
neighborhood of the identity channel that is open in the unitary-channel
manifold, where $B_\rho:=\{x:\|x\|_2<\rho\}$.

Let $x,y\in B_\rho$ and set $h:=y-x$. Then
$\|h\|_2<2\rho<r_0$. Taylor's formula with integral remainder, applied
along the segment $t\mapsto th$, gives
\begin{equation*}
\ell_x(h)
=
\int_0^1(1-t)
h^T\nabla^2\ell_x(th)h\,dt.
\end{equation*}
Equation~\eqref{eq:kcopy_loss_hessian_uniform} therefore gives
\begin{equation}\label{eq:kcopy_local_loss_comparison}
\mathcal L(U_x,U_{y})
\ge
\frac{1-\alpha}{d+1}\|y-x\|_2^2,
\qquad x,y\in B_\rho.
\end{equation}

To extend this inequality to every POVM outcome, define
$\mathcal N_\rho:=\{\cU_y:y\in B_\rho\}$. If the estimated
channel belongs to $\mathcal N_\rho$, let $\widehat x=y$ be its unique
coordinate. Otherwise, set $\widehat x=0$. The channel space
$\{\cU:U\in\Op{SU}(d)\}$ is compact, $\mathcal N_\rho$ is an open
neighborhood of the identity channel, and the continuous loss vanishes
only when its two channel arguments coincide. For
$V\in\Op{SU}(d)$, write $\cV:=\Op{Ad}_V$. Consequently,
\begin{equation*}
c_\rho
:=
\min_{\cV\notin\mathcal N_\rho}\mathcal L(I_d,V)
>0.
\end{equation*}
Here $V$ is any unitary representative of the channel $\cV$. By
Eq.~\eqref{eq:kcopy_unitary_loss}, $\mathcal L$ depends on its
arguments only through $|\tr(U^\dagger V)|$, so it is invariant under
central phases and descends to a well-defined function of unitary
channels.
The minimum is attained because the complement of $\mathcal N_\rho$
in the compact channel space is compact and $\mathcal L(I,\cdot)$ is
continuous.
The map from the POVM outcome to $\widehat x$ is Borel measurable because
it is the continuous inverse of the coordinate chart inside
$\mathcal N_\rho$ and is constant outside $\mathcal N_\rho$.
Uniform continuity permits a radius $0<r<\rho$ such that
$\mathcal L(U_x,V)\ge c_\rho/2$ for $x\in B_r$ and
$\cV\notin\mathcal N_\rho$. Decrease $r$ further so that
$(1-\alpha)r^2/(d+1)\le c_\rho/2$. For outcomes inside
$\mathcal N_\rho$, use Eq.~\eqref{eq:kcopy_local_loss_comparison}. For
outcomes outside it, use $\widehat x=0$ and the preceding two bounds.
Thus, for every $x\in B_r$ and every estimate $\widehat U$,
\begin{equation}\label{eq:kcopy_global_loss_comparison}
\mathcal L(U_x,\widehat U)
\ge
\frac{1-\alpha}{d+1}\|\widehat x-x\|_2^2.
\end{equation}

A Bayesian prior now converts the local comparison into an estimation
bound. Let
$w_r(x):=c_r(r^2-\|x\|_2^2)^2$ on $B_r$, where the constant $c_r>0$
normalizes the density. Extend $w_r$ continuously by zero from $B_r$
to the compact closure $\overline B_r$. Its Fisher information in
coordinate $a$ is
\begin{equation}\label{eq:kcopy_prior_fisher}
I_{w_r,a}
:=
\int_{B_r}\frac{(\partial_aw_r(x))^2}{w_r(x)}\,dx
\;=\frac{2(n_{\rm par}+4)}{r^2},
\qquad a=1,\ldots,n_{\rm par},
\end{equation}
where the value follows by radial integration. Both $w_r$ and its
first derivatives vanish on the boundary of $B_r$, and
$I_{w_r,a}$ is finite and independent of $k$.
For a covariant POVM, Eq.~\eqref{eq:kcopy_twirl_risk} makes its
pointwise risk constant. Equations~\eqref{eq:Rav_estimation} and
\eqref{eq:kcopy_twirl_risk} therefore give
\begin{equation}\label{eq:Rav_local_prior}
R_k^{\rm av}
\ge
\inf_{M\ {\rm covariant}}
\int_{B_r}w_r(x)r_M(U_x)\,dx.
\end{equation}
Let $\mathbb{E}_x$ denote expectation over the outcome of $M$ on
$\ketbra{p_x}{p_x}^{\ox k}$. Taking this expectation in
\eqref{eq:kcopy_global_loss_comparison} gives
\begin{equation}\label{eq:kcopy_risk_to_mse}
r_M(U_x)
\ge
\frac{1-\alpha}{d+1}
\mathbb{E}_x\|\widehat x-x\|_2^2,
\qquad x\in B_r.
\end{equation}

Let $\mu(B):=\tr M(B)$ be the finite trace measure of the POVM.
Finite dimensionality gives a positive operator density
$D_M(\widehat U)$ such that
\begin{equation*}
M(d\widehat U)=D_M(\widehat U)\mu(d\widehat U),
\qquad
\tr D_M(\widehat U)=1
\quad\text{for $\mu$-almost every $\widehat U$.}
\end{equation*}
The conditional outcome density at $x$ is
\begin{equation*}
f(\widehat U|x)
:=
\tr\!\left[D_M(\widehat U)\ketbra{p_x}{p_x}^{\ox k}\right].
\end{equation*}
It is normalized with respect to $\mu$ and real analytic in $x$.
Define the $a$th diagonal entry of its classical Fisher information by
\begin{equation*}
I_{M,aa}^{(k)}(x)
:=
\int_{\{f(\widehat U|x)>0\}}
\frac{[\partial_af(\widehat U|x)]^2}{f(\widehat U|x)}
\,\mu(d\widehat U),
\end{equation*}
with the integrand set to zero where $f(\widehat U|x)=0$.
The scalar Braunstein--Caves inequality applied along the coordinate
direction $e_a$, the $a$th standard basis vector of
$\RR^{n_{\rm par}}$, gives~\cite[Eqs.~(17) and~(24)]{BraunsteinCaves1994statistical}
\begin{equation}\label{eq:kcopy_BC_bound}
I_{M,aa}^{(k)}(x)
\le \mathcal J_{aa}^{(k)}(x)
\le\frac{4k}{d}.
\end{equation}
Here $M$ is one collective POVM on all $k$ copies. No product
measurement is assumed.

Boundedness of $D_M(\widehat U)$ and of the state derivatives on
$\overline B_r$ permits differentiation under the outcome integral,
while Eq.~\eqref{eq:kcopy_BC_bound} makes the required Fisher terms
finite. Together with the quadratic boundary zero of $w_r$, these
facts verify Conditions~1--5 of Gill and Levit on $\overline B_r$.
In the notation of their multivariate theorem, take
$\Theta=\overline B_r$, the scalar target $\psi(x)=x_a$,
$B(x)=1$, $C(x)=e_a^T$, and $n=1$, since the outcome of the single
collective $k$-copy POVM is treated as one observation. Their
Theorem~1 then gives the coordinatewise van Trees inequality without
an unbiasedness assumption~\cite[Theorem~1 and
Eqs.~(7)--(8)]{GillLevit1995applications}.
\begin{equation}\label{eq:kcopy_vantrees_scalar}
\int_{B_r}w_r(x)\,
\mathbb{E}_x(\widehat x_a-x_a)^2\,dx
\ge
\left[
I_{w_r,a}
+\int_{B_r}w_r(x)I_{M,aa}^{(k)}(x)\,dx
\right]^{-1}
\ge
\frac{1}{I_{w_r,a}+4k/d}.
\end{equation}
Summing Eq.~\eqref{eq:kcopy_vantrees_scalar} over the
$n_{\rm par}=d^2-1$ coordinates and using
Eq.~\eqref{eq:kcopy_prior_fisher} yields
\begin{equation*}
\int_{B_r}w_r(x)\,
\mathbb{E}_x\|\widehat x-x\|_2^2\,dx
\ge
\frac{n_{\rm par}}{2(n_{\rm par}+4)/r^2+4k/d}.
\end{equation*}
Equations~\eqref{eq:Rav_local_prior} and
\eqref{eq:kcopy_risk_to_mse} therefore imply
\begin{equation}\label{eq:fixed_choi_risk_finite_k}
R_k^{\rm av}
\ge
\frac{1-\alpha}{d+1}
\frac{n_{\rm par}}{2(n_{\rm par}+4)/r^2+4k/d}.
\end{equation}

For every fixed $0<\alpha<1$, the corresponding radius
$r=r(\alpha,d)>0$ and the prior $w_r$ were chosen independently of
$k$. Since
Eq.~\eqref{eq:Rav_le_delta} holds for every feasible retriever
$\cP$, taking the infimum over $\cP$ and using
Eq.~\eqref{eq:fixed_choi_risk_finite_k} gives the finite-$k$ bound
\begin{equation}\label{eq:kcopy_overhead_finite_k}
\nu_k(\Op{CPTP}_d)-1
\ge
\frac{(1-\alpha)(d^2-1)}
{2k+d(d^2+3)/r^2}.
\end{equation}
Equation~\eqref{eq:kcopy_overhead_finite_k} is the finite-$k$ bound
obtained here. Keeping $\alpha$ and $r$ fixed while $k\to\infty$
gives
\begin{equation*}
\liminf_{k\to\infty}
k\bigl(\nu_k(\Op{CPTP}_d)-1\bigr)
\ge\frac{(1-\alpha)(d^2-1)}{2}.
\end{equation*}
Letting $\alpha\downarrow0$ proves the proposition.
\end{proof}

\begin{proof}
Proposition~\ref{prop:kcopy_upper} bounds the limit superior in
\eqref{eq:kcopy_sharp_limit} by $(d^2-1)/2$, while
Proposition~\ref{prop:kcopy_lower} gives the matching limit inferior.
The limit therefore exists and equals $(d^2-1)/2$, which is equivalent
to \eqref{eq:kcopy_sharp_scaling}.
\end{proof}

\begin{corollary}[Asymptotic disappearance of the overhead]\label{cor:kcopy_compact_limit}
    For every nonempty channel set
    $\cS\subseteq\Op{CPTP}_d$ in fixed dimension,
    $\nu_\infty(\cS)=1$.
    In particular, the conclusion holds for every nonempty compact
    channel set.
\end{corollary}
\begin{proof}
    Any retriever feasible for $\Op{CPTP}_d$ is feasible
    for its subset $\cS$, so
    $\nu_k(\cS)\le \nu_k(\Op{CPTP}_d)$ for every $k$. Lemma~\ref{lem:kcopy_mono}
    gives the lower bound $\nu_\infty(\cS)\ge 1$, while
    Theorem~\ref{thm:kcopy_scaling} gives
    $\lim_{k\to\infty}\nu_k(\Op{CPTP}_d)=1$. Taking the limit in
    the sandwich
    $1\le\nu_k(\cS)\le\nu_k(\Op{CPTP}_d)$ proves the claim.
\end{proof}

The finite-$k$ analysis continues with a symmetry reduction. Suppose
that $\cS$ is
$(G,U,V)$-covariant in the sense of
Definition~\ref{def:G_covariance}, with self-conjugate
representations. The $k$-copy analog of the induced
representation~\eqref{eq:induced_rep} is
\begin{equation}\label{eq:kcopy_induced_rep}
\varrho_k(g)\;:=\;(U_g)_S\;\ox\;
\bigl[(U_g^*)_A\ox(V_g)_B\bigr]^{\ox k}\;\ox\;(V_g^*)_{S'}
\end{equation}
on $\cH_{\mathrm{tot}}$, and the associated commutant is
\begin{equation}\label{eq:kcopy_commutant_def}
\mathfrak{C}_G^{(k)}\;:=\;\bigl\{X\in\cB(\cH_{\mathrm{tot}})
\;:\;[\varrho_k(g),X]=0,\;\forall\,g\in G\bigr\}.
\end{equation}

For the all-channel family, $g=(U,V)$ ranges over
$\Op{SU}(d)\times\Op{SU}(d)$. The input matrix $U$ and output matrix
$V$ in this action therefore vary independently. Under the regrouping
below, their action factorizes as $U\ox(U^*)^{\ox k}$ on $\cH_U$ and
$V^{\ox k}\ox V^*$ on $\cH_V$.

Theorem~\ref{thm:G_symmetrisation} applies after replacing the program
representation by its $k$-fold tensor power. Indeed,
$R_P^\dagger J_\cE^T R_P$, where
$R_P:=(U_g^*)_A\ox(V_g)_B$, is the transposed Choi operator of the
same twisted channel as in the single-copy proof. Its $k$-fold power
therefore transforms the $k$ program copies simultaneously. Averaging the
two positive variables separately gives
\[
\bar J_\pm:=\int_G\varrho_k(g)J_\pm\varrho_k(g)^\dagger\,d\mu(g)
\in\mathfrak C_G^{(k)},
\]
without changing $p_\pm$ or feasibility. Combined with
Lemma~\ref{lem:kcopy_symmetric}, and using that the group and copy
permutation actions commute, this shows that an optimal
quasi-decomposition may be chosen in the joint $(G\times S_k)$-commutant.
For the explicit reduction below we retain the larger
$G$-commutant~\eqref{eq:kcopy_commutant_def}. The additional diagonal
$S_k$ fixed-point reduction is not included in the reported variable
or block counts.


The all-channel commutant has a mixed-tensor description. Let
$\cK_U\simeq\CC^d$ and $\cK_V\simeq\CC^d$ be carrier spaces for the
defining representations $U$ and $V$ of the independent input and
output copies of $\Op{SU}(d)$. After grouping input- and output-type
factors, the total Choi space becomes
$\cH_{\mathrm{tot}}\cong\cH_U\ox\cH_V$, where
\begin{equation}\label{eq:kcopy_rep_structure}
\cH_U=\cK_U\ox(\cK_U^*)^{\ox k},
\qquad
\cH_V=\cK_V^{\ox k}\ox\cK_V^*.
\end{equation}
The actions on these sectors are $U\ox(U^*)^{\ox k}$ and
$V^{\ox k}\ox V^*$, respectively.
Both sectors have dimension $D_{\mathrm{sec}}:=d^{k+1}$, and the total
dimension is $D=D_{\mathrm{sec}}^2=d^{2(k+1)}$.

Mixed Schur--Weyl duality describes the first sector through the
walled Brauer algebra
$B_{1,k}(d)$~\cite{brauer1937,stroomer1991,BCHLLS94}. Its abstract
diagram basis consists of pairings of two rows of $k+1$ vertices,
separated by a wall between the single $U$ position and the $k$ dual
positions. Vertical strands remain on one side of the wall, whereas
horizontal contraction strands cross it. After the standard bending
of the vertices on one side, these diagrams are in bijection with
$S_{k+1}$. Hence the abstract algebra has dimension $(k+1)!$ for every
value of $d$.

For a generic $d$-dimensional carrier $W$ and general $r,s$, let
\begin{equation}\label{eq:walled_brauer_rep}
\rho_{r,s}^{(d)}:
B_{r,s}(d)\longrightarrow
\cB\bigl(W^{\ox r}\ox(W^*)^{\ox s}\bigr)
\end{equation}
denote the natural mixed-tensor representation obtained by interpreting
each diagram as permutations and $W$--$W^*$ contractions. This
representation need not be faithful, so we distinguish the abstract
algebra from its represented image as follows.
\begin{equation}\label{eq:kcopy_brauer_image}
\mathcal A_{1,k}^{(d)}
:=\operatorname{im}\rho_{1,k}^{(d)}
\cong B_{1,k}(d)/\ker\rho_{1,k}^{(d)},
\end{equation}
and define $\mathcal A_{k,1}^{(d)}$ analogously. The representation is
faithful exactly when $d\geq k+1$~\cite{DipperDotyStoll2014}. Thus
$\dim\mathcal A_{1,k}^{(d)}=(k+1)!$ in this stable range, whereas
diagram relations can lower the image dimension for $d<k+1$.

\begin{theorem}[All-channel commutant]\label{thm:kcopy_cptp_comm}
For every $d\geq2$,
\begin{equation}\label{eq:kcopy_commutant}
\mathfrak{C}_{\Op{SU}(d)\times\Op{SU}(d)}^{(k)}
\cong
\mathcal A_{1,k}^{(d)}\ox\mathcal A_{k,1}^{(d)}.
\end{equation}
The $S_k$-invariant subalgebra is the fixed-point algebra of the
diagonal permutation action on the $k$ dual slots of $\cH_U$ and the
$k$ fundamental slots of $\cH_V$.
\end{theorem}
\begin{proof}
Mixed Schur--Weyl duality identifies the commutant on $\cH_U$ with
$\mathcal A_{1,k}^{(d)}$ and that on $\cH_V$ with
$\mathcal A_{k,1}^{(d)}$~\cite{BCHLLS94}. On a fixed mixed tensor
power, the central phase in $\Op{U}(d)$ acts as a scalar. The
$\Op{SU}(d)$, $\Op{U}(d)$, and complex general-linear actions therefore
have the same endomorphism commutant. Independence of the input and
output group factors gives the tensor product in
Eq.~\eqref{eq:kcopy_commutant}.
\end{proof}

For $k=1$, Eq.~\eqref{eq:kcopy_commutant} reproduces
Appendix~\ref{sec:cptp_appendix}. In particular, the represented sector
algebras are $\Op{span}\{I_{SA},\Omega_{SA}\}$ on $\cH_U$ and
$\Op{span}\{I_{BS'},\Omega_{BS'}\}$ on $\cH_V$. Hence the two-sector
commutant has dimension four.

The sector decompositions take the isotypic form
\begin{equation}\label{eq:kcopy_isotypic}
\cH_U
=\bigoplus_\lambda
\mathcal M_{1,\lambda}\ox\mathcal V_{1,\lambda},
\end{equation}
\begin{equation}\label{eq:kcopy_isotypic_V}
\cH_V
=\bigoplus_\mu
\mathcal M_{2,\mu}\ox\mathcal V_{2,\mu}.
\end{equation}
Here $\mathcal V_{i,\lambda}$ is an irreducible group representation,
and $\mathcal M_{i,\lambda}$ is its multiplicity space. Write their
dimensions as $v_{i,\lambda}$ and $m_{i,\lambda}$, respectively. For
$U\ox(U^*)^{\ox k}$, the labels are bipartitions
$(\alpha,\beta)$~\cite{BCHLLS94,CDDM05} for
which
\[
|\alpha|=1-t,
\qquad
|\beta|=k-t,
\qquad
t\in\{0,1\},
\qquad
\ell(\alpha)+\ell(\beta)\leq d.
\]
This condition, rather than the difference
$|\alpha|-|\beta|=1-k$ alone, specifies the components that occur.

Choose unitary intertwiners $\Gamma_1$ and $\Gamma_2$ implementing
Eqs.~\eqref{eq:kcopy_isotypic} and
\eqref{eq:kcopy_isotypic_V}. Schur's lemma gives
\begin{equation}\label{eq:kcopy_sector_blocks}
\Gamma_1\mathcal A_{1,k}^{(d)}\Gamma_1^\dagger
=\bigoplus_\lambda
\cB(\mathcal M_{1,\lambda})\ox I_{\mathcal V_{1,\lambda}},
\end{equation}
and the analogous identity holds for
$\Gamma_2\mathcal A_{k,1}^{(d)}\Gamma_2^\dagger$.
Consequently,
\begin{equation}\label{eq:kcopy_commutant_dimensions}
\ell_1:=\dim\mathcal A_{1,k}^{(d)}
=\sum_\lambda m_{1,\lambda}^2,
\qquad
\ell_2:=\dim\mathcal A_{k,1}^{(d)}
=\sum_\mu m_{2,\mu}^2.
\end{equation}
In the stable range $d\geq k+1$, both dimensions equal $(k+1)!$.

Choose real bases consisting of Hermitian operators,
$\{e^{(1)}_{j_1}\}_{j_1=1}^{\ell_1}$ and
$\{e^{(2)}_{j_2}\}_{j_2=1}^{\ell_2}$ for the Hermitian parts of the
two image algebras. They may be obtained from diagram images by taking
$D+D^\dagger$ and $i(D-D^\dagger)$ and discarding linear
dependencies. The reduced matrices $R_1(\lambda,j_1)$ and
$R_2(\mu,j_2)$ are uniquely defined by
\begin{equation}\label{eq:kcopy_red}
\begin{aligned}
\Gamma_1e^{(1)}_{j_1}\Gamma_1^\dagger
&=\bigoplus_\lambda
R_1(\lambda,j_1)\ox I_{\mathcal V_{1,\lambda}},\\
\Gamma_2e^{(2)}_{j_2}\Gamma_2^\dagger
&=\bigoplus_\mu
R_2(\mu,j_2)\ox I_{\mathcal V_{2,\mu}}.
\end{aligned}
\end{equation}
Because the intertwiners are unitary and the bases are Hermitian, every
$R_i$ is Hermitian and the block maps preserve adjoints and
eigenvalue signs.


The isotypic blocks now reduce the semidefinite constraints.
Let $Q$ be the permutation from the physical ordering to the grouped
ordering,
\[
Q:\ [S,(A_1,B_1),\ldots,(A_k,B_k),S']
\longrightarrow [S,A_1,\ldots,A_k,B_1,\ldots,B_k,S'].
\]
Because a $G$-invariant optimum exists, the variables admit the
expansion
\begin{equation}\label{eq:kcopy_expansion}
J_\pm = \sum_{j_1=1}^{\ell_1}\sum_{j_2=1}^{\ell_2}
  b_\pm(j_1,j_2)\;
  Q^\dagger\bigl(e^{(1)}_{j_1}\ox e^{(2)}_{j_2}\bigr)Q,
\end{equation}
where $b_\pm(j_1,j_2)\in\RR$ because the bases are Hermitian. Thus
$Q^\dagger(e^{(1)}_{j_1}\ox e^{(2)}_{j_2})Q$ is written in
physical coordinates. The number of real coefficients for each sign
is $n_{\rm var}:=\ell_1\ell_2$, instead of $D^2$. Equation
\eqref{eq:kcopy_expansion} uses the full $G$-commutant. Diagonal
$S_k$-averaging remains available, but no additional $S_k$ reduction
is included in the variable and block counts below.

The positivity, partial-trace, and programming constraints
of~\eqref{eq:kcopy_SDP_unreduced} now translate into reduced conditions
(C1$'$)--(C3$'$) on these coefficients.

Consider positivity first. By Schur's lemma for
$\Op{SU}(d)_U\times\Op{SU}(d)_V$, applying the unitary
$\Gamma_1\ox\Gamma_2$ after the permutation $Q$ and canonically regrouping the
multiplicity factors gives
$J_\pm\cong\bigoplus_{\lambda,\mu}J_{\pm;\lambda\mu}
\ox I_{v_{1,\lambda}\,v_{2,\mu}}$, where
\begin{equation}\label{eq:kcopy_Jblock}
J_{\pm;\lambda\mu}
= \sum_{j_1,j_2} b_\pm(j_1,j_2)\;
  R_1(\lambda,j_1)\ox R_2(\mu,j_2)
\end{equation}
is a
$(m_{1,\lambda}m_{2,\mu})\times
(m_{1,\lambda}m_{2,\mu})$ matrix.
Since all changes of frame are unitary, positivity of $J_\pm$ is
equivalent to positivity of this direct sum. The identity factor
$I_{v_{1,\lambda}\,v_{2,\mu}}$ does not affect the eigenvalue signs,
so the semidefinite constraint $J_\pm\ge 0$ reduces to
\begin{equation}\label{eq:kcopy_PSD}
\text{(C1$'$)}\qquad
\forall\,(\lambda,\mu):\quad
J_{\pm;\lambda\mu}\;\ge\;0.
\end{equation}
This replaces a single $D\times D$ semidefinite constraint with
$n_1 n_2$ independent constraints for each sign, of size at most
$(\max_\lambda m_{1,\lambda})(\max_\mu m_{2,\mu})$, where
$n_1$ and $n_2$ count the irreps in sectors~1 and~2.

For the partial-trace constraints, $\tr_{S'}$ acts only
on the $\cK_V^*$ factor of
$\cH_V$. In the grouped arrangement,
$\tr_{S'}[J_\pm]=\sum_{j_1,j_2}b_\pm(j_1,j_2)\,
e^{(1)}_{j_1}\ox T_2(j_2)$,
where
$T_2(j_2):=\tr_{\cK_V^*}[e^{(2)}_{j_2}]
\in\cB(\cK_V^{\ox k})$.
Since $T_2(j_2)$ commutes with $\Op{SU}(d)_V$ on
$\cK_V^{\ox k}$,
standard Schur--Weyl duality identifies it with the represented image
of $\CC[S_k]$ and
yields the decomposition
\begin{equation}
\cK_V^{\ox k}=\bigoplus_\nu\mathcal M'_\nu\ox\mathcal V'_\nu,
\end{equation}
with $\dim\mathcal M'_\nu=m'_\nu$. Let
$R'_2(\nu,j_2)\in M_{m'_\nu}$ denote the action of $T_2(j_2)$ on
this multiplicity block. Projecting
$\tr_{S'}[J_\pm]=p_\pm I$ onto each $(\lambda,\nu)$ block gives
\begin{equation}\label{eq:kcopy_TP}
\text{(C2$'$)}\qquad
\forall\,(\lambda,\nu):\quad
\sum_{j_1,j_2} b_\pm(j_1,j_2)\;
R_1(\lambda,j_1)\ox R'_2(\nu,j_2)
= p_\pm\, I_{m_{1,\lambda}m'_\nu}.
\end{equation}

The exact-programming constraint requires more care. Define the
programming coefficient
\begin{equation}\label{eq:kcopy_Pi}
\Pi(j_1,j_2,\cE) :=
\tr_{P^{\ox k}}\!\bigl[
  Q^\dagger(e^{(1)}_{j_1}\ox e^{(2)}_{j_2})Q\cdot
  (I_S\ox (J_\cE^T)^{\ox k}\ox I_{S'})
\bigr]
\;\in\; M_{d^2},
\end{equation}
where $\tr_{P^{\ox k}}$ traces over the $2k$ program subsystems.
Then~\eqref{eq:kconstraint} becomes
\begin{equation}\label{eq:kcopy_prog}
\text{(C3$'$)}\qquad
\forall\,\cE\in\cS:\quad
\sum_{j_1,j_2}
\bigl(b_+(j_1,j_2)-b_-(j_1,j_2)\bigr)\,
\Pi(j_1,j_2,\cE)
\;=\; d^k\,J_\cE.
\end{equation}
Together with (C2$'$), condition (C3$'$) automatically implies
$p_+-p_-=1$, by the trace argument following
Eq.~\eqref{eq:kcopy_SDP_unreduced}. No additional scalar normalization
constraint is therefore required in the reduced program.
Each $\Pi(j_1,j_2,\cE)$ is a $d^2\times d^2$ matrix whose size is
independent of $D$. The universal constraint is finite-dimensional.
Indeed, define
\[
\mathcal V_{d,k}
:=\operatorname{span}\left\{
(J_\cE^T)^{\ox k}:\cE\in\Op{CPTP}_d
\right\}.
\]
For $X=(J_\cE^T)^{\ox k}$, trace preservation gives
$d[\tr_{2,\ldots,k}X]^T=d^kJ_\cE$, where
$\tr_{2,\ldots,k}$ removes the complete Choi factors numbered
$2,\ldots,k$ and is understood as the identity when $k=1$. The
right-hand side of
Eq.~\eqref{eq:kcopy_prog} is therefore a linear function of the same
tensor power that determines its left-hand side. It is sufficient to
enforce the equality on channels whose Choi powers form a basis of
$\mathcal V_{d,k}$. This observation establishes the existence of a
finite exact formulation. The numerical implementation below instead
imposes the programming equations on finite channel ensembles and uses
the resulting values only to explore finite-copy behavior; no numerical
result is used in the analytic optimality proofs.

The sector-level structure of $\Pi$ can be exposed by writing
$J_\cE^T=\sum_r\sigma_r A_r\ox B_r$ in an operator Schmidt
decomposition across $A/B$, which aligns with the
$\cH_U$/$\cH_V$ grouping. This gives
\begin{equation}\label{eq:kcopy_Pfactor}
\Pi(j_1,j_2,\cE)
= \sum_{r_1,\ldots,r_k}\prod_{i=1}^k\sigma_{r_i}\;
  F_1(j_1;\vec r)\ox F_2(j_2;\vec r),
\end{equation}
with $d\times d$ sector traces
$F_1(j_1;\vec r)=\tr_{(\cK_U^*)^{\ox k}}\bigl[e^{(1)}_{j_1}
(I_{\cK_U}\ox A_{r_1}\ox\cdots\ox A_{r_k})\bigr]$ and
$F_2(j_2;\vec r)=\tr_{\cK_V^{\ox k}}\bigl[e^{(2)}_{j_2}
(B_{r_1}\ox\cdots\ox B_{r_k}\ox I_{\cK_V^*})\bigr]$.
These contractions admit a diagrammatic evaluation from the walled
Brauer basis, providing a route that avoids forming
$d^{k+1}\times d^{k+1}$ matrices.

Collecting (C1$'$)--(C3$'$), the block diagonalization yields the
following reduced program.

\begin{proposition}[Reduced SDP for $\nu_k(\Op{CPTP}_d)$]%
\label{prop:kcopy_reduced_SDP}
The $k$-copy programming overhead for all quantum channels,
$\cS=\Op{CPTP}_d$, is
\begin{equation}\label{eq:kcopy_SDP}
\nu_k(\Op{CPTP}_d) = \min_{\substack{
  b_\pm(j_1,j_2)\in\RR\\
  p_\pm\ge 0}}\; p_+ + p_-
\end{equation}
subject to the semidefinite
constraints~\eqref{eq:kcopy_PSD}, the linear
equalities~\eqref{eq:kcopy_TP}, and the programming
conditions~\eqref{eq:kcopy_prog} imposed on any Choi-power basis of
$\mathcal V_{d,k}$. With such a basis, these are finitely many
semidefinite and linear constraints.
\end{proposition}

\begin{proof}
Starting from a feasible point of
Eq.~\eqref{eq:kcopy_SDP_unreduced}, the sign-wise group average above
places both $J_+$ and $J_-$ in the all-channel commutant without
changing the objective. Expansion in the Hermitian bases then gives
Eq.~\eqref{eq:kcopy_expansion}. The unitary isotypic decomposition
makes positivity equivalent to (C1$'$), while taking the partial trace
and evaluating the programming contraction give (C2$'$) and (C3$'$),
respectively. Enforcing (C3$'$) on a basis of $\mathcal V_{d,k}$ is
equivalent to enforcing it for every channel by the linearity argument
preceding the proposition.
Conversely, coefficients satisfying (C1$'$)--(C3$'$) reconstruct
$J_\pm$ through Eq.~\eqref{eq:kcopy_expansion}. Condition (C1$'$)
gives $J_\pm\geq0$, condition (C2$'$) gives the scaled-channel
marginals, and the chosen Choi-power basis extends (C3$'$) to every
channel. The reconstructed pair is therefore feasible in
Eq.~\eqref{eq:kcopy_SDP_unreduced} with the same value $p_++p_-$,
which proves equality of the two programs.
\end{proof}

The coefficient count per sign,
$n_{\rm var}=\ell_1\ell_2$, and the maximum block size
$\max_{\lambda,\mu}(m_{1,\lambda}m_{2,\mu})$ depend on the
isotypic decompositions of the two mixed-tensor sectors.
Here $\ell_i=\sum_\lambda m_{i,\lambda}^2$, with
$\ell_i=(k+1)!$ when $d\geq k+1$. For smaller $d$, relations in the
mixed-tensor representation can reduce $\ell_i$.

As a concrete illustration, consider $d=5$, $k=5$.
The unreduced SDP has $D=5^{12}=244140625$. In each sector, the natural
representation of the 720-dimensional abstract algebra has a
one-dimensional kernel, so $\ell_1=\ell_2=719$. Each sector has
11 isotypic components and maximum multiplicity $15$. Consequently,
$n_{\rm var}=719^2=516961$ coefficients and 121 positivity blocks of
size at most $225\times225$ occur for each of the positive and negative
variables. Together, they contain 242 positivity blocks in total. These counts concern the
$G$-commutant reduction. They do not include a further diagonal
$S_5$ fixed-point reduction and follow directly from the sector
multiplicities.

\subsection{Finite-copy numerical implementation}

Table~\ref{tab:kcopy_cptp_numerics} reports two complementary finite-copy
calculations. The $k=1$ entries are analytic exact values. The remaining
SDP entries are numerical outputs of the implemented walled-Brauer-reduced
programs with finite channel ensembles; they demonstrate how the reduced
optimization can be evaluated and indicate its finite-copy behavior. They
are presented as numerical results of the implementation, rather than as
independent proofs of exact finite-$k$ optimality, and play no role in
Theorems~\ref{thm:cptp} and~\ref{thm:kcopy_main}. For $k\geq2$, the PBT
entries are independently derived rigorous achievable upper bounds; at
$k=1$, the PBT-inversion construction is unavailable.

\begin{figure}[t]
    \centering
    \includegraphics[width=0.5\linewidth]{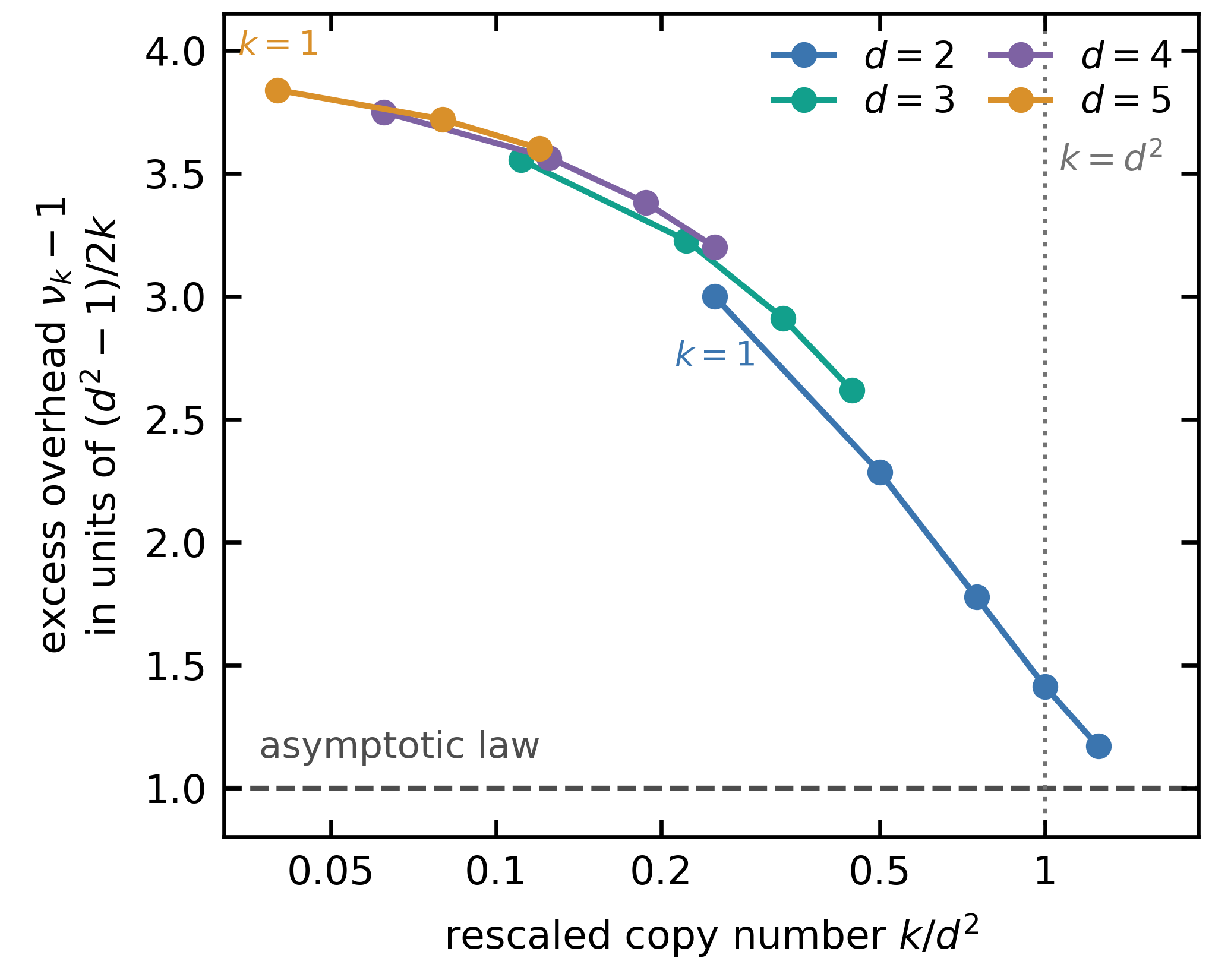}
    \caption{Numerical exploration of finite-copy behavior using the
    walled-Brauer reduction. The vertical coordinate is
    $R_{k,d}=2k(v_{k,d}-1)/(d^2-1)$, where $v_{k,d}$ denotes the
    displayed analytic or numerical SDP value, and the horizontal
    coordinate is $k/d^2$ on a logarithmic scale. The dashed line at
    $R_{k,d}=1$ shows the fixed-dimension asymptotic law of
    Theorem~\ref{thm:kcopy_main}; it is not a finite-$k$ fit. The dotted
    line marks $k=d^2$, and the solid segments only guide the eye. The
    $k=1$ markers are analytic exact values, while all $k\geq2$ markers are
    outputs of the finite-ensemble reduced-SDP implementation. The
    limited data do not establish a dimension-independent collapse or a
    finite-copy crossover. The figure is illustrative and is not used in
    either universal optimality proof. Table~\ref{tab:kcopy_cptp_numerics}
    lists the numerical values together with the rigorous PBT-achievable
    upper bounds available for $k\geq2$.}
    \label{fig:kcopy_decay}
\end{figure}

The archived computations used MATLAB R2024b, CVX 2.2, MOSEK 9.1.9,
and QETLAB 0.9. The standard sampled runs used seed zero and 500 channel
samples. The entries at $(d,k)=(4,4)$ and $(5,3)$ used 128 and 256
samples, respectively. Default solver tolerances were retained, while
the numerical row-rank threshold was
$\max(\operatorname{size}A)\epsilon_{\rm mach}
\|\operatorname{diag}R\|_\infty$ for each constraint matrix $A$ with QR
factor $R$. Increasing the sample count from 500
to 800 at $(d,k)=(3,2)$ with an independent seed changed the numerical
value
by $3\times10^{-9}$, which indicates numerical stability under this
sample increase. At $(d,k)=(3,4)$, the primary YALMIP--MOSEK value is
$3.619643423467$, while a CVX cross-check gives $3.619723965260$.
The table reports four decimals from the primary run. Fresh-channel
programming residuals for the least resolved entries are of order
$10^{-7}$. These diagnostics support the displayed numerical precision.

The upper entries follow from Proposition~\ref{prop:kcopy_upper} without
asymptotic expansion. The identity
$\|\cD_t\|_\diamond=1+2(1-d^{-2})(t-1)$ of
Eq.~\eqref{eq:inverse_depol_norm} is exact, so it suffices to evaluate
$\eta_{k,d}=(d^2F^{\rm std}_d(k)-1)/(d^2-1)$ at finite $k$. In the
Schur--Weyl analysis of the standard protocol, this entanglement fidelity
is a finite sum over Young
diagrams~\cite{ishizaka2008,Christandl2021asymptotic},
\begin{equation}\label{eq:pbt_exact_fidelity}
F^{\rm std}_d(k)
=\frac{1}{d^{k+2}}
\sum_{\alpha\vdash k-1}
\Biggl(\;\sum_{\mu=\alpha+\square}\sqrt{m_\mu\,d_\mu}\;\Biggr)^{2},
\end{equation}
where both partitions are restricted to at most $d$ rows, $\mu$ runs over
the diagrams obtained from $\alpha$ by adding one box, $d_\mu$ is the
dimension of the $S_k$ irreducible representation labeled by $\mu$, and
$m_\mu$ is the dimension of the corresponding $\Op{U}(d)$ irreducible
representation. Two checks fix this evaluation. At $k=1$ it returns
$F^{\rm std}_d(1)=1/d^2$ for every $d$, the entanglement fidelity of the
completely depolarizing channel. Hence $\eta_{1,d}=0$ and
$\cD_{\eta_{1,d}}^{-1}$ does not exist, so the PBT-inversion construction
provides no finite bound at $k=1$. This is a limitation of that
construction, not a statement that the one-copy optimum is infinite:
Theorem~\ref{thm:cptp} instead gives
$\nu_1(\Op{CPTP}_d)=2d^2-3+2/d^2$. Thus the one-copy optimum is not a
limiting case of the many-copy construction. Finally,
$4k\,[1-F^{\rm std}_d(k)]$ tends to $d^2-1$ in agreement with
Eq.~\eqref{eq:std_pbt_fidelity}, reaching $2.9972$ at $d=2$, $k=800$ and
$15.016$ at $d=4$, $k=100$. From $k=2$ onward the resulting values are
rigorous achievable upper bounds.

\begin{table}[htbp]
\centering
\caption{\label{tab:kcopy_cptp_numerics}
Finite-copy numerical SDP results and PBT-achievable upper bounds. The
symbols A and N denote an analytic exact value and a numerical output of
the finite-ensemble reduced-SDP implementation, respectively. The
numerical SDP values illustrate finite-copy behavior and are not used to
establish either universal theorem. PBT-UB gives
$\|\cD_{\eta_{k,d}}^{-1}\|_\diamond$ evaluated from the exact finite-$k$
standard-PBT fidelity and, for $k\geq2$, is a rigorous achievable upper
bound. At $k=1$, ``N/A'' records $\eta_{1,d}=0$ and the consequent
failure of the PBT-inversion construction; it does not indicate an
infinite one-copy optimum. Values are rounded. The symbol ``--'' denotes
an unreported entry.}
\smallskip
\setlength{\tabcolsep}{4pt}
\begin{tabular}{cl|ccccc}
\toprule
 & & $k=1$ & $k=2$ & $k=3$ & $k=4$ & $k=5$ \\
\midrule
\multirow{2}{*}{$d=2$}
  & SDP & $5.5000^{\rm A}$ & $2.7133^{\rm N}$ & $1.8883^{\rm N}$ & $1.5294^{\rm N}$ & $1.3509^{\rm N}$ \\
  & PBT-UB & N/A & 4.6962  & 2.5000  & 1.8300 & 1.5312 \\
\midrule
\multirow{2}{*}{$d=3$}
  & SDP & $15.2222^{\rm A}$ & $7.4564^{\rm N}$ & $4.8822^{\rm N}$ & $3.6196^{\rm N}$ & -- \\
  & PBT-UB & N/A & 14.3072 & 7.0115  & 4.6187 & -- \\
\midrule
\multirow{2}{*}{$d=4$}
  & SDP & $29.1250^{\rm A}$ & $14.3662^{\rm N}$ & $9.4538^{\rm N}$ & $7.0039^{\rm N}$ & -- \\
  & PBT-UB & N/A & 28.1724 & 13.8952 & 9.1445 & -- \\
\midrule
\multirow{2}{*}{$d=5$}
  & SDP & $47.0800^{\rm A}$ & $23.3244^{\rm N}$ & $15.4104^{\rm N}$ & -- & -- \\
  & PBT-UB & N/A & 46.1102 & 22.8425 & -- & -- \\
\bottomrule
\end{tabular}
\end{table}

\end{document}